\documentclass[aps,%
twocolumn,%
floatfix,%
prl,%
amsfonts,%
groupedaddress,%
superscriptaddress,%
nobibnotes,%
longbibliography]{revtex4-2}%

\usepackage{amsmath}
\usepackage{amssymb}
\usepackage{mathtools}
\usepackage{amsthm}
\usepackage{bm,bbm} 
\usepackage[table,dvipsnames,usenames]{xcolor}
\usepackage{graphicx}
\usepackage{indentfirst}
\usepackage{enumitem}
\usepackage{algorithm2e}
\usepackage{algpseudocode}
\definecolor{beamer@blendedblue}{rgb}{0.2,0.2,0.7}      
\usepackage[colorlinks=true,%
linkcolor=beamer@blendedblue,%
bookmarks=true,%
breaklinks=true,%
filecolor=beamer@blendedblue,%
anchorcolor=yellow,%
citecolor=beamer@blendedblue,%
urlcolor=beamer@blendedblue,%
pdfauthor={Si-Yuan Chen, Congcong Zheng, Kun Wang, Ming-Cheng Chen},%
pdfsubject={Pauli-resolved virtual distillation},%
CJKbookmarks=true]{hyperref}
\usepackage[titletoc,title]{appendix}
\usepackage{tikz}
\usetikzlibrary{graphs}
\usepackage{quantikz}
\usepackage{booktabs} 
\usepackage{float}
\usepackage{pifont}
\usepackage[skins,breakable,most]{tcolorbox}

\newtheorem{theorem}{Theorem}

\newtheorem{lemma}{Lemma}
\newtheorem{corollary}{Corollary}

\newtheorem{definition}{Definition}

\usepackage{pifont}
\newlist{todolist}{itemize}{2}
\setlist[todolist]{label=$\square$}
\DeclareMathOperator{\tr}{tr}

\renewcommand{\bra}[1]{\langle #1\rvert}
\renewcommand{\ket}[1]{\lvert #1\rangle}
\renewcommand{\proj}[1]{\lvert #1\rangle\!\langle #1\rvert}
\newcommand{\ox}{\otimes}
\newcommand{\1}{\mathbbm{1}}

\renewcommand{\braket}[2]{\langle #1\vert #2\rangle}
\newcommand*{\cF}{\mathcal{F}}
\newcommand*{\cG}{\mathcal{G}}
\newcommand*{\cH}{\mathcal{H}}

\newcommand*{\cP}{\mathcal{P}}
\newcommand*{\cS}{\mathcal{S}}

\newcommand*{\cA}{\mathcal{A}}
\newcommand*{\cD}{\mathcal{D}}
\newcommand*{\cM}{\mathcal{M}}
\newcommand*{\cU}{\mathcal{U}}
\newcommand*{\cJ}{\mathcal{J}}
\newcommand*{\cL}{\mathcal{L}}

\newcommand*{\cB}{\mathcal{B}}

\definecolor{alizarin}{rgb}{0.82, 0.1, 0.26}
\definecolor{googleblue}{HTML}{4285F4}
\definecolor{googlered}{HTML}{DB4437}
\definecolor{googleyellow}{HTML}{F4B400}
\definecolor{googlegreen}{HTML}{0F9D58}
\definecolor{klevinblue}{HTML}{002FA7}
\definecolor{tiffanyblue}{HTML}{0ABAB5}

\newcommand\prlsection[1]{\textit{\textbf{#1}}---}

\hypersetup{hypertexnames=false}
\hypersetup{hypertexnames=false,pdftitle={Pauli-resolved virtual distillation},pdfsubject={Pauli-resolved virtual distillation}}
\begin{document}

\newcommand{\thetitle}{Pauli-resolved virtual distillation}
\title{\thetitle}

\author{Si-Yuan Chen}\thanks{Contact author: chance.siyuan@gmail.com}
\affiliation{Hefei National Research Center for Physical Sciences at the Microscale and School of Physical Sciences, University of Science and Technology of China, Hefei 230026, China}%

\affiliation{Shanghai Research Center for Quantum Science and CAS Center for Excellence in Quantum Information and Quantum Physics, University of Science and Technology of China, Shanghai 201315, China}%

\affiliation{Hefei National Laboratory, University of Science and Technology of China, Hefei 230088, China}%

\author{Congcong Zheng}
\affiliation{State Key Lab of Millimeter Waves, Southeast University, Nanjing 211189, China}

\author{Kun Wang}\thanks{Contact author: nju.wangkun@gmail.com}
\affiliation{College of Computer Science and Technology,
National University of Defense Technology,
Changsha 410073, China}%

\author{Ming-Cheng Chen}\thanks{Contact author: cmc@ustc.edu.cn}
\affiliation{Hefei National Research Center for Physical Sciences at the Microscale and School of Physical Sciences, University of Science and Technology of China, Hefei 230026, China}%

\affiliation{Shanghai Research Center for Quantum Science and CAS Center for Excellence in Quantum Information and Quantum Physics, University of Science and Technology of China, Shanghai 201315, China}%

\affiliation{Hefei National Laboratory, University of Science and Technology of China, Hefei 230088, China}%

\begin{abstract}
Learning the full Pauli profile of the virtually distilled quantum state $\rho^m/\tr(\rho^m)$ has so far required exponentially many copies of $\rho$. We show that all $4^n$ squared Pauli moments $[\tr(P\rho^m)]^2$ can be learned to additive error $\varepsilon$ with confidence $1-\delta$ from one $2m$-replica measurement setting using $O(m[n+\log(1/\delta)]/\varepsilon^2)$ copies. This is an exponential speedup in the system size $n$ over previous protocols. For $m = 2$, we propose the coherent Bell difference sampling circuit that realizes this measurement on current devices. The speedup originates from paired replicas that cancel the anticommutation signs of Pauli operators, collapsing the incompatibility of the Pauli family. We certify this collapse by introducing the quantum Bernstein norm, a computable incompatibility measure for nonlinear functionals. A further information-theoretic $m$-replica protocol recovers the signed moments $\tr(P\rho^m)$ using $O(m[n+\log(1/\delta)]/\varepsilon^4)$ copies. For fixed $m$ in the stated lower-bound regime, it attains the optimal replica number, since any protocol with fewer replicas requires exponentially many copies.
\end{abstract}
\date{\today}
\maketitle

\prlsection{Introduction.} Predicting many properties of an unknown quantum state is a central task in quantum information science, with applications ranging from device benchmarking~\cite{silva2025hands,cai2023quantum} to identifying quantum algorithms with exponential speedups~\cite{aharonov2022quantum,huang2022quantum,huang2022provably}. For linear properties, shadow tomography predicts exponentially many observables by reusing a single measurement record~\cite{aaronson2018shadow}. The sample cost of such protocols, however, varies sharply with the choice of the target family~\cite{huang2020predicting,chen2024optimal}, and this variation has a physical origin: measurement incompatibility~\cite{heinosaari2015noise,guhne2023colloquium}.

A single qubit already illustrates how incompatibility inflates the sample cost. The observables $X$, $Y$, and $Z$ cannot be measured jointly~\cite{Busch1986}, so any unbiased estimator that extracts all three from one repeated measurement must take values outside $[-1, 1]$, and by Hoeffding's inequality~\cite{hoeffding1963probability} the sample cost grows as the square of this estimator range. For the $n$-qubit Pauli family, this cost grows exponentially under single-copy measurements~\cite{huang2020predicting,chen2024optimal}. Interestingly, collective measurements on two replicas remove the obstruction. When $PQ = -QP$, the anticommutation sign appears once in each replica of the pair and cancels, so the doubled operators $P \otimes P$ commute with each other, and Bell sampling returns all $4^n$ amplitudes $|\tr(P\rho)|$ from a single measurement setting~\cite{huang2022quantum,chen2024optimal,chen2022exponential}. We refer to this sign cancellation as the \textit{paired cancellation} mechanism. The signs of $\tr(P\rho)$, however, remain out of reach for this measurement.

More recently, the shadow tomography paradigm has been extended to nonlinear functionals. A prominent example is virtual distillation~\cite{huggins2021virtual,koczor2021exponential}, which suppresses incoherent errors by reporting expectation values on the filtered state $\rho^m/\tr(\rho^m)$~\cite{o2023purification,cotler2019quantum,seif2023shadow}. Among these targets, the Pauli moments $\{\tr(P\rho^m)\}_{P \in \mathcal{P}_n}$ are of particular interest, since they determine the error-mitigated energies of many-term Hamiltonians and certify the magic of the distilled state~\cite{leone2022stabilizer,haug2023scalable}. Nonlinear shadow tomography for Pauli moments has been developed in two directions. Fixing one observable $O$, Ref.~\cite{chen2025simultaneous} estimates the moment ladder $\{\tr(O\rho^m)\}_{m=1}^{k}$. Fixing the moment order, replica-shadow protocols~\cite{liu2026auxiliary,du2026optimal} build randomized shadows of $\rho^m$ that are efficient for low-rank or low-weight observables. For the complete family $\{\tr(P\rho^m)\}_{P \in \mathcal{P}_n}$ over all Pauli operators at a fixed $m$, both routes require exponentially many copies of $\rho$.

\begin{figure*}[htbp!]
    \centering
    \includegraphics[width=\textwidth]{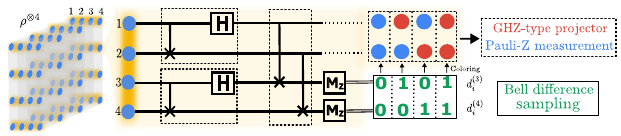}
    \caption{The coherent Bell difference sampling (CBDS) circuit. Four replicas of $\rho$ enter two Bell measurement blocks, acting on replicas $(1, 2)$ and $(3, 4)$. Transversal CNOT layers coherently subtract the Bell sampling outcomes of replicas $(1, 2)$ from replicas $(3, 4)$ before any readout. Measuring replicas $(3, 4)$ yields Bell difference outcomes $d_i$ that select a Pauli sector; conditioned on $d_i$, each site of replicas $(1, 2)$ is read either in the computational basis (blue) or within a collective GHZ-type projector (red), returning the eigenvalue of $O_P$ in that sector (SM~\cite{Supplemental}, Appendix A). A single run updates the estimators of all $4^n$ squared moments $[\tr(P\rho^2)]^2$ simultaneously.}
    \label{fig:4replica_fig2}
\end{figure*}

In this Letter, we show that the exponential barrier is again an incompatibility problem, and extend the paired cancellation to nonlinear moments. When $m$ replicas of $\rho$ are paired against another $m$ replicas, the anticommutation signs again appear twice and cancel, and the multi-replica observables representing the squared moments $\{[\tr(P\rho^m)]^2\}_{P \in \mathcal{P}_n}$ commute with each other. The entire squared-moment family thus becomes jointly measurable in one $2m$-replica setting, with every estimator bounded in $[-1, 1]$, at a cost of $O(m[n+\log(1/\delta)]/\varepsilon^2)$ copies of $\rho$. For $m = 2$, we propose a circuit that realizes this measurement and is of independent interests. This circuit generalizes the existing Bell difference sampling strategy~\cite{gross2021schur,grewal2023efficient,grewal2024improved}, which we call the coherent Bell difference sampling (CBDS) circuit. 

Two questions remain about this construction: how such unit-range estimators can be found systematically for other families, and whether a smaller estimator range is possible. To answer both, we introduce the quantum Bernstein norm, an incompatibility measure for families of nonlinear functionals. The norm is computable by a semidefinite program, contains the robustness of measurement incompatibility~\cite{heinosaari2015noise,guhne2023colloquium,carmeli2016quantum} as its linear special case, and certifies that our protocol reaches the smallest possible unbiased estimator range.

Finally, we recover the signed moments $\tr(P\rho^m)$. The joint measurement above reads only squared moments, so the sign information is supplied by interfering $\rho^{\otimes m}$ with known reference states, in the spirit of the mimicking-state technique~\cite{chen2024optimal,king2025triply}. The resulting $m$-replica protocol uses $O(m[n+\log(1/\delta)]/\varepsilon^4)$ copies (Theorem~\ref{main-pauli-vd}). For fixed $m$, its replica number is optimal in the parameter regime stated in SM~\cite{Supplemental}, Appendix A.

\prlsection{Squared Pauli moments.} We first make the task precise, before constructing the measurement that accomplishes it. We seek a single generalized measurement, repeated on fresh replicas of $\rho$, such that for every $P \in \mathcal{P}_n$ some classical post-processing of the outcomes gives an unbiased estimate of $[\tr(P\rho^2)]^2$. Our construction strategy is to find, for each $P$, a multi-replica observable that reproduces the target as its expectation value, such that all these observables mutually commute. Measuring in their common eigenbasis then yields unbiased estimates of all $4^n$ targets from a single measurement setting. We construct such a commuting family explicitly on four replicas. For each $P \in \mathcal{P}_n$, define
\begin{align}\label{eq:main_OP}
O_P := \frac{P_1 P_2 I_3 I_4 + I_1 I_2 P_3 P_4}{2}\, W_{13} W_{24},
\end{align}
where $P_j$ denotes the Pauli operator $P$ acting on replica $j$, and $W_{13}$, $W_{24}$ swap replicas $(1, 3)$ and $(2, 4)$. Two facts follow directly (SM~\cite{Supplemental}, Appendix A). First, $\tr(O_P \rho^{\otimes 4}) = [\tr(P\rho^2)]^2$, because the swaps stitch the four replicas into two factors of $\tr(P\rho^2)$. Second, all $4^n$ operators $O_P$ mutually commute. When $PQ = -QP$, the anticommutation sign is generated once in the replica pair $(1, 2)$ and once in the pair $(3, 4)$, and the two signs cancel. This is precisely the paired cancellation mechanism identified in the introduction, now operating at the level of nonlinear moments.

Each $O_P$ has spectrum in $[-1, 1]$, so measuring $\rho^{\otimes 4}$ in the common eigenbasis yields an unbiased estimate of every squared moment with estimator range one. Hoeffding's inequality~\cite{hoeffding1963probability} and a union bound over the $4^n$ targets give
\begin{align}\label{eq:main_N_squared}
N = O\left(\frac{n + \log(1/\delta)}{\varepsilon^2}\right)
\end{align}
measurement rounds for additive error $\varepsilon$ on all targets $[\tr(P\rho^2)]^2$ with confidence $1-\delta$. As a concrete example, for $n = 50$, $\varepsilon = 0.05$, and $\delta = 0.01$, about $6 \times 10^4$ rounds suffice to learn all $4^{50}$ squared moments, corresponding to $2.4 \times 10^5$ copies of $\rho$ processed on 200 qubits.

The measurement in the common eigenbasis is realized by the CBDS circuit of Fig.~\ref{fig:4replica_fig2}. Bell sampling uses two copies and reads the amplitudes $|\tr(P\rho)|$~\cite{montanaro2017learning}. Bell difference sampling uses four copies and subtracts two Bell outcomes classically~\cite{gross2021schur,grewal2023efficient,grewal2024improved}. CBDS also uses four copies, but the Bell outcomes of replicas $(1, 2)$ and $(3, 4)$ are subtracted coherently before any readout. This coherent step implements the pairing $W_{13}W_{24}$ of Eq.\eqref{eq:main_OP} at the level of measurement outcomes, and thereby raises the accessible quantity from $\tr(P\rho)$ to $\tr(P\rho^2)$. No classical post-processing of independent Bell sampling runs is known to achieve the same, indicating that the coherent subtraction is essential.

The pairing mechanism generalizes to arbitrary moment orders. Pairing $m$ replicas of $\rho$ against another $m$ replicas yields operators $O_P^{(m)}$ that commute by the same paired cancellation mechanism, have spectra in $[-1, 1]$, and satisfy $\tr(O_P^{(m)} \rho^{\otimes 2m}) = [\tr(P\rho^m)]^2$. Hence for every $m$ the squared-moment family
\begin{align}\label{eq:main_Fm2}
\mathcal{F}_m^{(2)} := \{[\tr(P\rho^m)]^2\}_{P \in \mathcal{P}_n}
\end{align}
is jointly measurable in a single $2m$-replica setting, at a cost of $O(m[n+\log(1/\delta)]/\varepsilon^2)$ copies.
This result, however, still leaves its systematic construction open: 

\textit{Given an arbitrary family of unit-range nonlinear functionals, how can we systematically construct jointly unbiased estimators with unit range?}

The commuting family construction above provides a sufficient but not a necessary condition for unit-range joint estimation (SM~\cite{Supplemental}, Appendix A). While the standard measures of measurement incompatibility is formulated for collections of POVMs~\cite{heinosaari2015noise,guhne2023colloquium,carmeli2016quantum}, it does not apply to nonlinear polynomials of a state.
This motivates us to introduce a new quantitative measure to quantify
the cost of joint estimation for nonlinear functionals.

\prlsection{The quantum Bernstein norm.} Consider a finite family of state functionals $\mathcal{F} = \{f_\alpha(\rho)\}_{\alpha \in \mathcal{A}}$. As shown in Fig.~\ref{fig:main}, a protocol repeats one collective POVM $\{\Pi_j\}_{j \in \mathcal{J}}$ on $\rho^{\otimes k}$ and builds an unbiased estimator for each target from the same outcomes, so that $f_\alpha(\rho) = \sum_{j} h_\alpha(j) \tr(\Pi_j \rho^{\otimes k})$ for all states $\rho$ and all $\alpha \in \mathcal{A}$. The $k$-replica quantum Bernstein norm $\|\mathcal{F}\|_k$ is the smallest achievable estimator range,
\begin{align}\label{eq:mainthe2}
\|\mathcal{F}\|_k := \min_{\{\Pi_j\}, \{h_\alpha\}} \max_{\alpha, j} |h_\alpha(j)|.
\end{align}
By Hoeffding's inequality and a union bound~\cite{hoeffding1963probability}, $N = 2\|\mathcal{F}\|_k^2 \varepsilon^{-2} \log(2|\mathcal{A}|/\delta)$ rounds suffice to estimate all targets to error $\varepsilon$ with confidence $1-\delta$, so the norm controls the sample cost of joint unbiased estimation. In the classical limit, where all operators are diagonal in a fixed basis, the POVM elements become Bernstein polynomials and the norm reduces to its classical counterpart~\cite{wu2020polynomial,lorentz2012bernstein}, which motivates the terminology (SM~\cite{Supplemental}, Appendix B).

\begin{figure}[htbp!]
    \centering
    \includegraphics[width=\columnwidth]{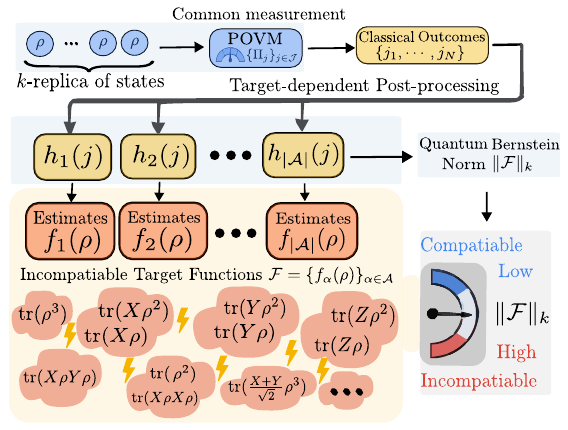}
    \caption{Operational meaning of the $k$-replica quantum Bernstein norm $\|\mathcal{F}\|_k$. A common $k$-copy POVM on $\rho^{\otimes k}$ produces a classical outcome $j$. The post-processing function $h_\alpha(j)$ forms an unbiased estimator for each target functional $f_\alpha$. The $\|\mathcal{F}\|_k$ is the minimum possible range of these post-processing functions. Thus, a small $\|\mathcal{F}\|_k$ means that the nonlinear properties in $\mathcal{F}$ can be jointly estimated without bias and with small overhead.}
    \label{fig:main}
\end{figure}

Once each degree-$k$ polynomial functional is written as a symmetric $k$-replica observable $F_\alpha$ with $f_\alpha(\rho) = \tr(F_\alpha \rho^{\otimes k})$, the norm becomes computable by a semidefinite program.

\begin{theorem}[Computability]\label{thm:main_dual}
    $\|\mathcal F\|_k$ equals the optimal value of an explicit semidefinite
    program with strong duality (SM Appendix~B~\cite{Supplemental}). In compact form:
    \begin{align}
      &\text{Primal:}\quad  \|\mathcal F\|_k
        = \min_{\{h_\alpha(j)\},\{\Pi_j\succeq 0\}}\;\lambda,
        \nonumber\\[-0.25em]
        &\text{s.t.}\,
        \sum_j h_\alpha(j)\Pi_j=F_\alpha, \;\,
        |h_\alpha(j)|\le\lambda, \;\,
        \sum_j \Pi_j=\mathbb{I}^{\otimes k},\label{eq:mainprimal}\\[0.45em]
        &\text{Dual:}\quad  \|\mathcal F\|_k
        = \max_{Z,\{Y_\alpha\}}\;
        \sum_{\alpha}\tr(Y_\alpha F_\alpha),
        \nonumber\\[-0.35em]
        &\text{s.t.}\quad\tr Z=1, \quad Z\succeq\sum_{\alpha}c_\alpha Y_\alpha,
        \quad\forall c_\alpha\in\{+1,-1\}.\label{eq:maindual2}
    \end{align}
\end{theorem}

Every primal feasible point of Eq.\eqref{eq:mainprimal} corresponds to an explicit estimation protocol and yields an upper bound on the norm. The commuting family $\{O_P^{(m)}\}_{P \in \mathcal{P}_n}$ together with its eigenvalues forms a feasible point of Eq.\eqref{eq:mainprimal}. Conversely, the estimator range can never fall below the largest attainable value of the target functional, so $\|\mathcal{F}\|_k \geq \sup_{\alpha, \rho} |f_\alpha(\rho)|$, and the right-hand side equals one for the family in Eq.\eqref{eq:main_Fm2}. The upper and lower bounds coincide, which establishes

\begin{theorem}[Collapse]\label{thm:main_collapse} For every $m$, $\|\mathcal{F}_m^{(2)}\|_{2m} = 1$.
\end{theorem}

This norm extends the established theory of measurement incompatibility. Restricted to the linear expectation values of binary POVMs $\mathcal{M} = \{M_{\pm|i}\}_{i=1}^K$, it recovers the $k$-replica homogeneous robustness of measurement incompatibility~\cite{heinosaari2015noise,carmeli2016quantum},
\begin{align}\label{eq:main_robustness}
r_k^*(\mathcal{M}) = \max\{\|\mathcal{F}_{\mathcal{M}}\|_k - 1, 0\},
\end{align}
where $\mathcal{F}_{\mathcal{M}}$ collects the expectation values $\tr[(M_{+|i} - M_{-|i})\rho]$, as shown in SM~\cite{Supplemental}, Appendix B. For a single qubit, $\|\{\tr(X\rho), \tr(Y\rho), \tr(Z\rho)\}\|_1 = \sqrt{3}$ reproduces the incompatibility of the three Pauli measurements~\cite{Busch1986}. For the $n$-qubit Pauli family under site-local measurements, the norm evaluates to $3^{n/2}$, reproducing the exponential lower bound~\cite{huang2020predicting} (SM~\cite{Supplemental}, Appendix D).

Every dual feasible point of Eq.\eqref{eq:maindual2} yields a lower bound on the norm. Symmetry reduction by the Clifford and unitary commutants~\cite{gross2021schur,collins2006integration,bittel2025} makes these bounds analytically computable on standard benchmarks, reproducing the known lower bounds for Pauli shadow tomography~\cite{chen2024optimal} and purity estimation~\cite{gong2024sample} (SM~\cite{Supplemental}, Appendix D).
Interestingly, the dual optimum in Theorem~\ref{thm:main_dual} also takes an operational form, a game in which Alice hides a sign and Bob tries to find it (SM~\cite{Supplemental}, Appendix C). This lifts the known link
between measurement incompatibility and state discrimination~\cite{skrzypczyk2019all} from POVMs to functional families.


\prlsection{Recovering the signs.} The squared moments delivered by Theorem~\ref{thm:main_collapse} already suffice for several applications. The stabilizer Rényi entropy~\cite{leone2022stabilizer} of the distilled state depends only on even powers, so these squared moments directly access its magic. 
The CBDS circuit also potentially accelerates variational quantum eigensolvers~\cite{cao2024accelerated}.
Many other applications, such as energies of many-term Hamiltonians on the distilled state, require the signed moments $\tr(P\rho^m)$, which cannot be determined by the amplitude alone. 
A protocol that recovers those signed moments is then designed to overcome this limitation (SM~\cite{Supplemental}, Appendix A).

We solve this problem by coherently constructing $\rho^{\otimes m}$ and further interfered it with a mimicking state $\chi$ through Bell measurement.
This interference technique jointly estimates $\tr(P\rho^m)\tr(P \chi)$. If the known factor $\tr(P \chi)$ has a large overlap with the significant $\tr(P\rho^m)$ moments, the signed moments can be estimated by dividing out $\tr(P \chi)$.
The protocol therefore has two steps.
First, Bell measurement is used to jointly estimate $q_P=\tr(P\rho P\rho)$ and identify significant moments. 
Built from $q_P$, a random mimicking state $\chi$ with large overlap with those moments is constructed, as a direct generalization of the mimicking state technique in Pauli shadow tomography~\cite{chen2024optimal,king2025triply}.
Then a controlled cyclic permutation circuit (SM~\cite{Supplemental}, Figure 2) are used to jointly estimate $\tr(P\rho^m)\tr(P \chi)$.
Concentration inequality therefore gives the $\varepsilon^{-4}$ dependence of Theorem~\ref{main-pauli-vd}. 

\begin{theorem}\label{main-pauli-vd} For $m \geq 2$, there exists an $m$-replica protocol that estimates the Pauli moments $\{\tr(P\rho^m)\}_{P \in \mathcal{P}_n}$ within additive error $\varepsilon$ with confidence at least $1-\delta$, using at most
\begin{align}\label{eq:main_N}
N = O\left(\frac{m[n + \log(1/\delta)]}{\varepsilon^4}\right)
\end{align}
copies of $\rho$.
\end{theorem}

Three remarks delimit the scope of Theorem~\ref{main-pauli-vd}. 
First, sample mimicking state $\chi$ requires solving a linear satisfiability program that is not always computationally efficient. This makes Theorem~\ref{main-pauli-vd} just an information-theoretic upper bound.
Second, the protocol matches the existing lower bounds~\cite{ye2025replica}, which states that any $(m-1)$-replica protocol requires at least $\Omega(2^{n/2})$ copies to estimate $\tr(P\rho^m)$ for a single $P$.
Third, Theorem~\ref{main-pauli-vd} trades the $\varepsilon^{-4}$ for an exponential saving in $n$. To our knowledge, Eq.\eqref{eq:main_N} is the first sample complexity upper bound for this family that reaches $\mathcal{O}(n)$ and minimum optimal replica number requirement. A detailed comparison table is given in SM~\cite{Supplemental} Appendix A.

\prlsection{Discussion and outlook.}
In this Letter, we show that extending the Bell sampling to its coherent version
can learn all $4^n$ squared Pauli moments $[\tr(P\rho^2)]^2$ with $O\!\left({[n + \log(1/\delta)]}{\varepsilon^{-2}}\right)$ copies of $\rho$. 
The structure of this construction is further revealed by the collapse of the quantum Bernstein norm, which serves as natural measure for the incompatibility between nonlinear functionals of $\rho$.
We also show how to information-theoretically recover the signs of the signed Pauli moments $\tr(P\rho^m)$ using only $O\!\left({m[n + \log(1/\delta)]}{\varepsilon^{-4}}\right)$ copies of $\rho$ and $m$-replica strategy, which is the first efficient protocol that achieves the optimal replica requirement for this family.

Several directions remain open.
First, the $4$-replica protocol for $[\tr(P\rho^2)]^2$ is computationally efficient while the $2$-replica protocol in Theorem~\ref{main-pauli-vd} for signed Pauli moments $\{\tr(P\rho^2)\}_{P\in\mathcal P_n}$ is not. How to design a computationally efficient estimation protocol for signed moments remains open. Second, in analogy with the polynomial methods in the nonparametric estimation~\cite{jiao2015minimax,polyanskiy2019dualizing,wu2020polynomial},
 it would be interesting to use the $\mathrm{QB}_k$-norm to design biased estimators for non-polynomial nonlinear properties including the R\'enyi entropies~\cite{acharya2020estimating} and Fisher information~\cite{vitale2024robust}.

\begin{acknowledgments}
\textit{Acknowledgments.} We thank Chao-Yang Lu, Yu-Hao Deng, Zi-Han Chen, and Boyan Zhang for helpful discussions. This work was supported by the Quantum Science and Technology--National Science and Technology Major Project (Grant No. 2025ZD0300300) and the National Natural Science Foundation of China (Grant No. 12504584).
\end{acknowledgments}


\nocite{cirac1999optimal,childs2025streaming,zhao2026power,zhou2024hybrid,mcclean2016theory,hangleiter2023computational,hangleiter2024bell,fan2012hoeffding,rao2014nonparametric,hoeffding1992class,halmos1946theory,leblanc2012estimating,gharbi2025bernstein,marco2010polynomial,lehmann2011completeness,landsberg2011tensors,dartois2024injective,harrow2013testing,friedlander2014gauge,boyd2004convex,lecam1973convergence,coles2017entropic,busch2007heisenberg,heinosaari2016invitation,designolle2019incompatibility,kadison1952generalized,innocenti2023shadow,mcnulty2023estimating,nguyen2022optimizing,huang2021efficient,hadfield2022measurements,korhonen2025improving,cha2026operator,oszmaniec2017simulating,guerini2017operational,cuffaro2024quantum,dankert2009exact,seddon2021quantifying,skowronek2009cones,horodecki2001separability,guhne2009entanglement}
\bibliography{references}

\clearpage
\onecolumngrid
\setcounter{figure}{0}
\setcounter{table}{0}
\setcounter{algocf}{0}
\makeatletter
\setcounter{theorem}{0}
\@addtoreset{theorem}{section}
\renewcommand{\thetheorem}{\thesection\arabic{theorem}}
\setcounter{proposition}{0}
\@addtoreset{proposition}{section}
\renewcommand{\theproposition}{\thesection\arabic{proposition}}
\setcounter{lemma}{0}
\@addtoreset{lemma}{section}
\renewcommand{\thelemma}{\thesection\arabic{lemma}}
\setcounter{corollary}{0}
\@addtoreset{corollary}{section}
\renewcommand{\thecorollary}{\thesection\arabic{corollary}}
\setcounter{problem}{0}
\@addtoreset{problem}{section}
\renewcommand{\theproblem}{\thesection\arabic{problem}}
\setcounter{definition}{0}
\@addtoreset{definition}{section}
\renewcommand{\thedefinition}{\thesection\arabic{definition}}
\setcounter{remark}{0}
\@addtoreset{remark}{section}
\renewcommand{\theremark}{\thesection\arabic{remark}}
\makeatother

\renewcommand{\thetitle}{Pauli-resolved virtual distillation}
\makeatletter
\newcommand{\appendixtitle}[1]{\gdef\@title{#1}}
\newcommand{\appendixauthor}[1]{\gdef\@author{#1}}
\newcommand{\appendixaffiliation}[1]{\gdef\@affiliation{#1}}
\newcommand{\appendixdate}[1]{\gdef\@date{#1}}
\makeatother

\makeatletter%
\newcommand{\appendixmaketitle}{%
\begin{center}%
\vspace{0.4in}%
{\large \@title \par}%
\end{center}%
\par%
}%
\makeatother%

\makeatletter
\newcommand{\appendixtableofcontents}{%
\begingroup
\setcounter{tocdepth}{3}%
\@starttoc{atoc}%
\endgroup}
\makeatother

\newcommand{\appendixtocdivider}{%
\par\medskip
\noindent\hbox to \linewidth{%
\leaders\hbox{\rule[0.6ex]{1pt}{0.4pt}}\hfill
\hspace{0.8em}\textsc{Table of Contents}\hspace{0.8em}%
\leaders\hbox{\rule[0.6ex]{1pt}{0.4pt}}\hfill}%
\par}

\newcommand{\appsection}[1]{%
\section{#1}%
\addcontentsline{atoc}{section}{\protect\numberline{\thesection}#1}}

\newcommand{\appsubsection}[1]{%
\subsection{#1}%
\addcontentsline{atoc}{subsection}{\protect\numberline{\thesubsection}#1}}

\newcommand{\appsubsubsection}[1]{%
\subsubsection{#1}%
\addcontentsline{atoc}{subsubsection}{\protect\numberline{\thesubsubsection}#1}}

\setcounter{secnumdepth}{3}
\appendix

\appendixtitle{\bf
Supplemental Material for ~"\thetitle"}
\appendixmaketitle
\vspace{0.2in}

The Supplemental Material is organized as follows.

In Appendix A, we prove all the technical results underlying the Pauli-resolved virtual distillation protocol of the Letter. 
 For $m=2$, we give the detailed construction of the coherent Bell difference sampling (CBDS) circuit (Figure 1 in the Letter) and the post-processing algorithm (Algorithm~\ref{alg:four-replica-squared-pauli}). 
For $m>2$, we show how an $m$-replica interference measurement with a mimicking ensemble, combined with a two-replica Bell sampling stage that identifies the active Pauli subset, estimates all signed moments $\{\operatorname{tr}(P\rho^m)\}_{P\in\mathcal{P}_n}$ using $O(m[n+\log(1/\delta)]/\varepsilon^4)$ copies of $\rho$. This proves Theorem 3 of the Letter. 
An example of non-commuting unit-range operators that can be jointly estimated with unit range is constructed as an endnote in Appendix A, leading into the introduction of the quantum Bernstein norm.

In Appendix B, we develop the quantum Bernstein norm $\|\mathcal{F}\|_k$ as an estimator-envelope measure for the joint unbiased estimation of polynomial functionals, extending classical polynomial estimation theory~\cite{lorentz2012bernstein,wu2020polynomial} to quantum states. 
We prove that the norm is the optimal value of a semidefinite program with strong duality (Theorem 1 of the Letter).
In particular, this norm reduces to the measurement incompatibility~\cite{guhne2023colloquium} when restricted to linear expectation values of binary POVMs.
This proves Eq. (7) of the Letter and indicates that the quantum Bernstein norm can be regarded as a measure of incompatibility between nonlinear functionals.

In Appendix C, we collect further properties of the quantum Bernstein norm that have independent interest.
We first discuss the operational meaning of the dual program as a state-discrimination game.
Then we give the proof of the subadditivity of this norm.
We also discuss the relation between the worst-case variance and the quantum Bernstein norm.

In Appendix D,
we benchmark the dual programs on the well-studied families $\{\operatorname{tr}(P\rho)\}_{P\in\mathcal{P}_n}$ and $\{\operatorname{tr}(\rho^2)\}$ under physically motivated measurement restrictions, recovering known results on their estimation complexity~\cite{chen2024optimal,gong2024sample}.

\appendixtocdivider
\appendixtableofcontents

\appsection{Pauli-resolved virtual distillation}\label{app:4-replica-strategy}

\appsubsection{Problem and our contributions}
Given multiple copies of a noisy state $\rho$,
the task of approximately preparing the principal eigenstate has a long history~\cite{cirac1999optimal,childs2025streaming,zhao2026power}.
More recently, focus has shifted from physical to virtual purification schemes~\cite{cai2023quantum,cotler2019quantum,huggins2021virtual}.
Given multiple copies of a noisy state $\rho$, \textit{virtual distillation} estimates nonlinear expectation values of the form $\langle O \rangle_{\mathrm{dist}} = {\tr(O\rho^m)}/{\tr(\rho^m)}$ for a given observable $O$.
Given that $\rho$ has the spectral decomposition $\rho = \sum_{i=1}^d \lambda_i \ket{\psi_i}\bra{\psi_i}$, the term:
\begin{align}
    \tilde{\rho} := \frac{\rho^m}{\tr(\rho^m)} = \sum_{i=1}^d \frac{\lambda^m_i}{\sum_{j=1}^d \lambda_j^m} \ket{\psi_i}\bra{\psi_i} = \sum_{i=1}^d \lambda^*_i \ket{\psi_i}\bra{\psi_i}\label{eq:A1}
\end{align}
concentrates the eigenvalues of $\rho$ to the eigenstates that have the largest eigenvalues, thus virtually distilling a quantum state with fewer orthogonal errors~\cite{huggins2021virtual}.
Typically, when the mixed state is a thermal state of a quantum many-body system with Hamiltonian $H$, the virtually distilled state $\tilde{\rho}$ gives the thermal state
\begin{align}
    \tilde{\rho} = \frac{[\rho(T)]^m}{\tr([\rho(T)]^m)} = \frac{e^{- m \beta H}}{Z(m\beta)} = \rho(T/m).\label{eq:A2}
\end{align}
This virtually cools the Gibbs state to a temperature $T/m$~\cite{cotler2019quantum}.

Many protocols have been proposed to estimate $\langle O\rangle_{\mathrm{dist}}$ for different observables $O$.
In this work, we focus on the full $n$-qubit Pauli observables $\mathcal{P}_n$, the most important observable set in experiments~\cite{o2023purification,seif2023shadow}.
A protocol that estimates $\langle P \rangle_{\mathrm{dist}}$ of all Pauli observables $P$ with a constant overhead is a \textit{Pauli-resolved virtual distillation protocol} (PVD protocol for short).
When $\tr(\rho^m)$ takes a constant value, the PVD protocol is equivalent to the joint estimation of all Pauli observables $\tr(P \rho^m)$ with a constant overhead.

Table~\ref{tab:vd-comparison} compares existing protocols for estimating Pauli power moments with our proposed protocol in the following sections.
For local $Z$ Pauli observables, Ref~\cite{huggins2021virtual} already provides a strategy efficient in qubit number $n$. 
Other protocols based on randomized measurements, SWAP-test circuits, or Bell measurements can be applied to any Pauli observable.
However, these strategies require higher-order unitary designs that are hard to implement in practice, and are not efficient in qubit number $n$. 
In comparison, we propose a protocol that is both sample-efficient and computationally efficient for Pauli moments $[\tr(P \rho^2)]^2$ (Algorithm~\ref{alg:four-replica-squared-pauli}).
We also propose a sample-efficient protocol for signed Pauli moments $\tr(P \rho^m)$ for $m>2$ that meets the optimal replica number requirements (Theorem~\ref{cor:sample-efficient-pauli-amplitude-spectroscopy}).

\begin{table*}[htbp!]
    \centering
    \small
    \caption{
        Comparison of protocols for estimating the Pauli power moments $\tr(P \rho^m)$.
        The column ``Replicas'' indicates the number of unknown copies of $\rho$ that must be collected in one round of this measurement strategy.
        For clarity, the sample complexity is specified by the estimation cost of the \textit{worst-case single target} in the target range.
        Covering all $M$ targets requires $\log(M/\delta)$ overhead, which gives the $\mathcal{O}(n)$ dependence in the main text if $M=4^n$.
    }
    \label{tab:vd-comparison}
    \setlength{\tabcolsep}{4pt}
    \begin{tabular}{@{}lllll@{}}
    \toprule
    \parbox[t]{0.19\textwidth}{\raggedright Protocol}
    & \parbox[t]{0.07\textwidth}{\centering Replicas}
    & \parbox[t]{0.19\textwidth}{\raggedright Target range}
    & \parbox[t]{0.17\textwidth}{\raggedright Sample complexity}
    & \parbox[t]{0.31\textwidth}{\raggedright Circuit requirements} \\
    \midrule
    \parbox[t]{0.19\textwidth}{\raggedright Bell sampling~\cite{huggins2021virtual}}
    & \parbox[t]{0.07\textwidth}{\centering $2$}
    & \parbox[t]{0.19\textwidth}{\raggedright $\tr(P\rho^2)$ for $P\in \{Z_i\}_{i=1}^n$}
    & \parbox[t]{0.17\textwidth}{\raggedright $O(\varepsilon^{-2}\log\delta^{-1})$}
    & \parbox[t]{0.31\textwidth}{\raggedright A local-rotated Bell measurement} \\
    \addlinespace
    \parbox[t]{0.19\textwidth}{\raggedright Clifford shadow~\cite{huang2020predicting}}
    & \parbox[t]{0.07\textwidth}{\centering $1$}
    & \parbox[t]{0.19\textwidth}{\raggedright $\tr(P\rho^m)$ for $P\in \mathcal P_n$}
    & \parbox[t]{0.17\textwidth}{\raggedright $\widetilde O(4^{mn}\varepsilon^{-2}\log\delta^{-1})$}
    & \parbox[t]{0.31\textwidth}{\raggedright Unitary 3-design} \\
    \addlinespace
    \parbox[t]{0.19\textwidth}{\raggedright BRM~\cite{du2026optimal}}
    & \parbox[t]{0.07\textwidth}{\centering $1$}
    & \parbox[t]{0.19\textwidth}{\raggedright $\tr(P\rho^2)$ for $P\in \mathcal P_n$}
    & \parbox[t]{0.17\textwidth}{\raggedright $O\!\left(2^{3n/2} \varepsilon^{-2}\log\delta^{-1}\right)$}
    & \parbox[t]{0.31\textwidth}{\raggedright Unitary 6-design} \\
    \addlinespace
    \parbox[t]{0.19\textwidth}{\raggedright Hybrid shadow~\cite{zhou2024hybrid}}
    & \parbox[t]{0.07\textwidth}{\centering $m$}
    & \parbox[t]{0.19\textwidth}{\raggedright $\tr(P\rho^m)$ for $P\in \mathcal P_n$}
    & \parbox[t]{0.17\textwidth}{\raggedright $O\!\left(2^{n} \varepsilon^{-2}\log\delta^{-1}\right)$}
    & \parbox[t]{0.31\textwidth}{\raggedright CSWAP circuit, ancilla qubits, and unitary 3-design} \\
    \addlinespace
    \parbox[t]{0.19\textwidth}{\raggedright AFRS~\cite{liu2026auxiliary}}
    & \parbox[t]{0.07\textwidth}{\centering $2$}
    & \parbox[t]{0.19\textwidth}{\raggedright $\tr(P\rho^2)$ for $P\in \mathcal P_n$}
    & \parbox[t]{0.17\textwidth}{\raggedright $O\!\left(2^{n} \varepsilon^{-2}\log\delta^{-1}\right)$}
    & \parbox[t]{0.31\textwidth}{\raggedright GHZ-type measurement and unitary 3-design} \\
    \addlinespace
    \parbox[t]{0.19\textwidth}{\raggedright Algorithm~\ref{alg:four-replica-squared-pauli} in this work}
    & \parbox[t]{0.07\textwidth}{\centering $4$}
    & \parbox[t]{0.19\textwidth}{\raggedright $[\tr(P\rho^2)]^2$ for $P\in \mathcal P_n$}
    & \parbox[t]{0.17\textwidth}{\raggedright $O(\varepsilon^{-2}\log\delta^{-1})$}
    & \parbox[t]{0.31\textwidth}{\raggedright Coherent Bell difference sampling circuit} \\
    \addlinespace
    \parbox[t]{0.19\textwidth}{\raggedright Theorem~\ref{cor:sample-efficient-pauli-amplitude-spectroscopy} in this work}
    & \parbox[t]{0.07\textwidth}{\centering $m$}
    & \parbox[t]{0.19\textwidth}{\raggedright $\tr(P\rho^m)$ for $P\in \mathcal P_n$}
    & \parbox[t]{0.17\textwidth}{\raggedright $O(\varepsilon^{-4}\log\delta^{-1})$}
    & \parbox[t]{0.31\textwidth}{\raggedright CSWAP circuit, ancilla qubits, and Bell measurement} \\
    \bottomrule
    \end{tabular}
\end{table*}

\appsubsection{Coherent Bell difference sampling}
\paragraph{Collapse of incompatibility.}
To estimate $\{[\tr(P\rho^m)]^2\}_{P \in \cP_n}$, we note that $f_P(\rho) := [\tr(P\rho^m)]^2$ is a degree-$2m$ polynomial function of $\rho$ with linearization:
\begin{align}
    [\tr(P\rho^m)]^2 = \tr(F_{P}\rho^{\ox 2m}), \quad F_{P} = \left[P_1\cdot W_{ (1,\cdots,m)}\right] \otimes \left[P_{m+1} \cdot W_{(m+1,\cdots,2m)}\right].\label{eq:A3}
\end{align}
Here the two permutation operators $W_{(1,\cdots,m)}$ and $W_{(m+1,\cdots,2m)}$ are the cyclic permutation operators for the first $m$ and last $m$ replicas respectively.
The Pauli operator $P_i$ denotes $P$ acting on the $i$-th replica, with the identity on all other replicas.
Measuring in the eigenbasis of $F_P$ builds an estimator $\hat{f}_P$ with range $\|F_P\|_\infty = 1$.
However, the operators $F_P$ do not commute with each other, which makes this strategy inefficient for learning all $4^n$ Pauli moments jointly.

Fortunately, such a linearization is not unique. Any replica permutation of $F_{P}$ still gives a valid linearization.
Their linear combinations then form an operator space that fully encodes the information of the Pauli moment $[\tr(P\rho^m)]^2$. 
We show that there exists a subset of mutually commuting operators in this space that can be measured jointly with a single projective measurement.

\begin{theorem}\label{thm:UB_nonlin_P_shad}
    For every integer $m\geq 1$, there exists a set of operators $\{O_P\}_{P\in\cP_n}$ supported on $2m$ replicas of the state $\rho$ such that:
    \begin{align}
        [O_P,O_Q]=0, \quad \|O_P\|_\infty\le 1, \quad \tr(O_P\rho^{\ox 2m})=[\tr(P\rho^m)]^2, \quad \forall P,Q\in\cP_n,
    \end{align}
\end{theorem}
    \begin{proof}
    Let $C:=W_{(1,\ldots,m)}\otimes W_{(m+1,\ldots,2m)}$ be the product of the cyclic shifts on the two blocks of $m$ replicas.
    For each $P\in\cP_n$, define
    \begin{align}
        R_{P,i}:=P_iP_{m+i},
        \qquad
        T_P:=\frac{1}{m}\sum_{i=1}^m R_{P,i}.\label{eq:A4}
\end{align}
    Here $P_i$ denotes $P$ acting on the $i$-th replica, with the identity on all other replicas.
    The cyclic shift permutes the paired operators, so
    \begin{align}
        C R_{P,i}C^\dagger=R_{P,i+1},\qquad C T_PC^\dagger=T_P,\label{eq:A5}
\end{align}
    where the pair index is taken modulo $m$.
    Thus $T_P$ commutes with $C$ and $C^\dagger$.
    Contracting each of the two disjoint cycles against $\rho^{\ox 2m}$ gives
    \begin{align}
        \tr(T_PC\rho^{\ox 2m})
        =\frac{1}{m}\sum_{i=1}^m\tr(R_{P,i}C\rho^{\ox 2m})
        =\left[\tr(P\rho^m)\right]^2.\label{eq:A6}
\end{align}
    Since $P$ and $\rho^m$ are Hermitian, $\tr(P\rho^m)$ is real.
    Taking the complex conjugate of Eq.~\eqref{eq:A6} and using $[T_P,C]=0$ therefore gives the same expectation for $T_PC^\dagger$.
    Consequently, the Hermitian operators
    \begin{align}
     \left\{O_P := T_P \cdot \frac{C + C^\dagger}{2}\right\}_{P\in\cP_n}\label{eq:A7}
\end{align}
    satisfy $\tr(O_P\rho^{\ox 2m})=f_P(\rho)$.
    We next show that they mutually commute.
    For any $P,Q\in\cP_n$, the operators $R_{P,i}$ and $R_{Q,j}$ act on disjoint replica pairs if $i\neq j$.
    If $i=j$, then
    \begin{align}
        (P_iP_{m+i})(Q_iQ_{m+i})=(PQ)_i(PQ)_{m+i}=(QP)_i(QP)_{m+i}=(Q_iQ_{m+i})(P_iP_{m+i}),\label{eq:A8}
\end{align}
    because the possible Pauli sign appears twice and cancels. Hence, $[T_P,T_Q]=0$ for all $P,Q\in\cP_n$.
    Together with $[T_P,C]=[T_P,C^\dagger]=0$, this implies
    \begin{align}
        [O_P,O_Q]=0,
        \qquad
        \forall P,Q\in\cP_n .\label{eq:A9}
\end{align}
    Each $R_{P,i}$ is a Hermitian Pauli operator with operator norm one, and $C$ is unitary. Hence
    \begin{align}
        \|T_P\|_\infty\le 1 , \quad \left\|\frac{C+C^\dagger}{2}\right\|_\infty\le 1 \Rightarrow \|O_P\|_\infty\le 1.\label{eq:A10}
\end{align}
\end{proof}

Theorem~\ref{thm:UB_nonlin_P_shad} yields an efficient strategy to estimate the mitigated expectation amplitudes $|\tr(P \rho^m)|^2$ of \textit{any} $M$ Pauli observables.
There exists a common rank-$1$ projective measurement $\{\Pi_j\}_{j\in\cJ}$, with $\sum_j\Pi_j=\1$, and spectral decompositions $O_P=\sum_{j\in\cJ}h_P(j)\Pi_j$, where $h_P(j)\in[-1,1]$.
Measuring $\rho^{\ox 2m}$ with this measurement gives an outcome $j$ with probability $\tr(\Pi_j\rho^{\ox 2m})$. For every $P\in\cP_n$,
\begin{align}
    \left[\tr(P\rho^m)\right]^2=\tr(O_P\rho^{\ox 2m})=\sum_{j\in\cJ}h_P(j)\tr(\Pi_j\rho^{\ox 2m}).\label{eq:A11}
\end{align}
Thus the same outcome $j$ supplies the estimators $\hat f_P=h_P(j)$ for all $P\in\cP_n$, with $\mathbb E_\rho[\hat f_P]=f_P(\rho)$ and $|\hat f_P|\le1$.
By Hoeffding's inequality and a union bound, for an additive error $\varepsilon$ and a confidence level of at least $1-\delta$, a sufficient copy count is $N = 2m\left\lceil 2\varepsilon^{-2}\log(2|\mathcal{P}_n|/\delta)\right\rceil$ to cover all squared Pauli moment targets.
We note that the joint estimator in the above strategy has the same scaling as the individual estimators, which indicates that $f_P$ are compatible with each other.

\paragraph{Compilation of the projective measurement.}
For the case $m=2$, we realize the projective measurement $\{\Pi_j\}_{j\in\cJ}$ by the coherent Bell difference sampling (CBDS) circuit. 
This measurement primitive is of independent interest since it is a coherent generalization of the widely used Bell difference sampling circuit~\cite{gross2021schur,grewal2023efficient,grewal2024improved}.
For simplicity, we relabel the four replicas into the
Bell-pair order $(1,2)$ and $(3,4)$.
Under this relabeling, the commuting
representative is $\{B_{P} := {2}^{-1}{(PP\1_n\1_n + \1_n\1_nPP)} W_{13}W_{24}\}$.
Specifically, we replace the Pauli labels $P$ with binary strings
$\mathbf r=(\mathbf r_z,\mathbf r_x)\in\mathbb F_2^{2n}$.
The corresponding operator $O_P$ changes to $B_{\mathbf{r}}$, where the index $\mathbf{r}$ indicates
\begin{align}
P = \sigma_{\mathbf r}:= i^{\mathbf r_z\cdot\mathbf r_x}\bigotimes_{i=1}^n  X_i^{r_{x,i}}Z_i^{r_{z,i}}.\label{eq:symplectic-representation}
\end{align}
 Similarly, Bell labels are indexed by $\mathbf{a} \in \mathbb F_2^{2n}$ through $\ket{\Phi_{\mathbf a}}:= \bigotimes_{i=1}^n \ket{\Phi_{a_{x,i}a_{z,i}}}$,
 where the Bell basis is taken as
 \begin{align}
     \ket{\Phi_{00}}&:= \frac{\ket{00} + \ket{11}}{\sqrt{2}},\quad \ket{\Phi_{10}}:= \frac{\ket{01} + \ket{10}}{\sqrt{2}},\quad
     \ket{\Phi_{01}}:= \frac{\ket{00} - \ket{11}}{\sqrt{2}},\quad\ket{\Phi_{11}}:= \frac{\ket{01} - \ket{10}}{\sqrt{2}}.\label{eq:Bell-basis}
\end{align}

 We also define $\mathbf{CNOT}_{12}:= \bigotimes_{i=1}^n \mathrm{CNOT}_{i,i+n}$ and $\mathbf{H}_1:= \bigotimes_{i=1}^n \mathrm{H}_{i}$ as the transversal controlled-NOT gate and Hadamard gate acting on the first replica with qubit labels $\{1,2,\cdots,n\}$ and $\mathbf{H}_2,\mathbf{H}_3,\mathbf{H}_4, \mathbf{CNOT}_{34}$ as the parallel gates acting on the 2nd, 3rd, and 4th replicas, respectively.
 One can trivially verify the following identities pointwise:
 \begin{align}
    \ket{\Phi_{\mathbf{a}}} &= \mathbf{CNOT}_{12} \mathbf{H}_1 \ket{\mathbf{a}_z} \otimes \ket{\mathbf{a}_x},\quad (\sigma_{\mathbf{r}} \ox \sigma_{\mathbf{r}}) \ket{\Phi_{\mathbf{a}}} = (-1)^{[\mathbf{r},\mathbf{a}] + \mathbf{r}_x\cdot\mathbf{r}_z}\ket{\Phi_{\mathbf{a}}},\label{eq:4rep-identity}
\end{align}
 where $[\mathbf{r},\mathbf{a}] :=\mathbf{r}_x \cdot \mathbf{a}_z + \mathbf{r}_z \cdot \mathbf{a}_x \text{~mod~} 2$ is the symplectic product on binary fields $\mathbf{a}:= (\mathbf{a}_z, \mathbf{a}_x) \in \mathbb{F}_2^{2n}$.
 We also define an order $\mathbf{a}< \mathbf{b}$ if and only if $\sum_{i = 1}^{2n} 2^{2n-i}a_i< \sum_{i = 1}^{2n} 2^{2n-i}b_i$.
Then the common eigenbasis of $\{B_\mathbf{r}\}_\mathbf{r}$ supported on four replicas has the form:
\begin{align}
   \ket{\Psi^{s}_{(\mathbf{a}, \mathbf{b})}}&=\frac{{\ket{\Phi_{\mathbf{a}}} \ox \ket{\Phi_{\mathbf{b}} }+s \ket{\Phi_{\mathbf{b}}} \ox \ket{\Phi_{\mathbf{a}}}}}{\sqrt{2}}, s \in \{\pm 1\}, \forall \mathbf{a}, \mathbf{b} \in \mathbb{F}_2^{2n}, \mathbf{a} < \mathbf{b},\label{eq:A15}
\\
   \ket{\Psi_{(\mathbf{a}, \mathbf{a})}} &= \ket{\Phi_{\mathbf{a}}}\ket{\Phi_{\mathbf{a}}}, \forall \mathbf{a} \in \mathbb{F}_2^{2n},\label{eq:A16}
\end{align}
with eigenvalues of the operator $B_\mathbf{r}$ respectively given by
\begin{align}
    {\lambda_{\mathbf{r}}(s,\mathbf{a},\mathbf{b})} &:=(-1)^{\mathbf{r}_x\cdot\mathbf{r}_z}\cdot s \cdot \frac{(-1)^{[\mathbf{r}, \mathbf{a}]} + (-1)^{[\mathbf{r}, \mathbf{b}]}}{2}, \quad \lambda_{\mathbf{r}}(\mathbf{a}) = (-1)^{[\mathbf{r},\mathbf{a}] + \mathbf{r}_x\cdot\mathbf{r}_z}.\label{eq:A17}
\end{align}
The projective measurement is then indexed by $\mathcal J:=\{(\mathbf a,\mathbf a):\mathbf a\in\mathbb F_2^{2n}\}\cup\{(s,\mathbf a,\mathbf b):\mathbf a<\mathbf b,\ s\in\{\pm1\}\}$.
The corresponding POVM is $\{\Pi_j\}_{j \in \mathcal{J}} := \{\Pi_{(\mathbf a,\mathbf a)}\}_{\mathbf a\in\mathbb F_2^{2n}} \cup \{\Pi_{(s,\mathbf a,\mathbf b)}\}_{s\in\{\pm1\}, \mathbf a<\mathbf b}$ with projectors defined as
\begin{align}
    \Pi_{\mathbf a,\mathbf a}:= \ket{\Psi_{(\mathbf{a}, \mathbf{a})}}\bra{\Psi_{(\mathbf{a}, \mathbf{a})}}, \quad
    \Pi^s_{\mathbf a,\mathbf b}:=\ket{\Psi^{s}_{(\mathbf{a}, \mathbf{b})}}\bra{\Psi^{s}_{(\mathbf{a}, \mathbf{b})}}.\label{eq:A18}
\end{align}
With respect to this projective measurement, the spectral decomposition of
$B_{\mathbf r}$ is
\begin{align}
    B_{\mathbf r}
    =
    \sum_{\mathbf a}
    \lambda_{\mathbf r}(\mathbf a,\mathbf a)\Pi_{\mathbf a,\mathbf a}
    +
    \sum_{\mathbf a<\mathbf b}\sum_{s=\pm1}
    \lambda_{\mathbf r}(s,\mathbf a,\mathbf b)\Pi^s_{\mathbf a,\mathbf b},\label{eq:A19}
\end{align}
thus giving the post-processing function
$
    h_{\mathbf r}(j)
    :=
    \begin{cases}
    \lambda_{\mathbf r}(\mathbf a,\mathbf a),
    & j=(\mathbf a,\mathbf a),\\[1mm]
    \lambda_{\mathbf r}(s,\mathbf a,\mathbf b),
    & j=(s,\mathbf a,\mathbf b),\ \mathbf a<\mathbf b .
    \end{cases}
$

\refstepcounter{algocf}\label{alg:four-replica-squared-pauli}
\begin{tcolorbox}[
    enhanced,
    breakable,
    colback=white,
    colframe=black!55,
    boxrule=0.4pt,
    arc=1pt,
    left=6pt,
    right=6pt,
    top=6pt,
    bottom=6pt,
    title={Algorithm~\thealgocf: Four-replica estimator for squared nonlinear Pauli signals},
    fonttitle=\bfseries
]
\small
\noindent\textbf{Input.}
An unknown $n$-qubit state $\rho$, a repetition number $T$, and a target Pauli label set
$\mathcal R\subseteq\mathbb F_2^{2n}$.
For $\mathbf r=(\mathbf r_z,\mathbf r_x)\in\mathbb F_2^{2n}$, write
$\sigma_{\mathbf r}:= i^{\mathbf r_x\cdot\mathbf r_z}\bigotimes_{i=1}^n X_i^{r_{x,i}}Z_i^{r_{z,i}},
    \, [\mathbf r,\mathbf a]:= \mathbf r_x\cdot\mathbf a_z+\mathbf r_z\cdot\mathbf a_x\pmod 2 $.

\textbf{Output.}
Estimates $\{\widehat{\theta}_{\mathbf r}\}_{\mathbf r\in\mathcal R}$ of
$\theta_{\mathbf r}=\bigl(\tr(\sigma_{\mathbf r}\rho^2)\bigr)^2$.

\begin{enumerate}[leftmargin=*, itemsep=0.35em]
    \item Initialize $S_{\mathbf r}\gets 0$ for every $\mathbf r\in\mathcal R$.

    \item For each repetition $t=1,\ldots,T$, prepare four fresh copies $\rho^{\otimes 4}$, with $[(j-1) \cdot n + 1, j\cdot n]$ labeling the qubits in the $j$-th replica for $j=1,2,3,4$ respectively.

    \item Apply the circuit $ \mathbf{CNOT}_{13}\mathbf{CNOT}_{24}\mathbf H_3\mathbf H_1\mathbf{CNOT}_{12}\mathbf{CNOT}_{34}$ that transversally acts on the $4$ replicas.

    \item Measure the third copy in the computational basis and get $y_{2n+1},\cdots,y_{3n}$, measure the fourth copy in the computational basis and get $y_{3n+1},\cdots,y_{4n}$, and set:
    \begin{align}
        \mathbf y_3 \gets (y_{2n+1},\cdots,y_{3n}),\quad
        \mathbf y_4 \gets (y_{3n+1},\cdots,y_{4n}), \quad  \boldsymbol{\Delta}:=(\mathbf y_3,\mathbf y_4)\in\mathbb F_2^{2n}. \notag
\end{align}
    \item If $\boldsymbol{\Delta}=\mathbf 0$, i.e., all elements of $\boldsymbol{\Delta}$ are zero, then measure the first and second copies in the computational basis with outcomes $\mathbf{y}_1 = (y_1, \cdots, y_{n})$ and $\mathbf{y}_2 = (y_{n+1}, \cdots, y_{2n})$ respectively.
    Set
    \begin{align*}
        \mathbf a\gets(\mathbf y_1,\mathbf y_2),\qquad
        \mathbf b\gets\mathbf a,\qquad
        j_t\gets(\mathbf a,\mathbf a).
    \end{align*}

    \item If $\boldsymbol{\Delta}\neq\mathbf 0$, define $E:=\{i\in[2n]:\Delta_i=0\},\quad
        I:=\{i\in[2n]:\Delta_i=1\}=\{i_1<i_2<\cdots<i_{|I|}\}.$

    \item On the qubits indexed by $E$, perform the computational-basis measurement and get outcome $\mathbf y^E$. Set $\mathbf a^E\gets\mathbf y^E$.
    \item On the qubits indexed by $I$, perform the GHZ-basis measurement implemented by
        $\left(\prod_{\ell=2}^{|I|}
        \mathrm{CNOT}_{i_1,i_\ell}\right)^{\dagger}
        \mathrm H_{i_1},$
    followed by computational-basis measurement with outcome $\mathbf y^I = (y_1, \cdots, y_{|I|})$.
    Decode
    \begin{align}
        a_{i_1}=0,
        \quad
        \forall 1<r\le |I|, \; a_{i_r}=\begin{cases} y_1 + y_{r}  & \text{if~} s = +1,\\ \bar{y}_1 + \bar{y}_{r} & \text{if~} s = -1.\end{cases},
        \quad
        s=\begin{cases}+1, & y_1=0,\\-1, & y_1=1,\end{cases}.\label{eq:A23}
\end{align}
    Together with $\mathbf a^E=\mathbf y^E$, reconstruct $\mathbf a$ and set $ \mathbf b\gets \mathbf a+\boldsymbol{\Delta}\pmod 2,\qquad
        j_t\gets(s,\mathbf a,\mathbf b)$.

    \item For each target $\mathbf r\in\mathcal R$, compute
    \begin{align*}
        h_{\mathbf r}(j_t)=
        \begin{cases}
        (-1)^{[\mathbf r,\mathbf a]+\mathbf r_x\cdot\mathbf r_z},
        & j_t=(\mathbf a,\mathbf a),\\[0.4em]
        (-1)^{\mathbf r_x\cdot\mathbf r_z}\,
        s\,\dfrac{(-1)^{[\mathbf r,\mathbf a]}+(-1)^{[\mathbf r,\mathbf b]}}{2},
        & j_t=(s,\mathbf a,\mathbf b),
        \end{cases}
    \end{align*}
    and update $S_{\mathbf r}\gets S_{\mathbf r}+h_{\mathbf r}(j_t)$.

    \item Return $  \widehat{\theta}_{\mathbf r}:={T}^{-1}S_{\mathbf r},
        \, \mathbf r\in\mathcal R.$
\end{enumerate}
\end{tcolorbox}

After performing this measurement $T$ times and recording outcomes
$j_1,\ldots,j_T\in\mathcal J$, we estimate the expectation value of
$B_{\mathbf r}$ by $\widehat B_{\mathbf r}:=
    \frac{1}{T}\sum_{t=1}^T h_{\mathbf r}(j_t)$.
This estimator is unbiased, with $\mathbb E[\widehat B_{\mathbf r}]=\tr(B_{\mathbf r}\rho^{\otimes 4}) = (\tr(P\rho^2))^2$.
Moreover, since $|\lambda_{\mathbf r}(j)|\le 1$ for all outcomes
$j\in\mathcal J$, standard concentration bounds apply uniformly over
$\mathbf r$.
To realize these POVMs, we refer to the adaptive construction in~\cite{liu2026auxiliary} and use Eq.~\eqref{eq:4rep-identity} to show that:
\begin{align}
    \left. \begin{aligned}
        \ket{\Psi^{s}_{(\mathbf{a}, \mathbf{b})}}\\ \ket{\Psi_{(\mathbf{a}, \mathbf{a})}} \end{aligned} \right\}
    &= \mathbf{CNOT}_{34}\mathbf{CNOT}_{12}\mathbf{H}_3\mathbf{H}_1 \cdot
\begin{cases}
    \frac{\ket{\mathbf{a}_z \mathbf{a}_x\mathbf{b}_z\mathbf{b}_x} +s \ket{\mathbf{b}_z \mathbf{b}_x\mathbf{a}_z\mathbf{a}_x}}{\sqrt{2}} &\text{if~} \mathbf{a} < \mathbf{b}\\
    \ket{\mathbf{a}_z \mathbf{a}_x\mathbf{b}_z\mathbf{b}_x}&\text{if~} \mathbf{a} = \mathbf{b}
\end{cases} \notag
\\
&= \mathbf{CNOT}_{34}\mathbf{CNOT}_{12}\mathbf{H}_3\mathbf{H}_1 \mathbf{CNOT}_{13}\mathbf{CNOT}_{24}\cdot
\begin{cases}
    \frac{\ket{\mathbf{a}_z \mathbf{a}_x} +s \ket{\mathbf{b}_z \mathbf{b}_x}}{\sqrt{2}} \ket{\mathbf{a}_z + \mathbf{b}_z} \ket{\mathbf{a}_x + \mathbf{b}_x} &\text{if~} \mathbf{a} < \mathbf{b}\\
    \ket{\mathbf{a}_z \mathbf{a}_x} \ket{\mathbf{0} \mathbf{0}}&\text{if~} \mathbf{a} = \mathbf{b}
\end{cases} \notag
\\
&= \mathbf{CNOT}_{34}\mathbf{CNOT}_{12}\mathbf{H}_3\mathbf{H}_1 \mathbf{CNOT}_{13}\mathbf{CNOT}_{24}\cdot
\begin{cases}
    \frac{\ket{\mathbf{a}^{I}} +s \ket{\bar{\mathbf{a}}^{I}}}{\sqrt{2}} \ket{\mathbf{a}^{E}}\ket{\mathbf{a}_z + \mathbf{b}_z} \ket{\mathbf{a}_x + \mathbf{b}_x} &\text{if~} \mathbf{a} < \mathbf{b}\\
    \ket{\mathbf{a}_z \mathbf{a}_x} \ket{\mathbf{0} \mathbf{0}}&\text{if~} \mathbf{a} = \mathbf{b}
\end{cases},\label{eq:A20}
\end{align}
where the second equality disentangles the qubits on the third and fourth copies.
The last equality further separates the qubits in the first and second copies
into two sets $E := \{i \in [2n]\mid a_i = b_i\}$ with identical computational-basis values and $I:= \{i \in [2n]\mid \bar{a}_i = b_i\} := \{i_1<i_2<\cdots<i_{|I|}\}$ with different computational-basis values.
The set $I$ is also ordered by qubit index.
 Thus, we can locally implement the gate sequence $ \mathbf{CNOT}_{13}\mathbf{CNOT}_{24}\mathbf{H}_3\mathbf{H}_1\mathbf{CNOT}_{12}\mathbf{CNOT}_{34}$ and then measure the 3rd and 4th copies in the computational basis $\ket{\mathbf{y}_{3}\mathbf{y}_{4}}$. This gives the difference of vectors $\mathbf{a}$ and $\mathbf{b}$ through $\mathbf{y}_{3} = \mathbf{a}_z + \mathbf{b}_z,\mathbf{y}_{4} = \mathbf{a}_x + \mathbf{b}_x$.
This information then indicates a {partition} of the qubit indices $\{1,\dotsc,2n\}$ of the first and second copies into two sets $E \cup I$.
We then implement the computational basis measurement $\ket{\mathbf{y}^{E}}$ on the qubits in set $E$ and take $\mathbf{a}^{E} = {\mathbf{y}^{E}}$.
For the qubits in set $I$, we note the identity:
\begin{align}
    \frac{1}{\sqrt{2}} \left[{\ket{\mathbf{a}^{I}} + s\ket{\bar{\mathbf{a}}^{I}}} \right] &=
    \begin{cases}
        \frac{1}{\sqrt{2}} \left(\bigotimes_{i \in I} X_i + Z_{i_1}\right) \ket{\bar{\mathbf{a}}^{I}}  &\text{if~} s = -1,\\
        \frac{1}{\sqrt{2}} \left(\bigotimes_{i \in I} X_i + Z_{i_1}\right) \ket{{\mathbf{a}}^{I}}  &\text{if~} s = +1,
    \end{cases} \notag
\\
    &=\left(\prod_{j = 2}^{|I|}\mathrm{CNOT}_{i_1,i_j}\right)^{\dagger} \mathrm{H}_{i_1} \left(\prod_{j = 2}^{|I|}\mathrm{CNOT}_{i_1,i_j}\right) \begin{cases}\ket{\bar{\mathbf{a}}^{I}} &\text{if~}s = -1,\\\ket{{\mathbf{a}}^{I}} &\text{if~}s = 1,\end{cases} \notag
\\
    &=\left(\prod_{j = 2}^{|I|}\mathrm{CNOT}_{i_1,i_j}\right)^{\dagger} \mathrm{H}_{i_1} \ket{y_1\cdots y_{|I|}}, \text{where~} \begin{cases} \mathbf{a}^{I} = (\bar{y}_1,\bar{y}_1 + \bar{y}_2,\cdots,\bar{y}_1 + \bar{y}_{|I|}) &\text{if~}s = -1,\\\mathbf{a}^{I} =  ({y}_1,{y}_1 + {y}_2,\cdots,{y}_1 + {y}_{|I|})  &\text{if~}s = 1.\end{cases}\label{eq:4rep-GHZ-measurement}
\end{align}
Here in the first equality we use the fact that $\mathbf{a}< \mathbf{b}$, which ensures $a_{i_1} = 0$ and gives $Z_{i_1} \ket{\bar{\mathbf{a}}^{I}} = - \ket{\bar{\mathbf{a}}^{I}}$ and $Z_{i_1} \ket{{\mathbf{a}}^{I}} = \ket{{\mathbf{a}}^{I}}$.
In the second equality, we use the decomposition of the Hadamard gate $H = \sqrt{2^{-1}}(X+Z)$ and the conjugation rule of CNOT gates on Pauli matrices.
In the third equality, we use the fact that the final $\text{CNOT}$ gates act immediately before the computational-basis measurement, so they can be absorbed into classical post-processing. 
We note that the arrangement of the CNOT gates in Eq.~\eqref{eq:4rep-GHZ-measurement} can be further optimized to $\mathcal{O}(\ln n)$ depth by using the parallelized version of the GHZ state preparation circuit.

If we perform this GHZ-basis measurement in Eq.~\eqref{eq:4rep-GHZ-measurement} and get the outcome on $I$ as $\mathbf{y}^I = (y_1, \cdots, y_{|I|})$, then the
whole vector $\mathbf{a}^{I} = (a_{i_1}, \cdots, a_{i_{|I|}})$ and the sign $s$ are given by
\begin{align}
    a_{i_1}=0,
    \quad
    \forall 1<r \leq |I|, \; a_{i_r}=\begin{cases} y_1 + y_{r}  & \text{if~} s = +1,\\ \bar{y}_1 + \bar{y}_{r} & \text{if~} s = -1.\end{cases},
    \quad
    s=\begin{cases}+1, & y_1=0,\\-1, & y_1=1,\end{cases}.\label{eq:A22}
\end{align}
Together with $\mathbf a^E=\mathbf y^E$, this reconstructs $\mathbf a$,
and then $\mathbf b$ from the values of $\mathbf{y}_{3} = \mathbf{a}_z + \mathbf{b}_z$ and $\mathbf{y}_{4} = \mathbf{a}_x + \mathbf{b}_x$.
\begin{figure}[htbp!]
    \centering
    \includegraphics[width=0.65\linewidth]{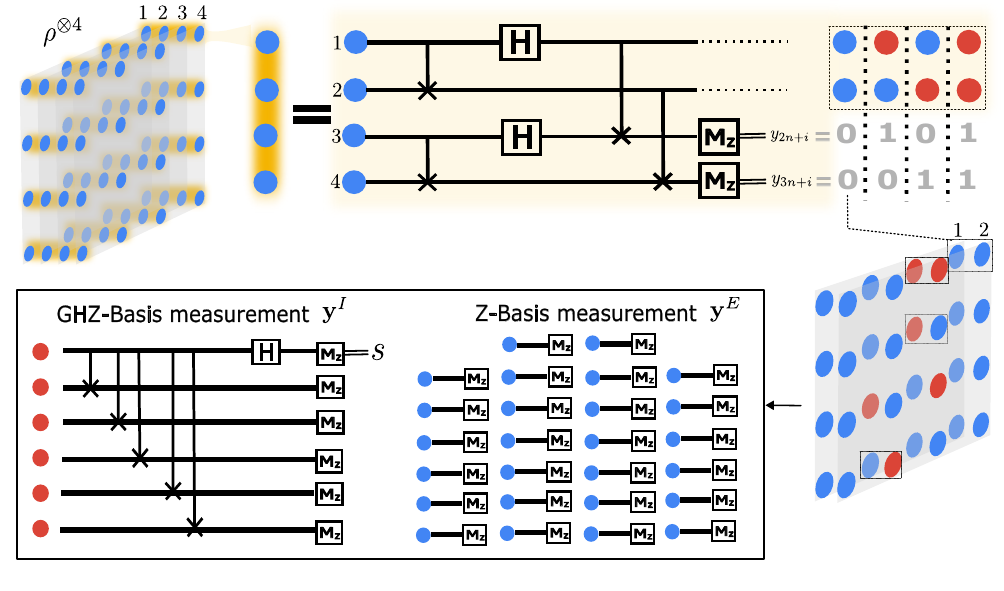}
    \caption{Demonstration of the 4-replica strategy whose measurement results can be used to estimate the nonlinear Pauli average $|\tr(P \rho^2)|$ for all Pauli operators $P$ by post-processing.}
    \label{fig:4rep_strategy}
\end{figure}

\paragraph{Other implications}
Algorithm~\ref{alg:four-replica-squared-pauli} has potential applications in many other fields.
Since a detailed discussion of this is beyond the scope of this work,
we briefly summarize those directions here.

First, in the circuit for Algorithm~\ref{alg:four-replica-squared-pauli}, as shown in Figure~\ref{fig:4rep_strategy},
the outcomes on the third and fourth replicas (i.e., $y_{2n+i}$ and $y_{3n+i}$ for $i = 1,\cdots,n$) have the same distribution as Bell difference sampling~\cite{gross2021schur}, namely,
\begin{align}
\Pr(\mathbf{y}_{34})&=
\sum_{\mathbf{x}\in\mathbb{F}_2^{2n}}
B_{\rho_1,\rho_2}(\mathbf{x})\,
B_{\rho_3,\rho_4}(\mathbf{x}+\mathbf{y}_{34}), \quad B_{\rho_j,\rho_m}(\mathbf{r}) = \frac{1}{2^n}\tr\left(\sigma_{\mathbf{r}}\,\rho_j\,\sigma_{\mathbf{r}}\,\rho_m^*\right)\label{eq:A24}
\end{align}
for $\mathbf{y}_{34} = (\mathbf{y}_3,\mathbf{y}_4) \in \mathbb{F}_2^{2n}$.
However, the whole circuit is inherently different from the ordinary Bell sampling circuit, as it cannot be realized by simply performing Bell sampling twice and then taking the difference of the measurement results by classical post-processing.
Instead, the circuit in Figure~\ref{fig:4rep_strategy} only measures the difference label and preserves coherence inside each unresolved fiber $\{\ket{\Phi_{\mathbf{a}}} \ket{\Phi_{\mathbf{a} + \mathbf{r}}}\}_{\mathbf{a} \in \mathbb{F}_2^{2n}}$
, which is subsequently exploited by the GHZ-type measurement.
The preserved coherence allows us to amplify the signal using GHZ-like measurements, thus enabling the joint estimation of exponentially many mitigated Pauli averages.
Since Algorithm~\ref{alg:four-replica-squared-pauli} only involves low-depth subroutines (Bell state and GHZ state preparations), it is particularly suitable for implementation on noisy intermediate-scale quantum devices, especially neutral-atom systems because of their reconfigurability and high parallelism.

Second, Algorithm~\ref{alg:four-replica-squared-pauli} also has broad applications in quantum chemistry and variational quantum algorithms.
For example, Pauli-term estimation in variational quantum eigensolvers for quantum chemistry~\cite{mcclean2016theory} now can use this primitive to estimate exponentially many mitigated Pauli terms with a single measurement setting.

Third,
many magic resources~\cite{haug2023scalable,leone2022stabilizer} can be regarded as functions of the Pauli spectrum $\{\tr(P \rho)\}_{P \in \cP_n}$.
However, it is still a challenging problem to construct a magic certification protocol for mixed states with constant overhead~\cite{montanaro2017learning,gross2021schur}.
Algorithm~\ref{alg:four-replica-squared-pauli} can be used to certify the magic resource on the virtually distilled state
$\rho^m/\tr(\rho^m)$,
since the purified state $\rho^m/\tr(\rho^m)$, in some cases, pushes the mixed state to its purified counterpart.
This indicates a direction to construct efficient mixed state magic resource certificates.

Last but not least, even though quantum advantage has been demonstrated on many devices, the verification of quantum advantage is extremely challenging~\cite{hangleiter2023computational}.
Recently, the Bell measurement has been introduced to the quantum advantage experiments~\cite{hangleiter2024bell} to solve this problem.
Besides constituting a hard sampling instance, the Bell sampling results allow
many non-linear indicators such as the purity, entanglement, and magic of the output state to be estimated, providing valid and efficient certification of the advantage.
It is thus interesting to investigate if the sampling results of Algorithm~\ref{alg:four-replica-squared-pauli} are hard instances and give noise-robust certification of the quantum advantage.

\appsubsection{Sign recovery and the optimal replica number}\label{app:sign-recovery}
As in Pauli shadow tomography~\cite{chen2024optimal},
the mimicking-state approach~\cite{king2025triply} can be used to recover the signs of the Pauli moments $\tr(P \rho^m)$.
The single-state mimicking condition asks for a state $\sigma$ such that $\bigl||\tr(P\sigma^m)|-|\tr(P\rho^m)|\bigr|\le \varepsilon$ for all $P\in \mathcal{P}_n$.
We generalize this condition by replacing $\sigma$ with a pure-state ensemble $\{(w_j,\chi_j)\}_{j=1}^L$.
Here the weights satisfy $w_j>0$ and $\sum_{j=1}^L w_j = 1$, and each $\chi_j$ is a pure-state density matrix.
We first introduce two oracles and defer their implementation to the subsequent construction.
\begin{definition}[Pauli support oracles]\label{def:mimicking-state-oracles}
For a density matrix $\rho$ and $m\geq 2$, we call $\mathcal{A} \subset \mathcal{P}_n$ an active Pauli subset of $\rho^m$ if and only if $P \notin \mathcal{A} \implies |\tr(P \rho^m)| \leq 3\varepsilon/4$.
We consider the following two oracles for identifying and mimicking the active Pauli subset:
\begin{enumerate}
    \item (\textit{Finding}): Given $\rho$, returns an active Pauli subset $\mathcal{A}$.
    \item (\textit{Mimicking}): Given $(\rho, \mathcal{A})$ from the finding oracle, returns a pure-state ensemble satisfying
    \begin{align}
        \forall P\in \mathcal{A},\quad  \sum_j w_j|\tr(P \chi_j)|^2 \geq 3\varepsilon^2/16.\label{eq:D1}
\end{align}
\end{enumerate}
\end{definition}

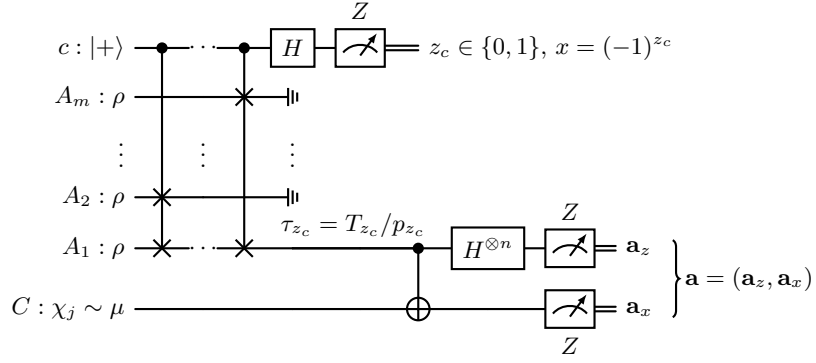
\begin{figure}[htbp!]
    \centering
    \begin{quantikz}[
      row sep=0.28cm, column sep=0.28cm,
      wire types={q,q,n,q,q,q},
      every node/.append style={font=\small}
    ]
    \lstick{$c:\ket{+}$}
      & \ctrl{3} & \push{\cdots} & \ctrl{1}
      & \gate{H} & \meter{Z}
      & \rstick[label style={align=left}]
          {$z_c\in\{0,1\}, \, x=(-1)^{z_c}$}\setwiretype{c} \\
    \lstick{$A_m:\rho$}
      & & & \swap{3} & \ground{} \\
    \lstick{$\vdots$}
      & & \push{\vdots} & & \push{\vdots} \\
    \lstick{$A_2:\rho$}
      & \swap{1} & & & \ground{} \\
    \lstick{$A_1:\rho$}
      & \targX{} & \push{\cdots} & \targX{} &
      &
      & \ctrl{1}\wire[l][2]["{\displaystyle\tau_{z_c}=T_{z_c}/p_{z_c}}"{anchor=south,yshift=2pt}]{q}
      & \gate{H^{\otimes n}} & \meter{Z}
      & \rstick{$\mathbf a_z$}\setwiretype{c}
      &[0.35cm] \rstick[2]{$\mathbf a=(\mathbf a_z,\mathbf a_x)$}\setwiretype{n} \\
    \lstick{$C:\chi_j \sim \mu$}
      & & & & & & \targ{} & & \meter[style={label={below:$Z$}}]{}
      & \rstick{$\mathbf a_x$}\setwiretype{c} & \setwiretype{n}
    \end{quantikz}
    \caption{An $m$-replica interference measurement for estimating the Pauli power moments $\tr(P \rho^m)$.}
    \label{fig:m-interference}
\end{figure}

Using the oracles in Definition~\ref{def:mimicking-state-oracles}, we derive a sample-complexity upper bound for the joint estimation of the signed Pauli moments $\{\tr(P \rho^m)\}_{P \in \mathcal{P}_n}$.
Figure~\ref{fig:m-interference} shows the interference circuit.
Lemma~\ref{lem:signal_estim} describes the estimators obtained by classical post-processing of the measurement records $\{(z_{c,t}, \mathbf{a}_{t})\}_{t = 1}^{R_2}$, where $t$ indexes the interference rounds.

\begin{lemma}[Interference measurement with a mimicking ensemble]\label{lem:signal_estim}
Given access to the oracles in Definition~\ref{def:mimicking-state-oracles}, the interference measurement in Figure~\ref{fig:m-interference} estimates $\tr(P \rho^m)$ for any $P\in \mathcal{P}_n$ within additive error $\varepsilon$, with success probability at least $1-\delta/2$, using at most $R_2 = O((n+\log(1/\delta))\varepsilon^{-4})$ rounds.
\end{lemma}
\begin{proof}
    We use the symplectic representation of the $4^n$ Pauli operators given in Eq.~\eqref{eq:symplectic-representation}.
    Write $f_{\mathbf r}=\tr(\sigma_{\mathbf r}\rho^m)$ for the target Pauli moments.
    In Figure~\ref{fig:m-interference},
    the reference state $\chi_j$ on register $C$ is sampled from the ensemble $\{(w_j,\chi_j)\}_{j=1}^L$, with probability $w_j$.
Registers $A_1,\ldots,A_m$ are prepared in $\rho^{\otimes m}$, and the control qubit $c$ is prepared in $\ket{+}$.
The sequence of controlled-SWAP gates implements the controlled-$W_{[m]}$ operation:
\begin{align}
 W_{[m]}=W_{1m}\cdots W_{13}W_{12},\qquad
 W_{[m]}\ket{i_1,\ldots,i_m}=\ket{i_m,i_1,\ldots,i_{m-1}}.\label{eq:D2}
\end{align}
Apply $H$ to $c$ and measure it in the computational basis, obtaining
$z_c\in\{0,1\}$.
The corresponding unnormalized conditional state of registers $A_1,\ldots,A_m$, given outcome $z_c$, is
\begin{align}
 B_{z_c}=\frac14
 \bigl(I+(-1)^{z_c}W_{[m]}\bigr)\rho^{\otimes m}
 \bigl(I+(-1)^{z_c}W_{[m]}^\dagger\bigr).\label{eq:D3}
\end{align}
Tracing out $A_2,\ldots,A_m$ leaves the normalized conditional state $\tau_{z_c}:=T_{z_c}/p_{z_c}$, where
 \begin{align}
 T_{z_c}:=\tr_{A_2\cdots A_m}(B_{z_c})
       =\frac{\rho+(-1)^{z_c}\rho^m}{2}, \quad
 p_{z_c}:=\tr(T_{z_c})
       =\frac{1+(-1)^{z_c}\tr(\rho^m)}{2},\label{eq:interference-branches}
\end{align}
Next, perform a Bell measurement on $A_1$ and $C$, and record the outcome as $\mathbf a$, using the Bell-basis convention in Eq.~\eqref{eq:Bell-basis}.
From the record $(z_c,\mathbf a)$, define $Y_{\mathbf r}=(-1)^{z_c+[\mathbf r,\mathbf a]+\mathbf r_x\cdot\mathbf r_z}$.
Its conditional expectation satisfies
\begin{align}
 \mathbb E[Y_{\mathbf{r}}\mid \chi_j]
 &=\sum_{z_c=0}^1(-1)^{z_c}
   \tr\!\left[(\sigma_{\mathbf{r}}\otimes \sigma_{\mathbf{r}})(T_{z_c}\otimes\chi_j)\right] =\tr[\sigma_{\mathbf{r}}(T_0-T_1)]\tr(\sigma_{\mathbf{r}}\chi_j)= \mathrm{tr}(\sigma_{\mathbf{r}}\rho^m) \tr(\sigma_{\mathbf{r}}\chi_j),\label{eq:cycle-bell-cross}
\end{align}
The first equality follows from Eq.~\eqref{eq:4rep-identity}.
For the ensemble $\mu=\{(w_j,\chi_j)\}_j$, define the normalization factor $h_{\mathbf{r}}=\sum_jw_j [\mathrm{tr}(\sigma_{\mathbf{r}}\chi_j)]^2$.
The random variable
$Z_{\mathbf{r}}={\mathrm{tr}(\sigma_{\mathbf{r}}\chi_j)\,Y_{\mathbf{r}}}/{h_{\mathbf{r}}}$ is an unbiased estimator of $\tr(\sigma_{\mathbf{r}}\rho^m)$, because
\begin{align}
    \mathbb E Z_{\mathbf r}
 =\sum_jw_j\frac{\tr(\sigma_{\mathbf r}\chi_j)}{h_{\mathbf r}}
       \mathbb E[Y_{\mathbf r}\mid \chi_j]
   =\tr(\sigma_{\mathbf r}\rho^m).\label{eq:D6}
\end{align}
Since $Y_{\mathbf{r}}^2=1$, we obtain the following bounds on the second moment and the centered magnitude of $Z_{\mathbf{r}}$:
\begin{align}
  &\mathbb E Z_{\mathbf r}^2
 =\frac{\sum_jw_j[\tr(\sigma_{\mathbf r}\chi_j)]^2}{h_{\mathbf r}^2}
   =\frac1{h_{\mathbf r}},\label{eq:variance-bound}
\\
&|Z_{\mathbf r}-\tr(\sigma_{\mathbf r}\rho^m)|
 \leq\frac{|\tr(\sigma_{\mathbf r}\chi_j)|\,|Y_{\mathbf r}|}{h_{\mathbf r}}
       +|\tr(\sigma_{\mathbf r}\rho^m)| \leq\frac1{h_{\mathbf r}}+\|\sigma_{\mathbf r}\|_\infty\tr(\rho^m)
   \leq\frac1{h_{\mathbf r}}+1
   \leq\frac2{h_{\mathbf r}} .\label{eq:concentration-bound}
\end{align}
Here we used $|Y_{\mathbf r}|=1$, $\|\sigma_{\mathbf r}\|_\infty=1$,
$|\tr(\sigma_{\mathbf r}\chi_j)|\leq1$, and $\tr(\rho^m)\leq1$.
The last inequality also uses
$0<h_{\mathbf r}=\sum_jw_j[\tr(\sigma_{\mathbf r}\chi_j)]^2
\leq\sum_jw_j=1$.
Repeating the circuit in Figure~\ref{fig:m-interference} independently $R_2$ times produces i.i.d. estimators $\{Z_{\mathbf r,t}\}_{t=1}^{R_2}$.
Using the active subset $\mathcal A$, define the final estimator of $f_{\mathbf r}$ by
\begin{align}
 \widehat f_{\mathbf r}=
 \begin{cases}
 \displaystyle\frac1{R_2}\sum_{t=1}^{R_2}Z_{\mathbf r,t} = \frac{1}{R_2} \sum_{t = 1}^{R_2} \frac{\tr(\sigma_{\mathbf r}\chi_{j_t})}{h_{\mathbf r}} Y_{\mathbf r,t},
       &\sigma_{\mathbf r}\in\mathcal A,\\[0.6em]
 0,    &\sigma_{\mathbf r}\in\mathcal{P}_n\setminus\mathcal A.
 \end{cases}\label{eq:designed-estimator}
\end{align}
For each inactive target $\sigma_{\mathbf r}\in\mathcal{P}_n\setminus\mathcal A$, the error is deterministically bounded by
\begin{align}
 |\widehat f_{\mathbf r}-f_{\mathbf r}|
 =|\tr(\sigma_{\mathbf r}\rho^m)|
 \leq\tfrac34\varepsilon<\varepsilon.\label{eq:D10}
\end{align}
This bound follows from the defining property of the active Pauli subset.
For each active target, the mimicking oracle guarantees
$h_{\mathbf r}>0$, so the single-round estimator is well defined.
Equations~\eqref{eq:variance-bound} and~\eqref{eq:concentration-bound} therefore give
\begin{align}
 \mathbb E\left[\frac{Z_{\mathbf r,t}}{R_2} - \frac{f_{\mathbf r}}{R_2}\right]=0, \quad
 \mathbb{E}\left(\left|\frac{Z_{\mathbf r,t}}{R_2} - \frac{f_{\mathbf r}}{R_2}\right|^2\right) \leq\mathbb E\left[\frac{Z_{\mathbf r,t}^2}{R_2^2}\right] =\frac1{R^2_2 h_{\mathbf r}}, \quad \left|\frac{Z_{\mathbf r,t}}{R_2} - \frac{f_{\mathbf r}}{R_2}\right| \leq\frac2{R_2 h_{\mathbf r}}\label{eq:conditional-estimator-bounds}
\end{align}
Bernstein's inequality~\cite{fan2012hoeffding} yields
\begin{align}
 \Pr\!\left\{
    |\widehat f_{\mathbf r}-f_{\mathbf r}|
       >\varepsilon\right\} &= \Pr\!\left\{
        \left|\sum_{t = 1}^{R_2}\left(\frac{Z_{\mathbf r,t}}{R_2} - \frac{f_{\mathbf r}}{R_2}\right)\right|
           >\varepsilon\right\}
           \leq 2\exp\!\left[-\frac{\varepsilon^2/2}
             {1/(R_2h_{\mathbf r})+2\varepsilon/(3R_2h_{\mathbf r})}\right] \notag
\\
             &\leq2\exp\!\left[-\frac3{10}R_2h_{\mathbf r}\varepsilon^2\right] \leq2\exp\!\left[-\frac{9}{160}R_2\varepsilon^4\right].\label{eq:active-concentration}
\end{align}
The second inequality uses $\varepsilon<1$, and the third uses the guarantee
$h_{\mathbf r}\geq3\varepsilon^2/16$ provided by the mimicking oracle.
Choosing
\begin{align}
 R_2=\left\lceil18\varepsilon^{-4}
                   \log\frac{4|\mathcal{P}_n|}{\delta}\right\rceil = \mathcal{O}((n+\log(1/\delta))\varepsilon^{-4})\label{eq:stage-two-size}
\end{align}
and applying a union bound over all Pauli targets gives
\begin{align}
 \Pr\!\left\{
   \max_{\mathbf r\in \mathbb{F}_2^{2n}}
   |\widehat f_{\mathbf r}-f_{\mathbf r}|
       >\varepsilon\right\}
 \leq\frac\delta2.\label{eq:conditional-simultaneous-guarantee}
\end{align}
All targets use the same $R_2$ measurement records; independence
between different Pauli estimates is unnecessary for this union bound.
\end{proof}

It remains to implement the oracles in Definition~\ref{def:mimicking-state-oracles} using measurements acting jointly on at most two copies of $\rho$.
We first show that estimating the quadratic functionals $q_{\mathbf r}=\tr( \sigma_{\mathbf r} \rho \sigma_{\mathbf r} \rho)$ suffices to implement the finding oracle.

\begin{lemma}[Construction of the finding oracle]
    \label{lem:build_finding}
An oracle that estimates $q_{\mathbf r}=\tr( \sigma_{\mathbf r} \rho \sigma_{\mathbf r} \rho)$ for any $\sigma_{\mathbf r}\in \mathcal{P}_n$ within additive error $\varepsilon^2/16$, with success probability at least $1-\delta/2$, suffices to implement the finding oracle in Definition~\ref{def:mimicking-state-oracles} with success probability at least $1-\delta/2$.
\end{lemma}
\begin{proof}
    Suppose the oracle returns estimates $\widehat q_{\mathbf r}$ satisfying $|\widehat q_{\mathbf r}- q_{\mathbf r}|\leq \varepsilon^2/16$ for every $\sigma_{\mathbf r}\in \mathcal{P}_n$.
    Applying the Hilbert--Schmidt Cauchy--Schwarz inequality $|\tr(A^\dagger B)|^2 \leq \tr(A^\dagger A)\tr(B^\dagger B)$ with $A=\rho^{m-1}$ and $B=\rho^{1/2}\sigma_{\mathbf r}\rho^{1/2}$ gives
    \begin{align}
     |\tr( \sigma_{\mathbf r} \rho^m)|^2
     \leq \tr(\rho^{2m-2})\cdot \tr( \sigma_{\mathbf r} \rho \sigma_{\mathbf r} \rho)\leq \tr(\sigma_{\mathbf r} \rho \sigma_{\mathbf r} \rho).\label{eq:power-quadratic-envelope}
\end{align}
    The last step uses $m\geq2$ and the fact that the eigenvalues of $\rho$
    lie in $[0,1]$ and sum to one.
    Define
    \begin{align}
        \mathcal{A} := \{ \sigma_{\mathbf r} \in \mathcal{P}_n : \widehat q_{\mathbf r} \geq \varepsilon^2/2\}\label{eq:active_set}
\end{align}
    To verify that $\mathcal{A}$ is an active Pauli subset, consider any $\sigma_{\mathbf r} \notin \mathcal{A}$. Then
    \begin{align}
        |\tr( \sigma_{\mathbf r} \rho^m)| \leq \sqrt{q_{\mathbf r}}\leq \sqrt{\widehat q_{\mathbf r}+\varepsilon^2/16} < \sqrt{\varepsilon^2/2 + \varepsilon^2/16} = {3\varepsilon}/{4}.\label{eq:D17}
\end{align}
    This establishes the required guarantee for the finding oracle in Definition~\ref{def:mimicking-state-oracles}.
\end{proof}

Lemma~\ref{lem:build_finding} selects the active Pauli subset $\mathcal{A}$ by thresholding $\widehat{q}_{\mathbf{r}}$.
To construct the corresponding mimicking oracle, we first establish the existence of a pure-state ensemble satisfying the required lower bounds on its squared Pauli expectations.

\begin{lemma}[A reference ensemble with sufficiently large Pauli squares]
  \label{lem:reference-ensemble}
  For every density matrix $\rho$, there is a probability distribution
  $\mu$ over pure states $\psi$ such that, for every Hermitian Pauli $\sigma_{\mathbf{r}}$,
$\mathbb E_{\psi\sim\mu}[\tr(\sigma_{\mathbf{r}}\psi)^2]\geq {\tr(\sigma_{\mathbf{r}}\rho \sigma_{\mathbf{r}}\rho)}/{2} = q_\mathbf{r}/2$.
  \end{lemma}

\begin{proof}
  Write the spectral decomposition of $\rho$ as $\rho=\sum_a\lambda_a|a\rangle\langle a|$.
  Sample an ordered pair $(a,b)$ with probability $\lambda_a\lambda_b$.
  If $a=b$, select $|a\rangle$. Otherwise, select one of the following four pure states uniformly at random:
  \begin{align}
   |\psi_{ab,\omega}\rangle=
   \frac{|a\rangle+\omega|b\rangle}{\sqrt2},
   \qquad \omega\in\{1,-1,i,-i\}.\label{eq:D18}
\end{align}
  For $a\ne b$, let $u=\bra{a}\sigma_{\mathbf{r}}\ket{a}$, $v=\bra{b}\sigma_{\mathbf{r}}\ket{b}$, and $z=\bra{a}\sigma_{\mathbf{r}}\ket{b}$.
  Averaging over the four phases gives
  \begin{align}
   \frac14\sum_{\omega}
   \langle\psi_{ab,\omega}|\sigma_{\mathbf{r}}|\psi_{ab,\omega}\rangle^2
   =\frac{(u+v)^2}{4}+\frac{|z|^2}{2}
   \geq\frac{|z|^2}{2}.\label{eq:D19}
\end{align}
  For $a=b$, the squared Pauli expectation of the selected state is $|\bra{a} \sigma_{\mathbf{r}}\ket{a}|^2$.
  Consequently,
  \begin{align}
    \mathbb E_{\psi\sim\mu}[\tr(\sigma_{\mathbf{r}}\psi)^2] &= \sum_a\lambda_a^2
    |\langle a|\sigma_{\mathbf r}|a\rangle|^2 +\sum_{a\ne b}\lambda_a\lambda_b
    \frac14\sum_\omega
    \langle\psi_{ab,\omega}|
    \sigma_{\mathbf r}|\psi_{ab,\omega}\rangle^2 \notag
\\
    &\geq\frac12\sum_{a,b}\lambda_a\lambda_b|\bra{a} \sigma_{\mathbf{r}}\ket{b}|^2
   =\frac12\operatorname{tr}( \sigma_{\mathbf{r}}\rho \sigma_{\mathbf{r}}\rho).\label{eq:D20}
\end{align}
\end{proof}

We next use Lemma~\ref{lem:reference-ensemble} to establish the feasibility of a finite-net construction of the mimicking oracle for the active subset $\mathcal{A}$ in Eq.~\eqref{eq:active_set}.

\begin{lemma}[Construction of the mimicking oracle]\label{lem:build_mimicking}
    An oracle that estimates $q_{\mathbf r}=\tr( \sigma_{\mathbf r} \rho \sigma_{\mathbf r} \rho)$ for any $\sigma_{\mathbf r}\in \mathcal{P}_n$ within additive error $\varepsilon^2/16$, with success probability at least $1-\delta/2$, suffices to implement the mimicking oracle in Definition~\ref{def:mimicking-state-oracles} for the active subset $\mathcal{A}$ defined in Eq.~\eqref{eq:active_set}, with success probability at least $1-\delta/2$.
\end{lemma}
\begin{proof}
A finite trace-norm $\eta$-net of pure-state density matrices~\cite{chen2022exponential} is a finite set $\mathcal{N}_{\eta} := \{\chi_j\}_j$ such that
\begin{align}
    \forall \psi, \exists \chi_j \in \mathcal{N}_{\eta} \text{ s.t. } \|\psi-\chi_j\|_1\leq\eta.\label{eq:D21}
\end{align}
Such a net can be constructed by rounding complex vectors to a sufficiently fine finite rational coordinate grid.
Fix such a net with $\eta = \varepsilon^2/64$.
Mapping each witness state $\psi$ from Lemma~\ref{lem:reference-ensemble} to a nearest net point induces a distribution $\{w_j\}_j$ on $\mathcal{N}_{\eta}$.
For each such pair $(\psi,\chi_j)$,
\begin{align}
 |\tr(\sigma_{\mathbf{r}}\psi)^2-
   \tr(\sigma_{\mathbf{r}}\chi_j)^2|\leq 2 |\tr(\sigma_{\mathbf{r}}\psi)- \tr(\sigma_{\mathbf{r}}\chi_j)|
 \leq 2\|\psi-\chi_j\|_1\leq2\eta.\label{eq:D22}
\end{align}
The second inequality uses trace-norm duality and the unit operator norm of Pauli operators.
The induced distribution on $\mathcal{N}_{\eta}$ therefore satisfies
\begin{align}
    \sum_j w_j [\tr(\sigma_{\mathbf{r}}\chi_j)]^2 \geq \mathbb{E}_{\psi \sim \mu}[\tr(\sigma_{\mathbf{r}}\psi)^2]-2\eta
 \geq q_{\mathbf{r}}/2-2\eta,\label{eq:D23}
\end{align}
The last inequality follows from Lemma~\ref{lem:reference-ensemble}.
Whenever the oracle estimates satisfy $|\widehat q_{\mathbf{r}}- q_{\mathbf{r}}|\leq \varepsilon^2/16$, the preceding bounds imply that the following finite rational linear feasibility problem is feasible:
\begin{align}
    w_j\geq0,\qquad \sum_j w_j=1,\qquad
    \sum_j w_j [\tr(\sigma_{\mathbf{r}}\chi_j)]^2 \geq \widehat q_{\mathbf{r}}/2 - \varepsilon^2/32 - 2\eta ,
    \quad \forall \sigma_{\mathbf{r}}\in\mathcal A.\label{eq:mimicking-feasible-point}
\end{align}
Any feasible solution $\{w_j\}_j$ defines a pure-state ensemble supported on $\mathcal{N}_{\eta}$.
For each $\sigma_{\mathbf{r}}\in\mathcal A$, Eq.~\eqref{eq:active_set} gives $\widehat q_{\mathbf{r}} \geq \varepsilon^2/2$.
Combining this bound with Eq.~\eqref{eq:mimicking-feasible-point} yields
\begin{align}
    \sum_j w_j [\tr(\sigma_{\mathbf{r}}\chi_j)]^2 \geq  \varepsilon^2/4 - \varepsilon^2/32 -2 \eta = 3\varepsilon^2/16.\label{eq:D25}
\end{align}
\end{proof}

It remains to construct estimators for the quadratic functionals $q_{\mathbf r}=\tr( \sigma_{\mathbf r} \rho \sigma_{\mathbf r} \rho)$.
We use the identities
\begin{align}
    \tr\!\left[W_{(12)}(\sigma_{\mathbf{r}}\otimes \sigma_{\mathbf{r}})\rho^{\otimes2}\right]
    =q_{\mathbf{r}}, \quad
    [W_{(12)}, \sigma_{\mathbf{r}}\otimes \sigma_{\mathbf{r}}] = 0, \quad \forall \sigma_{\mathbf{r}}\in\mathcal{P}_n,\label{eq:D26}
\end{align}
These quantities can be estimated by measuring in the Bell basis, which is a common eigenbasis of $W_{(12)}$ and the operators $\sigma_{\mathbf{r}}\otimes \sigma_{\mathbf{r}}$.
A Bell measurement outcome $\mathbf a=(\mathbf a_z,\mathbf a_x)$ determines the estimator
$Q_{\mathbf{r}}(\mathbf a)=(-1)^{\mathbf a_x\cdot\mathbf a_z+[\mathbf r,\mathbf a]+\mathbf r_x\cdot\mathbf r_z}$.
Perform
\begin{align}
R_1=\left\lceil\frac{2}{(\varepsilon^2/16)^2}\log\frac{4|\mathcal{P}_n|}{\delta}\right\rceil = \mathcal{O}((n+\log(1/\delta))\varepsilon^{-4})\label{eq:D27}
\end{align}
independent Bell measurement rounds and use their outcomes to form $\widehat q_{\mathbf{r}}$ for
$\sigma_{\mathbf{r}}\in\mathcal P_n$.
Hoeffding's inequality gives
$\Pr\{|\widehat q_{\mathbf{r}}-q_{\mathbf{r}}|>\varepsilon^2/16\}\leq2e^{-R_1(\varepsilon^2/16)^2/2} \leq \delta/(2 |\mathcal{P}_n|)$.
A union bound therefore shows that the event $\mathcal E=\{\max_{\sigma_{\mathbf{r}}\in\mathcal P_n}|\widehat q_{\mathbf{r}}-q_{\mathbf{r}}|\leq\varepsilon^2/16\}$
has probability at least $1-\delta/2$.

Combining these estimates with Lemmas~\ref{lem:build_finding} and~\ref{lem:build_mimicking} implements both Pauli support oracles using $R_1$ rounds of two-copy transversal Bell measurements.
Together with Lemma~\ref{lem:signal_estim}, this gives the following theorem.
 \renewcommand{\thetheorem}{3}
\begin{theorem}\label{cor:sample-efficient-pauli-amplitude-spectroscopy}
    For $m\geq 2$, there exists an $m$-replica protocol that estimates $\tr(P \rho^m)$ for any Pauli observable $P$ within additive error $\varepsilon$, with success probability at least $1-\delta$, using at most
    \begin{align}
        2R_1 + m R_2 =
        2 \times \left\lceil 512\varepsilon^{-4}
        \log\frac{4|\mathcal{P}_n|}{\delta}\right\rceil + m \times \left\lceil18\varepsilon^{-4}
        \log\frac{4|\mathcal{P}_n|}{\delta}\right\rceil\label{eq:D28}
\end{align}
    copies of $\rho$.
 \end{theorem}
\renewcommand{\thetheorem}{\thesection\arabic{theorem}}

\paragraph{Sample-complexity lower bound.}
We finally ask whether the replica count in Theorem~\ref{cor:sample-efficient-pauli-amplitude-spectroscopy} is optimal.
The lower bound of Ye \textit{et al.}~\cite{ye2025replica} shows that estimating the $m$-th Pauli moment already requires exponentially many samples in $n$ when joint measurements are restricted to $m-1$ copies.

\begin{lemma}[Specialization of Theorem 3 in~\cite{ye2025replica}]\label{lem:lower-bound-on-sampling-complexity}
    Given the condition
    \begin{align}
    2^n\ge 2m^{2m-2},
    \qquad
    \frac{50m^3}{\sqrt{2^n}}
    \le\varepsilon
    \le\frac{(2m)^{-(m-1)}}{10},\label{eq:D29}
\end{align}
    any $(m-1)$-replica protocol that estimates $\tr(\rho^{m}P)$ for an $n$-qubit Pauli observable $P$ within additive error $\varepsilon$, with success probability $0.9$, requires at least $\Omega\left({2^{n/2}}/{[(m-1)\varepsilon^{1/m}]}\right)$ copies of $\rho$.
\end{lemma}

Together, Lemma~\ref{lem:lower-bound-on-sampling-complexity} and Theorem~\ref{cor:sample-efficient-pauli-amplitude-spectroscopy} establish the optimal replica count for $m\geq 2$ in the parameter regime of Lemma~\ref{lem:lower-bound-on-sampling-complexity}.

\appsubsection{Collapse without a commuting family}
We close this appendix by showing a family of unit-range \textit{noncommuting} operators that can be jointly estimated with unit estimator range.
Consider the following linear functional family:
\begin{align}
\mathcal{F}_{XYZ} := \left\{ f_X(\rho) = \tr\left(\frac{X_1+X_2+X_3}{3}\rho\right),\; f_Y(\rho) = \tr\left(\frac{Y_1+Y_2+Y_3}{3}\rho\right),\; f_Z(\rho) = \tr\left(\frac{Z_1+Z_2+Z_3}{3}\rho\right) \right\},\label{eq:three-qubit-pauli-expectations}
\end{align}
where $X_i, Y_i, Z_i$ are the Pauli operators on the $i$-th qubit for $i = 1, 2, 3$.
Note that the operators $(X_1+X_2+X_3)/3, (Y_1+Y_2+Y_3)/3$, and $(Z_1+Z_2+Z_3)/3$ do not commute and have eigenvalues ranging in $[-1,1]$. Nevertheless, there exists a strategy to estimate every member with the same unit estimator range, as demonstrated below:
\begin{itemize}
    \item Uniformly sample an element from the permutation group $\pi \in S_3$.
    \item Measure $X,Y,Z$ on the $\pi(1), \pi(2), \pi(3)$ qubits with outcomes $r_1, r_2, r_3 \in \{+1, -1\}$ respectively.
    \item Record the measurement results as a tuple $(\pi,r_X, r_Y, r_Z)$.
    \item Construct the unbiased estimators $h_X(\pi,r_X, r_Y, r_Z) := r_X$, $h_Y(\pi,r_X, r_Y, r_Z) := r_Y$, and $h_Z(\pi,r_X, r_Y, r_Z) := r_Z$ for $f_X, f_Y, f_Z$ respectively.
\end{itemize}
This strategy realizes the following $3! \times 2^3 = 48$ outcome POVM labeled by $\pi \in S_3$ and $r_X, r_Y, r_Z \in \{+1, -1\}$:
\begin{align}
    \Pi_{\pi,r_X, r_Y, r_Z} := \frac{1}{6}\left(\frac{\mathbb{I} + r_XX_{\pi(1)}}{2}\right)\left(\frac{\mathbb{I} + r_YY_{\pi(2)}}{2}\right)\left(\frac{\mathbb{I} + r_ZZ_{\pi(3)}}{2}\right).\label{eq:three-qubit-pauli-expectations-POVM}
\end{align}
One can verify the unbiasedness of the estimator by direct calculation, for example,
\begin{align}
   \sum_{\pi \in S_3} \sum_{r_X, r_Y, r_Z \in \{+1, -1\}} r_X \tr(\Pi_{\pi,r_X, r_Y, r_Z} \rho) = \tr\left(\frac{X_1+X_2+X_3}{3}\rho\right) = f_X(\rho).
\end{align}
Furthermore, the joint estimators $h_X, h_Y, h_Z$ have range $[-1,1]$.
This indicates that for general functionals, searching for commutants as in Theorem~\ref{thm:UB_nonlin_P_shad} is not a necessary condition for the existence of a unit range joint estimator.
In the following section, we introduce the $\mathrm{QB}_k$-norm as a criterion that is both necessary and sufficient for the existence of unit range joint estimators. 
Surprisingly, this definition extends the robustness of measurement incompatibility from POVMs to nonlinear functionals, which we believe is of independent interest.

\appsection{The quantum Bernstein norm and its dual programs}\label{app:QBN}\label{app:dual-primal}

\appsubsection{Introduction from non-parametric estimation}
Given samples drawn from a discrete distribution $\mathbf{p}= (p_1,\cdots,p_d)$ in the simplex $\Delta_{d-1}:= \{\mathbf{p}|~ p_j \geq 0,\forall j \in [d],~ \sum_{j} p_j =1\}$, non-parametric estimation~\cite{rao2014nonparametric} aims to directly estimate a function $f(p_1,\cdots,p_d)$ of that distribution.
This problem includes the estimation of various information measures such as Shannon entropy, R\'enyi entropy, the support size, and the $L_q$ norm~\cite{wu2020polynomial}.
A well-known result from nonparametric statistics is the theory of U-statistics~\cite{hoeffding1992class,halmos1946theory}, which states that optimal unbiased estimation of a degree-$k$ polynomial function of a probability distribution $\mathbf{p}$ requires at least $k$ i.i.d.\ samples and can be expressed by a function ${h}: [d]^k \to \mathbb{R}$ such that $\mathbb{E}_{\mathbf{x} \sim \mathbf{p}^{\ox k}} {h}(\mathbf{x}) = f(\mathbf{p})$,
where $\mathbf{x}=(x_1,\cdots,x_k) \sim \mathbf{p}^{\ox k}$ is sampled from the $k$-product distribution $\mathbf{p}^{\ox k}$.

Although the function $f(\mathbf{p})$ may have different forms,
as established in~\cite{lorentz2012bernstein},
any continuous function $f(\mathbf{p})$ of the probability distribution $\mathbf{p}$ can be expanded in the multivariate Bernstein polynomial basis:
$$
b_{\mathbf{j}, k}(\mathbf{p}) = \binom{k}{\mathbf{j}} \mathbf{p}^{\mathbf{j}} = \frac{k!}{j_1! \cdots j_d!} p_1^{j_1} p_2^{j_2} \cdots p_d^{j_d},
$$
where $\mathbf{j} \vdash_d k$ is a partition of $k$ into $d$ parts, i.e., $\mathbf{j} = (j_1, \dots, j_d) \in \mathbb{Z}_{\geq 0}^d$ is a multi-index with $|\mathbf{j}| := \sum_{i=1}^d j_i = k$.
The quantity $\binom{k}{\mathbf{j}}={k!}/(j_1!\cdots j_d!)$ is the multinomial coefficient.

These polynomials are special because they satisfy the probability conditions $b_{\mathbf{j},k} \geq 0$ and $\sum_{\mathbf{j}}b_{\mathbf{j},k} = 1$ (in analogy with a $k$-replica POVM $\Pi_j \succeq 0, \sum_{j} \Pi_j = \mathbb{I}$).
These features help explain why Bernstein polynomials have been widely used in classical nonparametric estimation~\cite{leblanc2012estimating,gharbi2025bernstein}
and provide high relative accuracy in polynomial least-squares fitting~\cite{marco2010polynomial}.
Importantly, as a complete polynomial basis, any polynomial function can be uniquely expanded in this basis, and the expansion coefficients give the U-statistic~\cite{hoeffding1992class} for this polynomial function,
which coincides with the minimum-variance unbiased construction~\cite{halmos1946theory}.
Under this interpretation, the expansion coefficients, as quantified by the Bernstein norm, serve as a natural complexity calibrator for polynomial approximation and minimax estimation of distributional functionals~\cite{wu2020polynomial}.
We summarize these properties and some other known properties of the Bernstein polynomial basis in Lemma~\ref{lem:Bernstein} below.

\begin{lemma}[Estimator bounds via Bernstein coefficients~\cite{wu2020polynomial,lorentz2012bernstein}] \label{lem:Bernstein}
    Let $f(\mathbf{p})$ be a continuous function on the $(d-1)$-dimensional probability simplex $\Delta_{d-1}$ and $\{b_{\mathbf{j}, k}(\mathbf{p})\}_{\mathbf j\vdash_d k}$ be the $k$-th degree multivariate Bernstein polynomial basis. Then:
    \begin{itemize}
        \item The function $B_k(f)(\mathbf{p}) = \sum_{\mathbf{j} \vdash_d k} f\left({\mathbf{j}}/{k}\right) b_{\mathbf{j}, k}(\mathbf{p})$ converges uniformly to $f(\mathbf{p})$ as $k\to\infty$.
        \item Restrict $f(\mathbf{p})$ to a polynomial of degree at most $k$. By degree elevation utilizing $\sum_{i=1}^d p_i = 1$, it admits a unique and exact algebraic expansion $f(\mathbf{p}) \equiv \sum_{\mathbf{j} \vdash_d k} C_{\mathbf{j}}(f) b_{\mathbf{j}, k}(\mathbf{p})$, where $C_{\mathbf{j}}(f)$ are the Bernstein coefficients of $f(\mathbf{p})$.
        \item Given $k$ independent samples $\mathbf{x} \sim \mathbf{p}^{\ox k}$, any unbiased estimator $h(\mathbf{x})$ for $f(\mathbf{p})$ satisfies $\max_{\mathbf{x}} |h(\mathbf{x})| \geq \max_{\mathbf{j} \vdash_d k} |C_{\mathbf{j}}(f)|$. Furthermore, this lower bound is reached by $h(\mathbf{x}) = C_{\mathbf{j}(\mathbf{x})}(f)$, where $\mathbf{j}(\mathbf{x})$ is the empirical count vector of the sample $\mathbf{x}$, i.e., $[\mathbf{j}(\mathbf{x})]_t = \sum_{i=1}^k \mathbb{I}(x_i = t)$.
    \end{itemize}
\end{lemma}

\begin{proof}
    The first item is a well-known result in the theory of Bernstein polynomials, see Theorem 1.1.1 and its multivariate simplex generalization in Section 2.9 of~\cite{lorentz2012bernstein}.
    The second item is trivial since every polynomial $P(\mathbf{p})$ of uniform degree $k$ can be expanded in the monomial basis $P(\mathbf{p})  = \sum_{x_1=1}^d \dots \sum_{x_k=1}^d C_{x_1, \dots, x_k} p_{x_1} \dots p_{x_k}$, using the fact that $p_{x_1}\cdots p_{x_k} = p_{x_{\sigma(1)}}\cdots p_{x_{\sigma(k)}}$, where $\sigma \in S_k$ is a permutation of $[k]$, we can merge the terms with the same empirical count of symbol $\mathbf{j}$.
    The third item is a slight generalization of the well-known proof of the Rao-Blackwell theorem and the Lehmann-Scheffé theorem~\cite{lehmann2011completeness}, from the variance to the infinity norm.
    Given an arbitrary unbiased strategy $h(\mathbf{x}) := h(x_1, \dots, x_k)$, a new symmetric unbiased strategy $\tilde{h}(\mathbf{x}) = (k!)^{-1}\sum_{\sigma \in S_k} h(x_{\sigma(1)}, \dots, x_{\sigma(k)})$ can be constructed such that $\max_{\mathbf{x}\in[d]^k} |h(\mathbf{x})| \geq \max_{\mathbf{x}\in[d]^k} |\tilde{h}(\mathbf{x})|$ by the convexity of the infinity norm. Due to the symmetry, it again depends on the empirical count of symbol $\mathbf{j}$ only and gives $\tilde{h}(\mathbf{j})$. Given that $\mathbf{x} \sim \mathbf{p}^{\ox k}$, we have
    \begin{align}
        \mathrm{Pr}(\mathbf{j}) = \frac{k!}{j_1! \cdots j_d!} p_1^{j_1} p_2^{j_2} \cdots p_d^{j_d} = b_{\mathbf{j}, k}(\mathbf{p}).\label{eq:B1}
\end{align}
    The unbiasedness of $\tilde{h}(\mathbf{x})$ further implies $\mathbb{E}_{\mathbf{x} \sim \mathbf{p}^{\ox k}} [\tilde{h}(\mathbf{x})] = \sum_{\mathbf{j} \vdash_d k} \tilde{h}(\mathbf{j}) b_{\mathbf{j}, k}(\mathbf{p}) = f(\mathbf{p})$, which gives $\tilde{h}(\mathbf{j}) = C_{\mathbf{j}}(f)$ by the uniqueness of the Bernstein expansion.
\end{proof}

Lemma~\ref{lem:Bernstein} gives the Bernstein coefficients a direct
operational meaning in classical estimation theory as it
exactly characterizes the minimum worst-case amplitude required by
any unbiased estimator of a polynomial function.
 More precisely,
for a degree-$k$ polynomial function $f(\mathbf p)$, define the Bernstein norm as
\begin{align}
   \|f\|_k
    :=
    \max_{\mathbf j\vdash_d k}|C_{\mathbf j}(f)| .\label{eq:B2}
\end{align}
Then Lemma~\ref{lem:Bernstein} implies the variational identity
\begin{align}
    \|f\|_k
    =
    \inf_{h:[d]^k\to\mathbb R}
    \left\{
        \|h\|_\infty:
        \mathbb E_{\mathbf x\sim \mathbf p^{\otimes k}}
        h(\mathbf x)=f(\mathbf p),
        \ \forall \mathbf p\in\Delta_{d-1}
    \right\}.\label{eq:B3}
\end{align}
The Bernstein norm $\|f\|_k$ is not related to the functional range $\max_{\mathbf{p}} f(\mathbf{p})$.
In other words, even if the value of $f(\mathbf{p})$ is stable for inputs $\mathbf{p}$, its best unbiased estimator may fluctuate drastically, as quantified by $\|f\|_k$.
This prevents us from constructing efficient estimation strategies for some finite range functionals.

\appsubsection{The quantum Bernstein norm}

In the quantum setting, distributions are generated by positive operator-valued measures (POVMs), and the classical functions $f(\mathbf{p})$ generalize to scalar functions $f(\rho)$ with variables being the matrix entries $\rho_{j,k}$ of the density operator $\rho$.
As in the classical setting where the unbiased estimator $h(\mathbf{x})$ is built on the $k$-replica product extension $\mathbf{p}^{\ox k}$, we first discuss the $k$-linearization of functions in Definition~\ref{def:linearization}.

\begin{definition}[Linearization of $f(\rho)$]\label{def:linearization}
    For a real-valued function $f(\rho)$ with $\rho \in \mathcal{D}(\mathcal{H})$, the $k$-linearization is the operator $F$ acting on $\mathcal{H}^{\ox k}$ such that $f(\rho) = \tr(F\,\rho^{\ox k})$.
\end{definition}

We say $f_\alpha(\rho)$ is a degree-$r$ polynomial function of $\rho$ if it is a degree-$r$ multivariate polynomial in $\rho_{j,k}$.
This includes degree-$r$ polynomial functions such as $\left[\tr(O_i\rho)\right]^r$ and $\tr(O_i\rho^r)$ for $O_i \in \mathrm{Herm}(\mathcal{H})$.
For a degree-$r$ polynomial function of $\rho$, a $k$-replica linearization $F$ with $k\geq r$ acting on $\cH^{\ox k}$, such that $f(\rho) = \tr(F\,\rho^{\ox k})$, always exists.
Furthermore, the linearization operator $F$ that is invariant under $k$-replica permutation operation is unique, as shown in Lemma~\ref{lem:linearization_exist} below.

\begin{lemma}[Linearization of degree-$r$ polynomials]\label{lem:linearization_exist}
     For a real-valued polynomial function $f(\rho)$ with degree $r \leq k$, the $k$-linearization always exists. Furthermore, the linearization $\tilde{F}$ such that $\tilde{F}$ is Hermitian and invariant under $k$-replica permutation operation is unique.
\end{lemma}
\begin{proof}
    To show the existence, we note that a monomial $\mathbf{m} = \prod_{t = 1}^r \rho_{i_t, j_t}$ can always be written as $\tr( A_{m} \rho^{\ox k})$ for a certain linear operator $A_m$ on $\cH^{\ox k}$. By linearity, we show the existence of $k$-linearization for $k \geq r$.
    Using the standard polarization identity for symmetric tensors~\cite{landsberg2011tensors}, we conclude:
    \begin{align}
    \mathrm{span}_{\mathbb R}\{\rho^{\otimes k}:\rho\in\mathcal D(\cH)\} =\{ A \in \mathrm{Herm}(\cH^{\otimes k})~|~ W_{\pi} A W_{\pi^{-1}} = A, \forall \pi \in S_k\},\label{eq:B6}
\end{align}
     where $W_{\pi} \ket{\psi_1} \otimes \cdots \otimes \ket{\psi_k} = \ket{\psi_{\pi(1)}} \otimes \cdots \otimes \ket{\psi_{\pi(k)}}$ defines the representation of $S_k$ on $k$ replicas of the state.
    Thus, $\tr(F\, \rho^{\ox k}) = \tr(F'\, \rho^{\ox k})$ implies $\Phi_k(F) = \Phi_k(F')$, where $\Phi_k(O) = (k!)^{-1}\sum_{\pi \in S_k} W_{\pi} O W_{\pi}^\dagger$ is the symmetrization over the $k$-replica permutation group. This proves the uniqueness of $\tilde{F}:=\Phi_k(F)=\Phi_k(F')$. The Hermiticity can also be proved by the fact that $f^*(\rho) = f(\rho) = \tr(\Phi_k(F)\, \rho^{\ox k}) = \tr(\Phi_k(F)^\dagger\, \rho^{\ox k})$.
\end{proof}

Similar to the classical case, the condition $\max_{\rho \in \mathcal{D}(\mathcal{H})}|f(\rho)|<s$ does not necessarily indicate the existence of a bounded estimator with weights $|h(j)|<s, \forall j \in \cJ$. For a specific quantum function $f(\rho)$, the optimal POVM $\{\Pi_j\}_j$ that minimizes the problem
\begin{align}
   \min_{\{\Pi_j\}_j}[\max_{j} |h(j)|], \quad \mathrm{s.t.~} f(\rho) = \sum_{j} h(j) \tr(\Pi_j \rho^{\ox k}),\quad \Pi_j \geq 0, \quad \sum_j \Pi_j = \1_k,\label{eq:quantum-bernstein-basis}
\end{align}
plays the exact role of the \emph{$k$-replica quantum Bernstein basis}. The expansion coefficients $h(j)$ under this basis fully characterize the existence of bounded unbiased estimators for the quantum function $f(\rho)$.
However, instead of considering a single property $f(\rho)$, the true quantum bottleneck emerges when we consider the \textit{joint estimation} of \textit{a set of} properties $\cF := \{f_{\alpha}(\rho)\}_{\alpha \in \cA}$.

In classical statistics, as established in Lemma~\ref{lem:Bernstein}, Bernstein polynomials are \emph{fixed} for different choices of $f(\mathbf{p})$.
Consequently, joint estimation essentially comes for free.
One can simply evaluate all estimators on the same drawn classical sample $\mathbf{x} \sim \mathbf{p}^{\ox k}$.
The sample complexity overhead is merely logarithmic $O(\log |\cA|)$ via a union bound.
In sharp contrast, the quantum realm strictly forbids this.
Even if each \textit{single} property $f_\alpha(\rho)$ admits a highly efficient quantum Bernstein expansion $(\Pi^{(\alpha)}, h_\alpha)$ with coefficients individually bounded by $s$,
one cannot jointly deploy all these optimal bases $\{\Pi^{(\alpha)}\}_{\alpha \in \cA}$ on a single copy of $\rho^{\ox k}$ unless the POVMs are perfectly compatible (jointly measurable)~\cite{guhne2023colloquium}.
To estimate all $f_\alpha$ jointly from the same ensemble without dataset splitting, one is forced to construct a single \textit{joint} POVM basis $\Pi_{j}$ alongside a new set of estimator coefficients ${h}'_\alpha(j)$.

A familiar example is the joint estimation of the Pauli observables $\{X, Y, Z\}$. Individually, the expectation value of each Pauli operator $P \in \{X, Y, Z\}$ can be perfectly estimated with a bounded estimator $|h_P(\pm)| = 1$ by measuring in its own eigenbasis, which corresponds to the POVM $\Pi^{(P)}_\pm = ({1}/{2})(\1 \pm P)$. However, because these operators do not commute, their optimal individual POVMs are incompatible and cannot be jointly measured~\cite{heinosaari2015noise, carmeli2016quantum}. To jointly estimate all three properties from the same ensemble, one can prove that the optimal joint quantum Bernstein basis must be chosen as:
\begin{align}
    \left\{\tilde{\Pi}_c = \frac{1}{8}\Bigl(\1 + \frac{c_X X + c_Y Y + c_Z Z}{\sqrt{3}}\Bigr), \quad h'_P(c) = \sqrt{3}\,c_P\right\}_{c \in \{+,-\}^{\times 3}},\label{eq:case2-povm}
\end{align}
which inherently incurs a statistical penalty overhead of $s^* = \sqrt{3}$.
We thus generalize the Bernstein framework to both matrix-valued functions and a collection of properties $\cF$.

To quantify the existence of bounded unbiased estimators for a function set $\cF := \{f_{\alpha}(\rho)\}_{\alpha \in \cA}$ and to find the optimal $k$-replica strategy, we formalize the optimization problem~\eqref{eq:quantum-bernstein-basis} in Definition~\ref{def:k-copy-qBern_coeff}.

\begin{definition}[$\mathrm{QB}_k$-norm]\label{def:k-copy-qBern_coeff}
Given $s \in [0,\infty]$, we say that a set of state properties $\cF := \{f_{\alpha}(\rho)\}_{\alpha \in \cA}$ is \textit{$k$-replica jointly estimable with scaling $s$} if there exist a POVM $\{\Pi_j\}_{j \in \cJ}$ on $\cH^{\ox k}$ and classical functions $\{h_\alpha: \cJ \to \mathbb{R}\}$ such that $|h_\alpha(j)| \leq s, \, \forall j \in \cJ, \forall \alpha \in \cA$, and
\begin{align}
f_\alpha(\rho) = \sum_{j \in \cJ} h_\alpha(j)\;\tr(\Pi_j\;\rho^{\ox k}), \quad \forall\,\rho,\;\forall\,\alpha.\label{eq:k-copy-estimateability}
\end{align}
The minimum scaling $s^*_k(\cF) := \min_{s} \{s|~\cF\text{~is~} k\text{-replica jointly estimable with scaling } s \in [0, \infty]\} $ is defined as the $k$-replica quantum Bernstein norm $(\mathrm{QB}_k\text{-norm})$, also denoted $\|\cF\|_k := s^*_k(\cF)$.
The optimal solution $(\Pi^*_j, h^*_\alpha(j))$ that reaches $s^*_k(\cF)$ is defined as the $k$-replica quantum Bernstein basis and $k$-replica quantum Bernstein coefficients for $\cF$, respectively.
\end{definition}

 After defining the $\mathrm{QB}_k$-norm, we remark that the $\mathrm{QB}_k$-norm $s^*_k(\cF)$ of a set of degree-$r$ polynomial properties is determined by the unique $k$-replica symmetric linearization $\{\tilde{F}_{\alpha}\}_{\alpha \in \cA}$ of properties $\cF = \{f_{\alpha}\}_{\alpha \in \cA}$.
 For $k < r$ , the $k$-replica $\mathrm{QB}_k$-norm satisfies $s^*_k(\cF) = \infty$ since no $k$-replica linear estimator can estimate a polynomial of degree greater than $k$.
 This corresponds to Lemma~\ref{lem:linearization_exist}, which states the non-existence of $k<r$ linearization of degree-$r$ polynomial properties.
For the case $k >r$, we first show Lemma~\ref{lem:linearization_symmetrization}.

\begin{lemma}\label{lem:linearization_symmetrization}
    Let $(\Pi_j, h_\alpha(j))$ be an unbiased $k$-replica estimation strategy for a set of properties $\cF := \{f_{\alpha}(\rho)\}_{\alpha \in \cA}$ with symmetric $k$-linearization $\tilde{F}_\alpha$. Then $\sum_j h_\alpha(j)\,\Phi_k(\Pi_j) = \tilde{F}_\alpha$, where $\Phi_k(\Pi_j) = (k!)^{-1}\sum_{\pi \in S_k} P_\pi \Pi_j P_\pi^\dagger$ is the symmetrization over the symmetric group $S_k$.
\end{lemma}
\begin{proof}
    By Definition~\ref{def:k-copy-qBern_coeff} and Definition~\ref{def:linearization}, $f_\alpha(\rho) = \sum_j h_\alpha(j)\,\tr(\Pi_j\,\rho^{\ox k}) = \tr(F_\alpha\,\rho^{\ox k}), \forall \rho \in \mathcal{D}(\mathcal{H}), \forall \alpha \in \cA.$
    By the uniqueness of symmetric linearization, as proved in Lemma~\ref{lem:linearization_exist},
    this indicates $\sum_j h_\alpha(j)\,\Phi_k(\Pi_j) = \tilde{F}_\alpha$.
\end{proof}

The linearization operators $\{{F}_{\alpha}\}_{\alpha \in \cA}$ of $\cF := \{f_{\alpha}(\rho)\}_{\alpha \in \cA}$ can be considered as a set of linear functions $\{\tr(\sigma \cdot {F}_{\alpha})\}_{\alpha \in \cA}$ for an enlarged state $\sigma \in \mathcal{D}(\mathcal{H}^{\ox k})$.
For simplicity, we denote this set of linear functions as $\{F_{\alpha}\}_{\alpha \in \cA}$.
Based on Definition~\ref{def:k-copy-qBern_coeff}, one can also ask for the best single-copy estimation scaling factor $s_1^*(\{F_{\alpha}\}_{\alpha \in \cA})$ for the linear functions $\{F_{\alpha}\}_{\alpha \in \cA}$.
Based on Lemma~\ref{lem:linearization_symmetrization}, the unbiasedness condition of strategy $(\Pi_j, h_\alpha(j))$ also implies the unbiasedness condition on the linear functions $\{F_{\alpha}\}_{\alpha \in \cA}$.
We then have the following theorem.
\begin{theorem}[Linearization theorem]\label{thm:structure_fcomp}
    Given a set of real-valued functions $\cF$ with degree $r<k$ and one of its $k$-linearizations $\{{F}_{\alpha}\}_{\alpha \in \cA}$, we have $s_k^*(\cF) \leq s_1^*\left(\{F_{\alpha}\}_{\alpha \in \cA}\right)$ and equality holds if ${F}_{\alpha} = \tilde{F}_{\alpha}$ is the unique symmetric and Hermitian $k$-linearization (see Lemma~\ref{lem:linearization_exist}) of $\cF$.
\end{theorem}
\begin{proof}
        To prove $s_k^*(\cF) \leq s_1^*\left(\{F_{\alpha}\}_{\alpha \in \cA}\right)$, we suppose a $k$-linearization $\{F_\alpha\}$ can be single-copy unbiasedly estimated with scaling factor $s$ on $\cH^{\ox k}$. Then there exists a POVM $\{\Pi_j\}$ on $\cH^{\ox k}$ and bounded estimators $h_\alpha(j)$ with $|h_\alpha(j)| \leq s$ and $\sum_j h_\alpha(j)\,\Pi_j = F_\alpha$. Then for all $\rho$,
        \begin{align}
        \sum_j h_\alpha(j)\,\tr(\Pi_j\,\rho^{\ox k}) = \tr\!\left[\sum_j h_\alpha(j)\,\Pi_j\;\rho^{\ox k}\right] = \tr(F_\alpha\,\rho^{\ox k}) = f_\alpha(\rho),\label{eq:B10}
\end{align}
        which is exactly $k$-replica unbiased estimation of $\cF$.
        To show the equality, we prove $s_k^*(\cF) \geq s_1^*\left(\{\tilde{F}_{\alpha}\}_{\alpha \in \cA}\right)$ given that $\tilde{F}_{\alpha}$ is the unique symmetric $k$-linearization of $\cF$.
         Suppose $\cF$ can be $k$-replica unbiasedly estimated by $(\Pi_j, h_\alpha(j))$ with $|h_\alpha(j)| \leq s^*_k(\cF)$.
         Using Lemma~\ref{lem:linearization_symmetrization}, we have $\sum_j h_\alpha(j)\,\Phi_k(\Pi_j) = \tilde{F}_\alpha$.
        Then $\{\Phi_k(\Pi_j),h_\alpha(j)\}_j$ is a strategy on $\cH^{\ox k}$ satisfying $\sum_j h_\alpha(j)\,\tr(\Phi_k(\Pi_j)\,\sigma) = \tr(\tilde{F}_\alpha\,\sigma)$ for all states $\sigma\in\mathcal D(\cH^{\otimes k})$. By construction, $\{\tilde{F}_\alpha\}$ admits a feasible single-copy estimation with scaling $s^*_k(\cF)$ via $\{\Phi_k(\Pi_j), h_\alpha(j)\}$, which implies $s_k^*(\cF) \geq s_1^*\left(\{\tilde{F}_{\alpha}\}_{\alpha \in \cA}\right)$.
    \end{proof}
In what follows, without further specification, the $k$-linearization $\{F_{\alpha}\}_{\alpha}$ of $\cF := \{f_{\alpha}(\rho)\}_{\alpha}$ is always taken to be the unique symmetric one.
For example, for a set of linear functions $\{\tr(O_i\rho)\}$ for $\rho \in \mathcal{H}$, its $k$-linearization is the \emph{symmetric embedding} of $O_i$ into $\cH^{\ox k}$:
\begin{align}
\1^{\ox(k-1)}\odot O_i := \frac{1}{k}\sum_{i=1}^{k} \1^{\ox(i-1)}\ox O_i \ox \1^{\ox(k-i)},\label{eq:sym-embed}
\end{align}
which satisfies $\tr\left((\1^{\ox(k-1)}\odot O_i)\,\rho^{\ox k}\right) = \tr(O_i\rho)$ for all $\rho$.

\appsubsection{Relation to measurement incompatibility}
One of the most distinctive features of quantum theory is the incompatibility and unavoidable disturbance of measurement.
This feature gives rise to uncertainty relations for noncommuting observables~\cite{coles2017entropic,busch2007heisenberg}
and precludes the joint estimation of their values.
This notion of uncertainty can be refined through the resource theory of joint measurability~\cite{guhne2023colloquium, heinosaari2016invitation}.

\begin{definition}
    Consider a set of positive-operator-valued measures (POVMs) $\{E_{a|x}\}_{a}$,
with $E_{a|x}\geq 0$ and $\sum_a E_{a|x}=\1$ for $x=1,\ldots,M$.
This set is jointly measurable if there exists a parent POVM $\{G_\lambda\}$ such that $E_{a|x} = \sum_\lambda p(a|x,\lambda)\,G_\lambda$.
\end{definition}
This definition means all measurement settings $x \in \{1,\cdots,M\}$ can be jointly realized by first implementing $\{G_{\lambda}\}_{\lambda}$,
and then applying classical post-processing through $E_{a|x} = \sum_\lambda p(a|x,\lambda)\,G_\lambda$.
Different POVMs are not jointly measurable in general.
A famous example is that the Pauli measurements are not jointly measurable.
 However, their unsharp versions are jointly measurable, as shown in Theorem~\ref{thm:unsharp-pauli-jointly-measurable} below.

\begin{theorem}[Joint measurability of unsharp Pauli measurements~\cite{Busch1986}]\label{thm:unsharp-pauli-jointly-measurable}
Define the noisy unsharp versions of Pauli measurements as $E^{\eta_x}_{a|x = X} = (1/2)\cdot (\1 + a \eta_x X)$, $E^{\eta_y}_{a|x = Y} = (1/2)\cdot (\1 + a \eta_y Y )$  and $E^{\eta_z}_{a|x = Z} = (1/2)\cdot (\1+a\eta_z Z)$, with outcomes $a\in \{+,-\}$  and $0 \leq \eta_x, \eta_y, \eta_z \leq 1$.
The triple is jointly measurable if and only if $(\eta_x)^2 + (\eta_y)^2 +(\eta_z)^2 \leq 1$. The corresponding joint measurement, whose marginals are the unsharp Pauli observables, is $G(x, y, z) = \frac{1}{8}(\1 + x\eta_x X + y\eta_y Y + z\eta_z Z)$.
\end{theorem}
Here the admixture of $\1$ can be interpreted as a $(1-\eta_P)$ fraction of random guessing, i.e., depolarizing noise on Pauli measurement $P$.
$$
E^{\eta_{P}}_{a|x = P} = \frac{1}{1+r_P} \frac{\1 \pm P}{2} +  \frac{r_P}{1+r_P} \frac{\1}{2} =  \frac{\1 \pm P\eta_P}{2}, P \in \{X, Y, Z\}, \eta_{P} = \frac{1}{1+r_P}
$$
 Based on this example,
the robustness of incompatibility~\cite{heinosaari2015noise,designolle2019incompatibility} is used to
quantify this minimum amount of noise, or imprecision,
required to render incompatible POVMs jointly measurable.

\begin{definition}[Homogeneous robustness of incompatibility~\cite{heinosaari2015noise}]
Given a set of binary POVMs $\cM:= \{M_{\pm|x}\}_{x}$ on a $d$-dimensional Hilbert space, its homogeneous robustness $r^*_1(\cM)$ of incompatibility is defined as the minimum $r \in [0,\infty]$ such that for each $x$, the POVM $\{(1+r)^{-1}M_{\pm|x} + r(1+r)^{-1}{\1}/{2}\}_{x}$ is compatible.
\end{definition}

For $\mathcal{M}$ being the three Pauli projectors, the homogeneous robustness of incompatibility is exactly $r_1^{*}(\mathcal{M}) = \sqrt{3} - 1$.
In \cite{carmeli2016quantum}, the authors extend the single replica compatibility to the $k$-replica compatibility.
Given $K$ POVMs $\{M_{a|x}\}_{x = 1}^K, a \in \Omega_x$ on outcome sets $\{\Omega_1,\ldots,\Omega_K\}$ and $k$ copies of a quantum state $\rho$, a single POVM $\{G_{\lambda}\}$ with outcomes $(a_1,\ldots,a_K) \in \Omega_1\times\cdots\times\Omega_K$ on $\rho^{\ox k}$ can be designed such that for each $i$, the marginal statistics of $G$ on $\Omega_i$ reproduce the measurement statistics of $\{M_{a|i}\}$ on $\rho$.
In the case $k \geq K$, this is always achievable by implementing $\{M_{a|1}\}\ox \cdots \ox \{M_{a|K}\} \ox \1^{\ox (k-K)}$ on $\rho^{\ox k}$, and the marginals of the $i$-th copy correctly reproduce the individual statistics of the $i$-th distribution. In the case $k < K$, the problem becomes non-trivial.

\begin{definition}[Multicopy compatibility~\cite{carmeli2016quantum}]\label{def:multicopy-compatibility}
Given $K$ POVMs $\{M_{a|x}\}_{x = 1}^K$ on outcome sets $\{\Omega_i\}_{i=1}^K$, they are $k$-replica compatible if there exists a POVM $\{G_{\lambda}\}$ with outcomes $\lambda = (a_1,\ldots,a_K) \in \Omega_1\times\cdots\times\Omega_K$ on the product state $\rho^{\ox k}$, such that
\begin{align}
\sum_{\lambda:\,\lambda_x = a} \tr\!\left[\rho^{\ox k}\, G_\lambda \right] = \tr\!\left[\rho\, M_{a|x}\right],\quad \forall x = 1, \ldots, K, \forall a \in \Omega_x.\label{eq:k-replica-compatibility}
\end{align}
\end{definition}

For generic POVM sets, there may exist no $k$ with $1<k<K$ such that $\{M_{a|x}\}_x$ is $k$-replica compatible.
To quantify the improvement from lifting to the $k$-replica case,
we thus define the homogeneous robustness of the $k$-replica incompatibility on binary POVMs as follows:

\begin{definition}[k-replica homogeneous robustness]\label{def:k-replica-homogeneous-robustness}
 Given a set of binary POVMs $\cM:=\{M_{\pm|x}\}_{x = 1}^K$ on a $d$-dimensional Hilbert space, its homogeneous robustness $r^*_k(\cM)$ of $k$-replica incompatibility is defined as the minimum $r \in [0, \infty]$ such that the POVM set $\left\{(1+r)^{-1}M_{\pm|x} + r(1+r)^{-1}{\1}/{2}\right\}_{x}$ is $k$-replica compatible.
\end{definition}

We now show that the $\mathrm{QB}_k$-norm problem is related to the standard compatibility problem.
For traceless Hermitian operators $O_i$ with bounded eigenvalues, $\|O_i\|_{\infty} \leq 1$, each induces a two-outcome POVM $\{M_{\pm|i}:= {(\1 \pm O_i)}/{2}\}$.
We then note that the $\mathrm{QB}_k$-norm problem of the set of properties $\cF:= \{\tr\left(O_i\rho\right)\}_i$ is equivalent to the $k$-replica joint measurability problem of the POVMs $\{M_i^{(+)},M_i^{(-)}\}$, as proved in Theorem~\ref{thm:equiv_boundedop}.
A familiar example is the Pauli compatibility problem: given a set of Pauli observables $\{P_i\}_i$, jointly estimating their averages within constant additive error is as hard as compatibly realizing their 2-outcome POVMs $\{(\1 \pm P_i)/{2}\}$.
\begin{theorem}[Equivalence for bounded operators]\label{thm:equiv_boundedop}
    Given a set of binary POVMs $\cM:=\{M_{\pm|i}\}_{i=1}^K$,
    define $\cF:= \{\tr(O_i\rho)\}_{i=1}^K$ with $O_i = M_{+|i} - M_{-|i}$.
    Then the homogeneous robustness of $\mathcal{M}$'s $k$-replica incompatibility satisfies $r^*_{k}(\cM)= \max(\|\cF\|_k-1, 0)$.
\end{theorem}

\begin{proof}
    Suppose $\{M_{+|i}, M_{-|i}\}_{i=1}^K$ are $k$-replica jointly measurable with robustness $r = r^*_k(\cM)$.
    By Definition~\ref{def:k-replica-homogeneous-robustness}, there exists a parent POVM $\{G_\lambda\}_\lambda$ supported on $\cH^{\ox k}$ and conditional distributions $p(a|i,\lambda)$ with $a \in \{+1,-1\}$ such that:
\begin{align}
\1^{\ox(k-1)}\odot \frac{M_{a|i} + r\cdot 2^{-1}\1 }{1 + r} = \sum_\lambda p(a|i,\lambda)\, \Phi_k(G_\lambda), \quad \forall\, a, i.\label{eq:C71}
\end{align}
where $\Phi_k(O) = (k!)^{-1}\sum_{\pi \in S_k} W_{\pi} O W_{\pi}^\dagger$ is the symmetrization over the $k$-replica permutation group.
Taking the difference $a=+1$ minus $a=-1$:
\begin{align}
     \sum_\lambda h_i(\lambda)\,\Phi_k(G_\lambda)&:= (1+r)\sum_\lambda \bigl[p(+1|i,\lambda) - p(-1|i,\lambda)\bigr]\, \Phi_k(G_\lambda)\label{eq:C72}
\\
     &=  \1^{\ox(k-1)}\odot  (M_{+|i} - M_{-|i}) = \1^{\ox(k-1)}\odot  O_i.\label{eq:C73}
\end{align}
Since $p(+1|i,\lambda)$ and $p(-1|i,\lambda)$ are both non-negative and sum to 1, we have $|p(+1|i,\lambda) - p(-1|i,\lambda)| \leq 1$, hence $|h_i(\lambda)| \leq (1+r)$.
 Based on Definition~\ref{def:k-copy-qBern_coeff} and the fact that $\tr(O_i\rho) = \sum_\lambda h_i(\lambda)\,\tr(G_\lambda\rho^{\ox k})$ for all $\rho$, we have $\|\cF\|_k \leq 1+r = r^*_k(\cM) + 1$.
Conversely, suppose that $\|\cF\|_k=s$. By the definition of
$\|\cF\|_k$, there exists a POVM $\{\Pi_j\}_j$ on $\cH^{\ox k}$
and real-valued post-processing functions $\{h_i(j)\}_{i,j}$ such that
\begin{align}
    |h_i(j)|\le s,\quad\tr(O_i\rho) = \sum_j h_i(j)\tr(\Pi_j\rho^{\ox k}),\quad \forall \rho,\ \forall i .\label{eq:C74}
\end{align}
Equivalently, after symmetrizing the POVM elements, we may assume
without loss of generality that
\begin{align}
    \sum_j h_i(j)\,\Phi_k(\Pi_j)=\1^{\ox(k-1)}\odot O_i,\quad \forall i.\label{eq:C75}
\end{align}
Indeed, $\{\Phi_k(\Pi_j)\}_j$ is still a POVM and gives the same
statistics on product states $\rho^{\ox k}$.
If $s\le 1$, define $p(\pm1|i,j):={(1\pm h_i(j))}/{2}$.
Since $|h_i(j)|\le 1$, these are valid conditional probabilities.
Moreover,
\begin{align}
    \sum_j p(\pm1|i,j)\,\Phi_k(\Pi_j)
    &=
    \frac{1}{2}\sum_j \Phi_k(\Pi_j)\pm \frac{1}{2}\sum_j h_i(j)\Phi_k(\Pi_j)
    = \frac{\1^{\ox k}\pm \1^{\ox(k-1)}\odot O_i}{2}\label{eq:C76}
\\
&=\1^{\ox(k-1)}\odot\frac{\1\pm (M_{+|i} - M_{-|i})}{2} =\1^{\ox(k-1)}\odot M_{\pm|i},\label{eq:C77}
\end{align}
where in the last equality we used the fact that $M_{+|i} + M_{-|i} = \1$.
Thus the original binary POVMs are already $k$-replica compatible,
and hence $r^*_{k}(\cM)=0$.
Now assume $s>1$.
Define $p(\pm1|i,j):={1\pm s^{-1}h_i(j)}/{2}$.
Since $|h_i(j)|\le s$, these are valid conditional probabilities. Similarly to the case $s\le 1$, we have:
\begin{align}
    \sum_j p(\pm1|i,j)\,\Phi_k(\Pi_j)
    &=\frac{1}{2}\sum_j \Phi_k(\Pi_j)\pm\frac{1}{2s}
    \sum_j h_i(j)\,\Phi_k(\Pi_j) =\1^{\ox(k-1)}\odot\frac{\1\pm s^{-1}O_i}{2}.\label{eq:C78}
\end{align}
Notice that the last term is the noisy binary POVMs as defined in Definition~\ref{def:k-replica-homogeneous-robustness}, as:
\begin{align}
    \frac{\1\pm s^{-1}O_i}{2} = \frac{s^{-1}(s - 1)\1 + s^{-1}\1 \pm s^{-1}(M_{+|i} - M_{-|i})}{2}=\frac{1}{1+(s-1)}M_{\pm|i} + \frac{s-1}{1+(s-1)}\frac{\1}{2},\label{eq:C79}
\end{align}
Hence $r^*_{k}(\cM)\le s-1$.
Combining this with the first direction gives $r^*_{k}(\cM)= \max(\|\cF\|_k-1, 0)$.
\end{proof}
Theorem~\ref{thm:equiv_boundedop} shows that the compatibility problem of
binary POVMs $\{M_{\pm|i}\}_{i=1}^K$ is exactly equivalent to the
$\mathrm{QB}_k$-norm problem for the associated scalar linear
functions $\{\tr(O_i\rho)\}_i$ with $O_i = M_{+|i} - M_{-|i}$.
Thus, the scalar QB norm recovers the usual notion of measurement
incompatibility robustness in the binary linear setting.

We emphasize that the restriction to binary POVMs is not merely a
technical convenience, but is naturally matched to the scalar function
framework considered in this work. A binary POVM is completely specified
by a single real-valued function $\tr(O_i\rho)$.
One may also consider more general noise models. Instead of
mixing with the unbiased random-guessing POVM $\{\1/2,\1/2\}$, one may
allow arbitrary POVMs $N_{a|x}$ to be mixed with $M_{a|x}$, as in the
general robustness framework~\cite{skrzypczyk2019all}. Developing the
corresponding function-estimation interpretation of such generalized
robustness measures is an interesting direction for future work.

\appsubsection{The bias-variance trade-off}
The $\mathrm{QB}_k$-norm is restricted to unbiased estimation
and therefore does not fully reflect the sampling complexity of practical estimation protocols.
Motivated by the classical polynomial approximation~\cite{wu2020polynomial}, biased estimation can be incorporated by optimizing bias and variance.
Given $f(\mathbf{p})$ as a functional of a classical distribution and $h(\mathbf{x})$ as a biased estimator,
the minimax risk is the worst-case mean-squared error~\cite{jiao2015minimax,wu2020polynomial}
\begin{align}
    R := \min_{h(\mathbf{x})}\max_{\mathbf{p}} \mathbb{E}_{\mathbf{x} \sim \mathbf{p}^{\ox k}} [h(\mathbf{x}) - f(\mathbf{p})]^2.\label{eq:minimax-risk}
\end{align}
The direct calculation of this optimization problem is generally intractable.
Thus, a general methodology is to first select a biased polynomial $P(\mathbf{p})$ that approximates $f(\mathbb{p})$ with a bias $|f(\mathbf{p}) - P(\mathbf{p})| \leq b$.
Then we seek a variance-minimizing unbiased estimator for the polynomial $P$.
Although its Bernstein norm $\|P\|_k$ is not the variance itself, it provides a tight bound
\begin{align}
    \operatorname{Var}_{\mathbf p}(h)
    \leq
    \mathbb E_{\mathbf p}[h^2(\mathbf{x})]
    \leq
    \|P\|_k^2 .\label{eq:bn_bd_variance}
\end{align}
By minimizing both $b$ and $\|P\|_k$, many efficient biased estimators for nonsmooth functionals such as Shannon entropy, R\'enyi entropy and support size can be constructed~\cite{wu2020polynomial}.    

 For the quantum case, by analogy, we seek degree-$k$ polynomials $w_\alpha(\rho)$ and residuals $g_\alpha(\rho)$ with $|g_\alpha(\rho)|\le b$ such that $f_\alpha=w_\alpha+g_\alpha$ and  $\|\{w_{\alpha}(\rho)\}_{\alpha \in \cA}\|_k\le r$.
The parameter $b$ quantifies how much uniformly bounded bias must be mixed into $f$ in order to obtain a degree-$k$ polynomial whose $\mathrm{QB}_k$-norm is at most $r$.
This leads to the following bias-variance trade-off notion.

\begin{definition}[Bias-variance trade-off]\label{def:bias-robustness}
    Given a set of properties $\cF=\{f_\alpha\}_{\alpha\in\cA}$,
    the robustness of $k$-replica incompatibility of $\cF$ is $s_{k,b}^*(\cF):=\min\Bigl\{r\geq 0 \;\Big|\; \exists\, \mathcal{G},
    \mathcal{W}:={\cF+\cG}
    \text{~satisfies~}
    \|\mathcal{W}\|_k\le r
    \Bigr\}$,
    where $\mathcal{G} =\{g_\alpha(\rho)\}_{\alpha\in\cA}$ such that $|g_{\alpha}(\rho)| \leq b, \forall \rho, \forall \alpha \in \cA$.
\end{definition}
As defined in Definition~\ref{def:bias-robustness}, by restricting to the case $b = 0$, it can be verified that $s_{k,0}^*(\cF)$ is identical to $ s_k^*(\cF) = \|\cF\|_k$, which is the $\mathrm{QB}_k$-norm of $\cF$.
In the case that $\cG$ contains non-zero functions, the value of $s_{k,b}^*(\cF)$ reflects the sampling complexity of joint estimation of $\cF$ when a bias $g_{\alpha}(\rho)$  for $f_{\alpha}(\rho)$ is allowed for each $\alpha \in \cA$ respectively.
More generally, any biased joint estimation of $\cF = \{f_{\alpha}(\rho)\}_{\alpha \in \cA}$ can be built through unbiased estimation of auxiliary sets and post-processing functions.

\begin{definition}[Compositional biased estimator]
    Let $\mathcal F=\{f_\alpha\}_{\alpha\in\mathcal A}$
    be a target function set.
    Suppose that there exist auxiliary function
    sets $\mathcal{G}_1,\dots,\mathcal{G}_r$,
    where $\forall i, \|\mathcal{G}_i\|_k \leq r$ and $\mathcal G_i$ is adaptively chosen based on the sampling result of previous function sets $\mathcal{G}_1,\cdots,\mathcal{G}_{i - 1}$, and along with these auxiliary function sets,
    we can set up the classical post-processing map
    $Q_{\alpha}: \mathbb{R}^{|\mathcal{G}_1|} \times \cdots \times \mathbb{R}^{|\mathcal{G}_r|} \to \mathbb{R}$ such that
   \begin{align}
    \left|f_\alpha(\rho)-Q_\alpha(\mathcal{G}_1(\rho),\dots,\mathcal{G}_r(\rho))\right|\le b, \forall \rho,\alpha.\label{eq:E18}
\end{align}
    Then we say the function set $\mathcal F$ admits a $k$-replica compositional biased estimator with scaling factor $r$ and bias $b$.
    \end{definition}

    For example, the memory-assisted Pauli estimation protocol of Ref.~\cite{chen2024optimal}
    can be interpreted within this compositional biased-estimation viewpoint for
    the target family $\cF=\{\tr(P\rho)\}_{P\in\cP_n}$.
    First, one estimates the degree-$2$ auxiliary family
    $\mathcal{G}_1:=\{g_{1,P}(\rho)=[\tr(P\rho)]^2\}_{P\in\cP_n}$,
    which provides estimates of the magnitudes $|\tr(P\rho)|$.
    Using these estimates, one then adaptively chooses an auxiliary state $\sigma$
    satisfying $\bigl|\tr(P\sigma)-\sqrt{g_{1,P}(\rho)}\bigr|\leq \varepsilon$,
    and estimates the linear-in-$\rho$ auxiliary family
    $\mathcal{G}_2:=\{g_{2,P}(\rho)=\tr(P\rho)\tr(P\sigma)\}_{P\in\cP_n}$.
    The final estimate is obtained by the post-processing map
    \begin{align}
        Q_P(\mathcal{G}_1,\mathcal{G}_2)
        =
        \frac{g_{2,P}}{\tr(P\sigma)}
        \cdot \mathbbm{1}\{g_{1,P}\geq 4\varepsilon^2\}.\label{eq:E19}
\end{align}
    Since both auxiliary families admit bounded unbiased estimators with
    $\|\mathcal{G}_1\|_k=\|\mathcal{G}_2\|_k=1$,
    and the thresholding step only discards components with small Pauli expectation,
    we have
    \begin{align}
        \bigl|Q_P(\mathcal{G}_1,\mathcal{G}_2)-\tr(P\rho)\bigr|
        =
        \bigl|\tr(P\rho)\bigr|\,
        \mathbbm{1}\{|\tr(P\rho)|\leq 2\varepsilon\}
        \leq 2\varepsilon .\label{eq:E20}
\end{align}
    Therefore, the Pauli expectation-value family
    $\cF=\{\tr(P\rho)\}_{P\in\cP_n}$
    admits a $k$-replica compositional biased estimator with scaling factor $1$.
    We also note that the primitive strategy for efficient Pauli-resolved virtual distillation, as discussed in Section~\ref{app:sign-recovery}, also shares the same adaptive structure.
    These examples suggest that the $\mathrm{QB}_k$-norm can serve not only as a
    complexity measure for unbiased joint estimation, but also as a useful
    building block for designing multi-stage biased estimation protocols for
    general functional families.
    We leave the systematic development of this
    multi-stage biased framework for future work.

\appsubsection{Derivation of the primal and dual problems}
In this section, we write the conic convex program for the calculation of the bias-robustness $s_{k,b}^*(\cF)$ (Definition~\ref{def:bias-robustness}) when our target function $\cF$ contains polynomial functions with degree $r \leq k$.
When $b = 0$, this problem reduces to the standard $\mathrm{QB}_k$-norm $\|\cF\|_k$ (Definition~\ref{def:k-copy-qBern_coeff}).
Before that, we note that the POVM operators $\{\Pi_j\}_j$ might also be restricted to a certain physical convex set $\mathcal{C}$ among all possible POVMs on $k$-replica space $\cH^{\ox k}$.
Many physical restrictions on the POVM operators exist.
Our conic convex program framework naturally accommodates hardware
limitations by restricting the primal POVM variables
to specific convex cones $\mathcal{C}$. For illustration, defining
$\mathcal{H} = (\mathbb{C}^2)^{\otimes n}$ as the
single-copy Hilbert space and $\mathcal{H}^{\otimes k}$ as the $k$-replica space, we specify two physical restrictions on the POVM operators.

 \paragraph{Memory free strategy.}
Although $k$ copies of the unknown state, $\rho^{\otimes k}$, are available,
the quantum device can only store and jointly measure
$c < k$ copies at a time.
 The protocol proceeds in
$L$ adaptive rounds with $k = c_1 + \cdots + c_L, c_{\ell} \leq c$.
Take  $\Pi_{j_\ell | j_{<\ell}}^{(r)} \in
\mathcal{L}(\mathcal{H}^{\otimes c_{\ell}})$ to be the POVM
element in round $\ell$ conditioned on prior outcomes.
Then in the whole protocol, a $k$-copy POVM:
\begin{align}
    \Pi_j = \bigotimes_{\ell=1}^{L}
    \Pi_{j_\ell | j_{<\ell}}^{(\ell)}
   ,\quad \text{~s.t.~}
    \forall\,j_{\ell}, \,  j_{<\ell} = (j_1, \ldots, j_{\ell-1}),\;
    \Pi_{j_\ell | j_{<\ell}}^{(\ell)} \succeq 0,\;
    \sum_{j_\ell} \Pi_{j_\ell | j_{<\ell}}^{(\ell)}
    = \1^{\otimes c_{\ell}},\label{eq:B12}
\end{align}
is
applied,
with each POVM element depending on all previous outcomes,
and the post-measurement state is discarded.
Specifically, in the case $c = 1$, the resulting convex restriction on POVMs has the form:
\begin{align}
    \mathcal{C}_{\mathrm{mem}}
    = \mathrm{conv}\!\left\{ \bigotimes_{\ell=1}^{k}
    \Pi^{(\ell)} \;|\; \Pi^{(\ell)} \in \mathrm{Herm}(\mathcal{H}) , \,  0 \preceq \Pi^{(\ell)} \preceq  \1
    \right\}.\label{eq:block-positive-cone-mem}
\end{align}
Specifically, when $c = 1$, this restriction recovers positive semidefinite operators on $\mathcal{H}^{\otimes k}$ that are separable with respect to the replica indices.

\paragraph{Local multi-copy strategy.}
In the distributed setting, $n$ parties labeled by $l = 1,\cdots,n$ each hold $k$ qubits so that the $k$-replica state $\rho^{\otimes k} \in \mathcal{H}_0^{\otimes k}$ is
shared across sites with the $l$-th party holding
the $l$-th qubit of each copy. Then each party
independently measures their $k$ qubits. This gives the convex restriction on POVMs:
\begin{align}
    \mathcal{C}_{\mathrm{loc}} = \mathrm{conv}\!\left\{ \bigotimes_{l=1}^n \Pi_{j_l}^{(l)} \;|\; \Pi_{j_l}^{(l)} \in \mathrm{Herm}(\mathcal{H}^{\otimes k}) , \, 0 \preceq \Pi_{j_l}^{(l)} \preceq \1^{\otimes k} \right\}.\label{eq:site-product-cone-loc}
\end{align}
This is the natural measurement constraint for
local shadow protocols where each site applies
an independent $k$-qubit unitary and measures in
the computational basis.

After specifying the physical restriction $\mathcal{C}$, based on Theorem~\ref{thm:structure_fcomp}, the value $s_k^*(\mathcal{W}) := \|\mathcal{W}\|_k$ in Definition~\ref{def:bias-robustness} fully depends on the $\mathrm{QB}_1$-norm of its $k$-replica symmetric linearization $\{W_{\alpha}\}_{\alpha \in \cA}$.
The primal problem of $s_{k,b}^*(\cF)$ is then given by:

\begin{equation}
 {
\begin{aligned}
s_{k,b}^*(\cF,\mathcal{C}) := \min\quad & r\\
\text{Subject to:}\quad &
\sum_j h_\alpha(j)\Pi_j = W_\alpha,\qquad \forall \alpha,\\
& \sum_j \Pi_j=\1^{\otimes k},\qquad \Pi_j\in \mathcal C,\\
& |h_\alpha(j)|\le r,\qquad \forall \alpha,j,\\
& \left|\tr(W_\alpha\rho^{\otimes k})-f_\alpha(\rho)\right|\le b,
\qquad \forall \rho,\ \forall \alpha.
\end{aligned}}\label{eq:primal-robustness}
\end{equation}

If we further restrict the target function $\cF = \{f_\alpha\}_{\alpha \in \cA}$ to be a set of polynomials with degree $r < k$, in the last equation, $f_\alpha$ can also be replaced by its symmetric $k$-linearization $\{F_\alpha\}_{\alpha \in \cA}$ and transformed to the restriction
\begin{align}
   \|W_\alpha - F_{\alpha}\|_{s,\mathrm{inj}} \leq b, \quad \text{where~} \|A\|_{s,\mathrm{inj}} := \sup_{\rho \in \cD(\cH)} \left|\tr\left(A\rho^{\otimes k}\right)\right|.\label{eq:injective-norm}
\end{align}
It is easy to see that $\|\cdot\|_{s,\mathrm{inj}}$ forms a semi-norm on $\mathrm{Herm}(\cH^{\ox k})$ and becomes a norm on the space $\mathrm{V}_k:= \mathrm{span}_{\mathbb R}(\rho^{\otimes k}: \rho \in \cD(\cH))$, which we call the symmetric injective norm~\cite{dartois2024injective,harrow2013testing}.
The supremum in Eq.\eqref{eq:injective-norm} can then be extended to the convex body $\mathcal E_k := \mathrm{conv}\{\rho^{\otimes k}, \rho \in \cD(\cH)\}$, which forms a convex subset of $\mathrm{V}_k$.
To characterize the dual problem of Eq.\eqref{eq:primal-robustness}, we first introduce the dual norm of the symmetric injective norm.

\begin{lemma}[Dual norm of symmetric injective norm]\label{lem:dual-atomic}
    Given $X \in \mathrm{Herm}(\cH^{\ox k})$, the dual norm of the symmetric injective norm $\|\cdot\|_{s,\mathrm{inj}}$, defined as $ \|X\|_{s,\mathrm{inj}}^\circ :=\sup_{\|A\|_{s,\mathrm{inj}}\le 1}\tr(AX)$, is given by $ \|X\|_{s,\mathrm{inj}}^\circ = +\infty$ if $X \notin \mathrm{V}_k$, otherwise it is given by:
    \begin{align}
        \|X\|_{s,\mathrm{inj}}^\circ
        =
        \inf\Bigl\{\sum_{i=1}^m c^{i,+} + c^{i,-}:\;
        X=\sum_{i=1}^m c^{i,+}\,\rho^{\otimes k}_{i,+} - c^{i,-}\,\rho^{\otimes k}_{i,-},\quad
        \rho_{i,+}, \rho_{i,-}\in\cD(\cH), c^{i,+}, c^{i,-} \geq 0\Bigr\}, \forall X \in \mathrm{V}_k.\label{eq:dual-atomic}
\end{align}
\end{lemma}

\begin{proof}
For arbitrary $\sigma = \sum_{i} p_i \rho_i^{\otimes k} \in \mathcal{E}_k$, we conclude that:
\begin{align}
    |\tr(A \sigma)| = \left|\sum_{i} p_i \tr(A \rho^{\otimes k})\right| \leq \sum_{i} p_i |\tr(A \rho_i^{\otimes k})| \leq \sup_{\rho \in \mathcal{D}(\mathcal{H})} |\tr(A \rho^{\otimes k})|.\label{eq:B18}
\end{align}
Besides, $\sup_{\sigma \in \mathcal{E}_k}|\tr(A \sigma)| \geq \sup_{\rho \in \mathcal{D}(\mathcal{H})} |\tr(A \rho^{\otimes k})|$ is trivially satisfied since $\{\rho^{\otimes k}| \rho \in \mathcal{D}(\mathcal{H})\} \subset \mathcal E_k := \mathrm{conv}\{\rho^{\otimes k}, \rho \in \cD(\cH)\}$.
Thus, we conclude $\sup_{\sigma \in \mathcal{E}_k}|\tr(A\sigma)| =\sup_{\rho \in \mathcal{D}(\mathcal{H})} |\tr(A \rho^{\otimes k})|$.
The restricted symmetric injective norm $\|A\|_{s,\mathrm{inj}}
=\sup_{\omega\in\mathcal E_k}|\tr(A\omega)|$ then can be transformed to the supremum over the balanced hull $\mathrm{bal}(\mathcal E_k):=\mathrm{conv}(\mathcal E_k\cup -\mathcal E_k)$ with
\begin{align}
\|A\|_{s,\mathrm{inj}} &:= \sup_{\rho \in \cD(\cH)} \left|\tr\left[A\rho^{\otimes k}\right]\right| = \sup_{\sigma \in \mathcal{E}_k}  \left|\tr\left[A\cdot \sigma\right]\right|=\sup_{\sigma \in \mathcal{E}_k}\max_{c = \{+1,-1\}} \tr\left[A\cdot (c\cdot \sigma)\right] \notag
\\
&=\sup_{\omega \in \mathcal E_k\cup -\mathcal E_k}\tr(A\omega)=\sup_{\omega\in\operatorname{bal}(\mathcal E_k)}\tr(A\omega),\label{eq:B19}
\end{align}
where in the last equality, we use the fact that $\tr(A\omega)$ is a linear function, thus its supremum over $\mathcal E_k\cup -\mathcal E_k$ is the same as over its convex hull $\operatorname{bal}(\mathcal E_k)$.
Thus, on the real span $\mathrm V_k$, the quantity $\|\cdot\|_{s,\mathrm{inj}}$ is precisely the support function of the closed, convex, balanced body $K:=\operatorname{bal}(\mathcal E_k)$; therefore, by the standard gauge--support duality for balanced convex sets~\cite{friedlander2014gauge}, its dual norm is the Minkowski function of $K$.
\begin{align}
\|X\|_{s,\mathrm{inj}}^\circ
:=\sup_{\|A\|_{s,\mathrm{inj}}\le 1}\tr(AX)
=\inf\{t\ge 0: X\in t\,\operatorname{bal}(\mathcal E_k)\}.\label{eq:B20}
\end{align}
By expanding $X \in \mathrm{bal}(\mathcal E_k)$ as a signed combination of vertex elements of the convex set $\mathcal E_k$, we obtain the claimed decomposition.
To prove the case $\|X\|_{s,\mathrm{inj}}^\circ = +\infty$ when $X \notin \mathrm{V}_k$, we note there exists $A_0 \in \mathrm{Herm}(\cH^{\ox k})$ such that $\mathrm{tr}(A_0 \rho^{\otimes k}) = 0$ for all $\rho \in \cD(\cH)$ while $\tr(X A_0) \neq 0$. Then for any $\lambda\in \mathbb{R}$ such that $ \lambda\tr(A_0X) > 0$, we have
\begin{align}
    \|X\|_{s,\mathrm{inj}}^\circ := \sup_{\|A\|_{s,\mathrm{inj}}\le 1}\tr(AX) \geq \tr(\lambda A_0 X) = |\lambda|\cdot |\tr(A_0 X)| \xrightarrow{|\lambda| \to +\infty} +\infty.\label{eq:B21}
\end{align}
\end{proof}

Based on Lemma~\ref{lem:dual-atomic}, the dual norm $\|X\|_{s,\mathrm{inj}}^\circ$ indicates the minimum cost to decompose any $k$-replica observable $X \in \mathrm{Herm}(\cH^{\ox k})$ into a signed difference of $k$-replica moments of two mixed state ensembles $\mathcal{E}_+ := \{p_{i,+},\rho_{i,+}\}_{i=1}^m$ and $\mathcal{E}_- := \{p_{i,-},\rho_{i,-}\}_{i=1}^m$.
By restricting $\tr(X)$, the optimal solution in Eq.~\eqref{eq:dual-atomic} gives the following decomposition
\begin{align}
    X &= \sum_{i=1}^m c^{i,+}\,\rho^{\otimes k}_{i,+} - \sum_{i=1}^m c^{i,-}\,\rho^{\otimes k}_{i,-} =\left(\sum_{i=1}^m c^{i,+}\right) \cdot  \sum_{i=1}^m p_{i,+}\,\rho^{\otimes k}_{i,+} -\left(\sum_{i=1}^m c^{i,-}\right) \cdot \sum_{i=1}^m p_{i,-}\,\rho^{\otimes k}_{i,-} \notag
\\
    &= \frac{t(X) + \tr(X)}{2} \sum_{i=1}^m p_{i,+}\,\rho^{\otimes k}_{i,+} - \frac{t(X) - \tr(X)}{2} \sum_{i=1}^m p_{i,-}\,\rho^{\otimes k}_{i,-}.\label{eq:B22}
\end{align}
with $p_{i,+} = {c^{i,+}}/{(\sum_{i=1}^m c^{i,+})}$, $p_{i,-} ={c^{i,-}}/{(\sum_{i=1}^m c^{i,-})}$ and
\begin{align}
  t(X):= \sum_{i=1}^m (c^{i,-} + c^{i,+}) > \|X\|^{\circ}_{s,\mathrm{inj}} .\label{eq:B23}
\end{align}
The condition $\|X\|^{\circ}_{s,\mathrm{inj}} = +\infty, \forall X \notin \mathrm{V}_k$ means that only when $X \in \mathrm{V}_k$, the decomposition of $X$ into a signed difference of $k$-replica moments of two mixed state ensembles is possible.

Based on this understanding of the dual norm $\|X\|_{s,\mathrm{inj}}^\circ$, the dual problem of Eq.\eqref{eq:primal-robustness} is given below.
\begin{lemma}[Dual problem of the robustness problem \eqref{eq:primal-robustness}]\label{lem:dual-robustness}
    Given that the trivial $k$-replica POVM $\{\1^{\otimes k}\}$ is in the interior of the POVM restriction $\mathcal{C}$, the primal problem \eqref{eq:primal-robustness} satisfies the strong duality condition and its dual problem is given in conic convex program form:
    \begin{align}
         {
        \begin{aligned}
        s_{k,b}^*(\cF, \mathcal{C}) = \max_{Z, \{Y_\alpha\}} \quad &  \sum_{\alpha \in \cA} \tr(Y_\alpha F_\alpha) -b \sum_{\alpha}\|Y_{\alpha}\|_{s,\mathrm{inj}}^\circ  \\
        \text{s.t.} \quad & Z - \sum_{\alpha \in \cA} c_\alpha Y_\alpha \in \mathcal{K}_{\mathcal{C}}^*, \quad \forall\, c \in \{-1, 1\}^{|\cA|} \\
        & \tr Z = 1,\\
        & Z \in \mathrm{Herm}(\cH^{\ox k}),\quad Y_\alpha \in \mathrm{V}_k,
        \end{aligned}
        }\label{eq:dual-robustness}
\end{align}
 where $\mathcal{K}_{\mathcal{C}}^*$ is the dual cone of the cone $\mathcal{K}_{\mathcal{C}}:=\{ \lambda \Pi \mid \lambda \ge 0, \Pi \in \mathcal{C} \}$.
 \end{lemma}
 \begin{proof}
 Defining $c = \{c_{\alpha}\}_{ \alpha \in \cA}  \in \{-1,1\}^{\cA}$ and
 \begin{align}
    \lambda_j(c) =2^{-|\cA|} \prod_{\alpha \in \cA} \left(1 + r^{-1}c_{\alpha} h_{\alpha}(j)\right).\label{eq:B25}
\end{align}
 Using the fact that $h_{\alpha}(j) \in [-r,r]$ for all $\alpha \in \cA$ and $j \in \cJ$, we have
 \begin{align}
   1 \geq \lambda_j(c) \geq 0, \quad \sum_{c \in \{-1,1\}^{|\cA|}} \lambda_j(c) = 1, \sum_{c \in \{-1,1\}^{|\cA|}} c_{\alpha}\lambda_j(c) = r^{-1}h_{\alpha}(j).\label{eq:B26}
\end{align}
 Then we introduce the coarse-grained operators $ M_c := \sum_{j} \lambda_j(c)\, \Pi_j \in \mathcal{K}_{\mathcal{C}},c \in \{-1,1\}^{|\cA|}$, where $\mathcal{K}_{\mathcal{C}} = \{ \lambda \Pi \mid \lambda \ge 0, \Pi \in \mathcal{C} \}$ is the conic hull of $\mathcal{C}$.
 This gives the identity $\sum_c  c_{\alpha} M_c = r^{-1}\sum_{j} h_{\alpha}(j) \Pi_j = r^{-1} W_{\alpha}$ and $\sum_c  M_c = \1^{\otimes k}$.
Further introducing variables $N_c$ and $B_\alpha$ as:
\begin{align}
    &N_c := rM_c \in \mathcal{K}_{\mathcal{C}}, \quad \forall c \in \{-1,1\}^{|\cA|},\label{eq:B27}
\\
    & B_\alpha := \sum_c c_\alpha N_c -  F_{\alpha}, \quad \forall \alpha \in \cA,\label{eq:B28}
\end{align}
Eq.\eqref{eq:primal-robustness} reduces to the following convex program:
\begin{align}
     {
    \begin{aligned}
        s_{k,b}^*(\cF,\mathcal{C}) = \min_{B_\alpha, N_c} \quad & r \\
    \text{Subject to:}\quad & \sum_c N_c - r\1^{\ox k} = 0, \\
    & F_\alpha + B_\alpha - \sum_c c_\alpha N_c = 0, \quad \forall \alpha \in \cA,\\
    & N_c \in \mathcal{K}_{\mathcal{C}}, \quad \forall c \in \{-1,1\}^{|\cA|},\\
    &\|B_\alpha\|_{s,\mathrm{inj}} \leq b, \quad \forall \alpha \in \cA.
    \end{aligned}
    }\label{eq:primal2}
\end{align}
To derive the dual problem, we first introduce the Lagrange multipliers $Z \in \mathrm{Herm}(\cH^{\ox k})$ for $\sum_c N_c - r\1^{\ox k} = 0$, $Y_\alpha \in \mathrm{Herm}(\cH^{\ox k})$ for $F_\alpha + B_\alpha - \sum_c c_\alpha N_c = 0$, the Lagrangian is:
\begin{align}
    \cL(Z, Y_\alpha, N_c, B_\alpha, r) &= r + \tr\left[Z\left(\sum_c N_c - r\1^{\ox k}\right)\right] + \sum_\alpha \tr\left[Y_\alpha\left(F_\alpha + B_\alpha - \sum_c c_\alpha N_c\right)\right]\label{eq:B30}
\\
    &=  \sum_{\alpha}\tr(Y_\alpha F_\alpha) +\sum_{\alpha}\tr(Y_\alpha B_\alpha) + \sum_{c} \tr\left[(Z -\sum_{\alpha} c_\alpha Y_\alpha) N_c\right] + r(1 -\tr Z).\label{eq:B31}
\end{align}
Here $\inf_{N_c, B_\alpha, r}\sup_{Z, Y_\alpha} \cL(Z, Y_\alpha, N_c, B_\alpha, r)$ indicates the minimization result of Eq.\eqref{eq:primal2}. Its dual problem can be generated through exchanging the sup and inf in the Lagrangian, which gives:
\begin{align}
  &g(Z, Y_\alpha) := \inf_{\substack{r, N_c \in \mathcal{K}_{\mathcal{C}}, \\ B_\alpha,\|B_\alpha\|_{s,\mathrm{inj}} \leq b}} \cL(Z, Y_\alpha, N_c, B_\alpha, r) \notag
\\
    &= \sum_{\alpha}\tr(Y_\alpha F_\alpha)  + \inf_{\substack{B_\alpha, \\ \|B_\alpha\|_{s,\mathrm{inj}} \leq b}}\sum_{\alpha}\tr(Y_\alpha B_\alpha) + \sum_{c} \inf_{N_c \in \mathcal{K}_{\mathcal{C}}} \tr\left[(Z -\sum_{\alpha} c_\alpha Y_\alpha) N_c\right] + \inf_{ r} \left[r(1 -\tr Z)\right].\label{eq:dual-temp-eq}
\end{align}
Using the fact that $\mathcal{K}_{\mathcal{C}}$ is a cone and $\|\cdot\|_{s,\mathrm{inj}}$ is a semi-norm on $\mathrm{Herm}(\cH^{\ox k})$ with its dual norm specified in Lemma~\ref{lem:dual-atomic}, we have
\begin{align}
    \inf_{N_c \in \mathcal{K}_{\mathcal{C}}} \tr\left[X N_c\right]  = \begin{cases}
        0, & \text{if } X \in \mathcal{K}^*_{\mathcal{C}}, \\
        -\infty, & \text{otherwise}.
    \end{cases},\quad  \inf_{B_\alpha,\|B_\alpha\|_{s,\mathrm{inj}} \leq b}\tr(X B_\alpha) = - b\|X\|_{s,\mathrm{inj}}^\circ\label{eq:B33}
\end{align}
where $\mathcal{K}_{\mathcal{C}}^*$ is the dual cone of $\mathcal{K}_{\mathcal{C}}$, defined as $\mathcal{K}_{\mathcal{C}}^* := \left\{ W \in \mathrm{Herm}(\cH^{\ox k}) \;\Big|\; \tr(WM) \ge 0, \;\forall M \in \mathcal{K}_{\mathcal{C}} \right\}$.
 Substituting into Eq.\eqref{eq:primal2}, we have $g(Z, Y_\alpha) = -\infty$ if $Z - \sum_{\alpha} c_\alpha Y_\alpha \notin \mathcal{K}^*_{\mathcal{C}}$ or $\tr Z \neq 1$, otherwise
 \begin{align}
    g(Z, Y_\alpha)
    &= \sum_{\alpha}\tr(Y_\alpha F_\alpha) - b\sum_{\alpha}\|Y_{\alpha}\|_{s,\mathrm{inj}}^\circ.\label{eq:B34}
\end{align}
Based on Lemma~\ref{lem:dual-atomic}, to keep the dual norm $\|Y_{\alpha}\|_{s,\mathrm{inj}}^\circ$ finite rather than $+\infty$, we restrict $Y_\alpha \in \mathrm{V}_k$. Then the dual problem of Eq.\eqref{eq:primal2}, given by the supremum of $g(Z, Y_\alpha)$, is given by:
\begin{align}
    d_{k,b}^*(\cF) := \max_{Z, \{Y_\alpha\}} \quad & \sum_{\alpha \in \cA} \text{Tr}(Y_\alpha F_\alpha) - b\sum_{\alpha \in \cA}\|Y_{\alpha}\|_{s,\mathrm{inj}}^\circ\label{eq:B35}
\\
    \text{s.t.} \quad & Z - \sum_{\alpha \in \cA} c_\alpha Y_\alpha \in \mathcal{K}_{\mathcal{C}}^*, \quad \forall\, c \in \{-1, 1\}^{|\cA|}\label{eq:B36}
\\
    & \tr Z = 1,\quad  Y_\alpha \in \mathrm{V}_k, Z \in \mathrm{Herm}(\cH^{\ox k}),\label{eq:B37}
\end{align}
To verify that the problem in Eq.\eqref{eq:primal2} satisfies the strong duality condition $d_{k,b}^*(\cF) = s_{k,b}^*(\cF)$, we show that Slater's condition is satisfied.
Using the fact that $\1^{\otimes k} \in \mathrm{int}(\mathcal{K}_{\mathcal{C}})$, we take:
\begin{align}
    N_c = 2^{-|\cA|} \left(r \1^{\otimes k} + \sum_{ \alpha \in \cA} c_{\alpha} F_{\alpha}\right), \quad B_{\alpha} = 0,\label{eq:B38}
\end{align}
such that for $r$ sufficiently large, $N_c \in \mathrm{int}(\mathcal{K}_{\mathcal{C}})$ is always satisfied and $\|B_{\alpha}\|_{s, \mathrm{inj}} = 0 < b$.
When $b=0$, using the fact that $F_{\alpha}$ is the unique symmetric $k$-replica linearization and $\Pi_j$ is also symmetric (see Theorem~\ref{thm:structure_fcomp}),
 i.e., $F_{\alpha}, N_c \in \mathrm{V}_k$,
the condition $\|B_\alpha\|_{s,\mathrm{inj}}\le0$ forces $B_\alpha=0$ within $\mathrm{V}_k$.
The reduced exact program then satisfies Slater's condition by the same
choice of $N_c$ above, with $r$ sufficiently large.
 \end{proof}

We note that the $\mathrm{QB}_k$-norm problem is recovered from this robustness program by restricting the admissible bias family to $\cG=\{0\}$, equivalently, by imposing $b=0$ in Eq.~\eqref{eq:primal2}. Accordingly, the resource-restricted $\mathrm{QB}_k$-norm is defined as $\|\cF\|_{k,\mathcal{C}}:=s_{k,0}^{*}(\cF,\mathcal{C})$; note that $s_{k,b}^{*}(\cF,\mathcal{C})$ is genuinely a norm only at $b=0$. In this case the resulting dual problem is
\begin{align}
 {
\begin{aligned}
    \|\cF\|_{k,\mathcal{C}} = s_{k,0}^{*}(\cF,\mathcal{C}) =\max_{Z,\{Y_\alpha\}}\quad &
\sum_{\alpha\in\cA}\tr(Y_\alpha F_\alpha) \\
\text{s.t.}\quad &
Z-\sum_{\alpha\in\cA} c_\alpha Y_\alpha \in \mathcal K_{\mathcal C}^*,
\quad \forall c\in\{-1,1\}^{|\cA|},\\
& \tr Z = 1,\quad  Y_\alpha, Z\in \mathrm{Herm}(\cH^{\otimes k}).
\end{aligned}}\label{eq:dual_unbiased}
\end{align}
Note that the effect of $b = 0$ still leaves a restriction $Y_\alpha \in \mathrm{V}_k$.
However, if we use the fact that $F_{\alpha}$ is the unique symmetric $k$-replica linearization, i.e., $F_{\alpha} \in \mathrm{V}_k$,
we conclude that the restriction $Y_\alpha \in \mathrm{V}_k$ in the dual problem Eq.\eqref{eq:dual-robustness} can be relaxed to the $Y_{\alpha} \in \mathrm{Herm}(\mathcal{H}^{\otimes k})$ in Eq.\eqref{eq:dual_unbiased} without changing the optimal value.
Here below, we will focus on the case $b = 0$ only.

The dual problem in Eq.\eqref{eq:dual-robustness} induces the restriction on the specific dual cones $\mathcal{K}_{\mathcal{C}}^* :=
\{W \in \mathrm{Herm}(\mathcal{H}_0^{\otimes k})
\mid \tr(W\,M) \geq 0,\;\forall\, M \in
\mathcal{C}\}$.
Taking the memory-free restriction in Eq.\eqref{eq:block-positive-cone-mem} and site-product restriction in Eq.\eqref{eq:site-product-cone-loc} as examples, the dual cones $\mathcal{K}_{\mathcal{C}}^*$ are respectively:
\begin{align}
    &\mathcal{K}_{\mathrm{mem}}^{*}
    = \left\{ W \;\Big|\;
    \bra{\Phi_1}\!\cdots\!\bra{\Phi_{k}}
    W
    \ket{\Phi_1}\!\cdots\!\ket{\Phi_{k}}
    \geq 0,\; \forall\, \ket{\Phi_{\ell}} \in
    (\mathbb{C}^2)^{\otimes n}\right\},\label{eq:block-positive-cone}
\\
    &\mathcal{K}_{\mathrm{loc}}^{*}
    = \left\{ W \;\Big|\;
    \bra{\phi_1}\!\cdots\!\bra{\phi_n} W
    \ket{\phi_1}\!\cdots\!\ket{\phi_n} \geq 0,\;
    \forall\, \ket{\phi_l} \in
    (\mathbb{C}^2)^{\otimes k} \right\}.\label{eq:site-product-cone}
\end{align}

\paragraph{Primal recovery via KKT conditions.}
Since the problem Eq.\eqref{eq:primal2} and its dual problem Eq.\eqref{eq:dual-robustness} are standard conic convex optimizations, by strong duality and the KKT conditions~\cite{boyd2004convex} for conic convex optimization, there exist an optimal primal solution $r^*=s_{k,b}^*(\cF,\mathcal{C}),\{N_c^*\},\{B_\alpha^*\}$ and dual solutions $Y^*_{\alpha}, Z^*$ satisfying all the dual and primal feasibility conditions as well as
\begin{align}
&\text{Cone complementary slackness:}\quad
\tr(N_c^*\Delta_c^*)=0,\quad \forall c,\label{eq:cone-cs}
\\
&\text{Norm-ball complementary slackness:}\quad
-\tr(Y_\alpha^*B_\alpha^*)=b\|Y_\alpha^*\|_{s,\mathrm{inj}}^\circ,\quad \forall \alpha.\label{eq:norm-ball-cs}
\end{align}
Hence the optimal dual variables $Z^*, \{Y^*_{\alpha}\}_{\alpha}$ first determine an optimal residual family $\{B_\alpha^*\}$ through Eq.\eqref{eq:norm-ball-cs}, and then determine the optimal conic variables $\{N_c^*\}$ through Eq.\eqref{eq:cone-cs}. The recovered optimal strategy is
\begin{align}
\Pi_c^*:=\frac{N_c^*}{r^*},\qquad h_\alpha^*(c):=r^* c_\alpha,  \qquad \sum_c h_\alpha^*(c)\Pi_c^*={F_\alpha+B_\alpha^*}.\label{eq:POVM_recovery}
\end{align}

\appsection{Further properties of the quantum Bernstein norm}\label{app:further-properties}

\appsubsection{The operational meaning of the dual problem}
\paragraph{The shift-invariant relaxation.}
The raw robustness $s_{k,b}^*(\cF,\mathcal{C})$ indicates the existence of a $b$-bias estimator $|h_{\alpha}(j)| \leq r$ averaged over the measurement results $j$.
However, the current optimal $h_{\alpha}(j)$ still contains an unnecessary constant that might increase the value $|h_{\alpha}(j)|$.
For example, although the estimation of constant properties $\cF := \{\tr\rho = 1\}$ requires no information about $\rho$, any POVM strategy still requires estimators $h(j) \equiv 1$.
Operationally,  without changing the POVM or the number of state copies,
any estimation strategy for $\cF:=\{f_\alpha(\rho)\}_{\alpha\in\cA}$
can be converted into one for the shifted family $\cF+\mathbf t := \{f_\alpha(\rho)+t_\alpha\}_{\alpha\in\cA}$
by just adding the state-independent constants $\{t_\alpha\}_{\alpha\in\cA}$ to each final result.
This motivates introducing a shift-invariant version of the robustness by quotienting out these label-dependent additive offsets:
\begin{align}
    s_{k,b}^{*,\mathrm{shift}}(\cF,\mathcal{C})
    :=
    \inf_{\{t_\alpha\in\mathbb R\}}
    s_{k,b}^*(\{f_\alpha+t_\alpha\}_{\alpha\in\cA}, \mathcal{C}).\label{eq:B45}
\end{align}
On the primal side in Eq.\eqref{eq:primal2}, this is equivalent to replacing the matching constraints by
\begin{align}
    F_\alpha + B_\alpha + t_\alpha \1^{\otimes k} - \sum_c c_\alpha N_c = 0,
    \qquad \forall \alpha\in\cA,\label{eq:B46}
\end{align}
with additional optimization variables $\{t_\alpha\}_{\alpha\in\cA}$.
Accordingly, the original Lagrangian in Eq.\eqref{eq:dual-temp-eq} acquires the extra term $ \sum_{\alpha\in\cA} t_\alpha \tr(Y_\alpha)$.
Since each $t_\alpha$ is unrestricted, finiteness of the dual function gives the shift-invariant restriction on the original dual problem in Eq.\eqref{eq:dual-robustness}:
\begin{align}
    \tr Y_\alpha=0,\qquad \forall \alpha\in\cA.\label{eq:B47}
\end{align}
Adding this restriction to the dual problem will give $s^{*, \mathrm{shift}}_{k,b}(\cF,\mathcal{C}) \leq s^*_{k,b}(\cF,\mathcal{C})$.
This indeed corresponds to a shift assisted estimation strategy for $\cF$ with estimators bounded by $s_{k,b}^*(\cF,\mathcal{C})$.

\paragraph{The state discrimination lower bound.}
The dual variable $(Z^*, \{Y_\alpha^*\})$ of the dual problem~\eqref{eq:dual-robustness} can be interpreted as a family of dual-feasible binary ensemble discrimination tasks, indexed by $\alpha \sim \lambda$, whose average distinguishability is at most $1/\Lambda$, while their average function gap certifies a lower bound on $s_{k,b}^{*,\mathrm{shift}}(\cF,\mathcal{C})$.
Using Lemma~\ref{lem:dual-atomic}, for any $Y_\alpha\in \mathrm V_k$, the finiteness of $\|Y_\alpha\|_{s,\mathrm{inj}}^\circ$ is equivalent to the existence of a decomposition $\mathcal{E}_{\alpha, +} :=(p_{\alpha,+,i},\rho_{\alpha,i,+})$ and $\mathcal{E}_{\alpha, -} :=(p_{\alpha,-,i},\rho_{\alpha,i,-})$ such that
\begin{align}
    Y_{\alpha} &= \frac{\|Y_{\alpha}\|_{s,\mathrm{inj}}^\circ + \tr(Y_{\alpha})}{2} \sum_{i=1}^m p_{\alpha,+,i} \rho_{\alpha,i,+}^{\otimes k} - \frac{\|Y_{\alpha}\|_{s,\mathrm{inj}}^\circ - \tr(Y_{\alpha})}{2} \sum_{i=1}^m p_{\alpha,-,i} \rho_{\alpha,i,-}^{\otimes k} \notag
\\
    &:= \frac{\Lambda}{2} \cdot \lambda_{\alpha} \left[\Omega_{\alpha,+} - \Omega_{\alpha,-}\right].\label{eq:B48}
\end{align}
Here we use the shift-invariant restriction $\tr Y_{\alpha} = 0$ and set
\begin{align}
    \Lambda = \sum_{\alpha}\|Y_\alpha\|_{s,\mathrm{inj}}^\circ,  \quad \lambda_{\alpha} = \frac{\|Y_{\alpha}\|_{s,\mathrm{inj}}^\circ }{\Lambda}, \quad \Omega_{\alpha,\pm}:=\mathbb E_{\rho\sim\mathcal E_{\alpha,\pm}}[\rho^{\otimes k}].\label{eq:normalized_prob}
\end{align}
Note that $\Omega_{\alpha,\pm}$ are the $k$-replica moments of the mixed state ensembles $\mathcal{E}_{\alpha,\pm}$.
Such a moment difference also appears in the classical best polynomial approximation~\cite{wu2020polynomial} and further yields an optimal lower bound based on Le Cam's two-point method.
As a generalization, we show that this also induces a tight operational lower bound for the $\mathrm{QB}_k$-norm of the set of functions $\cF$.
By Lemma~\ref{lem:dual-atomic}, we also conclude that $\|Y_{\alpha}\|_{s,\mathrm{inj}}^\circ \geq 0$, thus giving $\lambda_{\alpha}\geq 0$ and $\sum_{\alpha} \lambda_{\alpha} = 1$, which define a valid probability distribution on $\cA$.
Plugging these equations into the dual problem~\eqref{eq:dual-robustness} to replace all $Y_{\alpha}$ variables transforms the dual problem into
\begin{align}
    s^{*,\mathrm{shift}}_{k,b}(\cF) &= \max_{\{\mathcal{E}_{\alpha, \pm}\},\{ \lambda_{\alpha}\}, \Lambda, Z} \frac{\Lambda}{2} \cdot \sum_{\alpha \in \cA} \lambda_{\alpha} \left[\mathbb{E}_{\rho \sim \mathcal{E}_{\alpha,+}}f_{\alpha}(\rho) -  \mathbb{E}_{\rho \sim \mathcal{E}_{\alpha,-}}f_{\alpha}(\rho)\right]  - b \Lambda\label{eq:discrimination-game-mid-obj}
\\
\text{s.t.} \quad & Z - \frac{\Lambda}{2}\sum_{\alpha \in \cA} \lambda_{\alpha}c_\alpha\left(\Omega_{\alpha,+} - \Omega_{\alpha,-}\right) \in \mathcal{K}_{\mathcal{C}}^*, \quad \forall\, c \in \{-1, 1\}^{|\cA|}\label{eq:mid-op-bound}
\\
& \tr Z = 1,\quad Z \in \mathrm{Herm}(\cH^{\ox k}).\label{eq:mid-op-bound2}
\end{align}
Here we use the fact that $F_{\alpha}$ is a $k$-linearization of $f_{\alpha}$, such that $\tr(F_{\alpha} \rho^{\otimes k}) = f_{\alpha}(\rho)$.
We then show that Eqs.\eqref{eq:discrimination-game-mid-obj}-\eqref{eq:mid-op-bound2} admit physical interpretations.

The maximization objective in Eq.\eqref{eq:discrimination-game-mid-obj} contains the average function value gap between the two mixed state ensembles $\mathcal{E}_{\alpha,+}$ and $\mathcal{E}_{\alpha,-}$, given that the index $\alpha$ is selected with probability $\lambda_{\alpha}$:
\begin{align}
    \Delta_\cF(\mathbf{W}) := \sum_{\alpha \in \cA} \lambda_\alpha \left| \mathbb{E}_{\rho \sim \mathcal{E}_{\alpha, +}} f_\alpha(\rho) - \mathbb{E}_{\rho \sim \mathcal{E}_{\alpha, -}} f_\alpha(\rho) \right|,\label{eq:B53}
\end{align}
where $\mathbf{W} = (\mathcal{E}_{\alpha,+}, \mathcal{E}_{\alpha,-}, \lambda_{\alpha})$ denotes the mixed-state ensemble pairs for $\alpha \in \mathcal{A}$ sampled from the distribution $\lambda$.
Taking an allowed POVM $\mathsf M=\{M_j\}_{j \in \cJ}$ with $M_j\in \mathcal C$ on $k$-replica state moments, we also get the classical distribution $\{q_{\alpha,\pm}^{\mathsf M}(j) :=\tr(M_j\Omega_{\alpha,\pm}) \}_{j \in \cJ}$.
In terms of these post-measurement distributions $\{q_{\alpha,\pm}^{\mathsf M}(j)\}_{j \in \cJ, \alpha \in \cA}$,
Eq.\eqref{eq:mid-op-bound} translates to the inequality
\begin{align}
&\forall c\in\{-1,1\}^{|\cA|},\forall j, \forall M_j \in \mathcal{C},\quad \frac{\Lambda}{2}\sum_{\alpha\in\cA}c_\alpha  \lambda_{\alpha} \left[q_{\alpha,+}^{\mathsf M}(j)-q_{\alpha,-}^{\mathsf M}(j)\right]
\le \tr (Z M_j),\label{eq:B54}
\\
&\Leftrightarrow \forall j,\quad \sum_{\alpha\in\cA}\lambda_{\alpha} \frac{\left|q_{\alpha,+}^{\mathsf M}(j)-q_{\alpha,-}^{\mathsf M}(j)\right|}{2}
\le \frac{ \tr (Z M_j) }{\Lambda} \Rightarrow \mathbb{E}_{\alpha \sim \{\lambda_{\alpha}\}} \|q_{\alpha, +}^{\mathsf M}- q_{\alpha, -}^{\mathsf M}\|_{\mathrm {TV}} \leq \frac{1}{\Lambda},\label{eq:tvd_equiv}
\end{align}
where the total variation distance is defined as $\|p - q\|_{\mathrm {TV}}:=\frac{1}{2}\sum_{i = 1}^n|p(i) - q(i)|$.
This inequality admits an interpretation in terms of distinguishability of the post-measurement distribution.
\begin{lemma}[Le Cam's two-point method~\cite{lecam1973convergence}]
    \label{lemma:TV_bound}
    Given two distributions $p,q$ on the set $S = \{1,\cdots,n\}$,
    the maximum success rate to distinguish $p,q$ satisfies:
    \begin{align}
        \mathrm{Pr}_s(p,q) = \frac{1}{2} + \frac{1}{2} \|p - q\|_{\mathrm {TV}} .\label{eq:B56}
\end{align}
    \end{lemma}
    \begin{proof}
        We first show that when $\|p - q\|_{\mathrm {TV}} = \varepsilon$, there is no single sample algorithm $\mathcal{A}: S \to \{p,q\}$ that is $\varepsilon/2$ better than random guessing, in other words, satisfies $\text{Pr}_{x \sim p}[\mathcal{A}(x) = p] \geq 1/2 + \varepsilon/2$ and $ \text{Pr}_{x \sim q}[\mathcal{A}(x) = q] \geq 1/2 + \varepsilon/2$.
        It is sufficient to show that for an arbitrary algorithm $\mathcal{A}$, either $\text{Pr}_{x \sim p}[\mathcal{A}(x) = q] \geq (1-\varepsilon)/2$ or $\text{Pr}_{x \sim q}[\mathcal{A}(x) = p] \geq  (1-\varepsilon)/2.$
        Notice that
        \begin{align}
            \text{Pr}_{x \sim p}[\mathcal{A}(x) = q] + \text{Pr}_{x \sim q}[\mathcal{A}(x) = p] & = 1 - \text{Pr}_{x \sim p}[\mathcal{A}(x) = p] +\text{Pr}_{x \sim q}[\mathcal{A}(x) = p]\label{eq:B57}
\\
            &\geq 1 - \sup_{U \subseteq S} \left |\text{Pr}_{x \sim p}[U] -\text{Pr}_{x \sim q}[U]\right |\label{eq:B58}
\\
            &= 1 - \|p - q\|_{\mathrm{TV}} = 1 - \varepsilon.\label{eq:B59}
\end{align}
        Then either $\text{Pr}_{x \sim p}[\mathcal{A}(x) = q] \geq   (1-\varepsilon)/2$ or $\text{Pr}_{x \sim q}[\mathcal{A}(x) = p] \geq  (1-\varepsilon)/2$, which means that either $\text{Pr}_{x \sim p}[\mathcal{A}(x) = p] \leq   (1+\varepsilon)/2$ or $\text{Pr}_{x \sim q}[\mathcal{A}(x) = p] \leq  (1+\varepsilon)/2$. One can verify that this algorithm is reachable by the likelihood-ratio test $\mathcal{A}^*(x) = p, \text{~if~} p(x)\geq q(x)$ and $\mathcal{A}^*(x) = q, \text{if~} p(x) < q(x)$.
    \end{proof}
Based on Lemma~\ref{lemma:TV_bound}, Eq.\eqref{eq:tvd_equiv} implies that,
given an index $\alpha$ selected with probability $\lambda_{\alpha}$ and $k$ replicas of the state $\rho$ prepared with either $\rho \sim \mathcal{E}_{\alpha,+}$ or $\rho \sim \mathcal{E}_{\alpha,-}$, any allowed $k$-replica POVM $\mathsf M=\{M_j\}_{j \in \cJ}$ with $M_j\in {\mathcal C}$ cannot distinguish the cases indexed by $\{+,-\}$ with probability better than ${1}/{2} + {1}/{(2\Lambda)}$. We thus define:
\begin{align}
    \mathrm{Pr}_s(\mathbf{W}) :=\frac{1}{2} + \frac{1}{2} \max_{M_j \in \mathcal{C}} \mathbb{E}_{\alpha \sim \{\lambda_{\alpha}\}} \|q_{\alpha, +}^{\mathsf M}- q_{\alpha, -}^{\mathsf M}\|_{\mathrm {TV}} \leq \frac{1}{2} + \frac{1}{2\Lambda},\label{eq:B60}
\end{align}
With this interpretation in mind, we further show an operational and attainable lower bound for $s_{k,b}^*(\mathcal{F},\mathcal{C})$.

\begin{lemma}[Functional gap--distinguishability ratio bound]\label{lem:function-gap-distinguishability-ratio}
    Given any mixed state ensemble pairs $(\mathcal{E}_{\alpha,+}, \mathcal{E}_{\alpha, -})$ for $\alpha \in \mathcal{A}$ sampled from the distribution $\lambda$.
    Then:
    \begin{align}
        s_{k,b}^*(\mathcal{F},\mathcal{C}) \geq \frac{\Delta_\cF(\mathbf{W}) - 2b}{4\cdot \left[\mathrm{Pr}_s(\mathbf{W})  - 1/2\right]}.\label{eq:B61}
\end{align}
\end{lemma}
\begin{proof}
    Take any $k$-replica joint estimator $\{h_{\alpha}(j), \Pi_j\}_{j \in \mathcal{J}}$ of $\mathcal{F}$ with bias bound $b$ and scaling $s = \max_{\alpha, j} |h_{\alpha}(j)|$, we conclude that:
    \begin{align}
        \Delta_\cF(\mathbf{W}) - 2b &\leq \sum_{\alpha,j} \lambda_{\alpha} h_{\alpha}(j)  \left|q^{M}_{\alpha,+}(j) - q^{M}_{\alpha,-}(j)\right| \notag
\\
        &\leq 2s \cdot \mathbb{E}_{\alpha} \|q^{M}_{\alpha,+}(j) - q^{M}_{\alpha,-}(j)\|_{\mathrm{TV}} \leq 2s \cdot \left[2\cdot \mathrm{Pr}_s(\mathbf{W}) - 1\right],\label{eq:B62}
\end{align}
where $\{q_{\alpha,\pm}^{\mathsf M}(j) :=\tr(\Pi_j\Omega_{\alpha,\pm}) \}_{j \in \cJ}$ are distributions generated by sampling $k$ replicas from ensembles using POVMs $\{\Pi_j\}_{j \in \mathcal{J}}$. This then gives the desired lower bound for $s$. Since this lower bound does not depend on the choice of strategy, it is valid for any $k$-replica joint estimator and also holds for the optimal value $s^*_{k,b}(\mathcal{F},\mathcal{C})$.
\end{proof}

Lemma~\ref{lem:function-gap-distinguishability-ratio} gives an operational method to witness the incompatibility of a set of properties $\mathcal{F}$ by measuring the function gap and the distinguishability of the post-measurement distributions.
To be specific, one can define the discrimination game in Fig.~\ref{fig:alice-bob-discrimination-witness} to witness the incompatibility of a set of properties $\mathcal{F}$. In this game:
\begin{enumerate}
\item Alice chooses a signal $\{+,-\}$ and prepares the mixed state ensembles $\mathcal{E}_{\alpha,\pm}$ based on the index $\alpha \in \mathcal{A}$ sampled from the distribution $\lambda_{\alpha} = p_{\alpha}$ and the signal she chooses.
\item Alice sends $k$ replicas of the mixed state to Bob. Bob is required to use an allowed $k$-replica POVM $\mathcal{C}$ to distinguish the signal $\{+,-\}$, with a maximum success probability $\mathrm{Pr}_s(\mathbf{W})$. 
\item Alice implements a certain experiment (not necessarily optimal or required to be jointly measurable) to get the function gap $\Delta_\cF(\mathbf{W})$.  
\item Use Lemma~\ref{lem:function-gap-distinguishability-ratio} to get the lower bound for $s_{k,b}^*(\mathcal{F},\mathcal{C})$.
\end{enumerate}

\begin{figure}[htbp!]
    \centering
    \includegraphics[width=\textwidth]{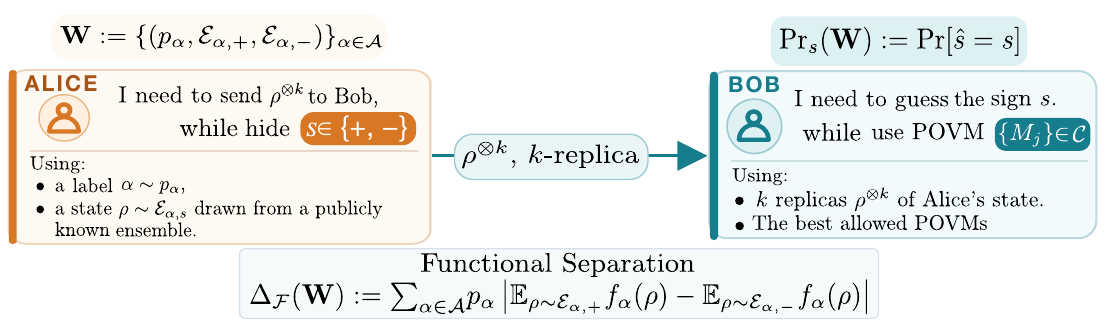}
    \caption{Alice and Bob play a discrimination game to witness the incompatibility of a set of properties $\mathcal{F}$.}
    \label{fig:alice-bob-discrimination-witness}
\end{figure}

Using Eq.\eqref{eq:discrimination-game-mid-obj}, we further show that this desired lower bound for $s_{k,b}^*(\mathcal{F},\mathcal{C})$ is attainable and tight.

\begin{lemma}[Tightness of the function gap--distinguishability ratio bound]\label{lem:reachability-function-gap-distinguishability-ratio}
    Let $\cF=\{f_\alpha\}_{\alpha\in\cA}$ be a family of polynomial functions, and let $\mathcal C$ be a class of allowed POVMs.
    Then there exist a probability distribution $\lambda$ on $\cA$ and a family of ensemble pairs
    $\{\mathcal E_{\alpha,+},\mathcal E_{\alpha,-}\}_{\alpha\in\cA}$ such that no allowed POVM in $\mathcal C$ can distinguish the binary label $\{+,-\}$ with average success probability larger than $p^*$, while for small enough $b$, the corresponding average function gap satisfies
    \begin{align}
        s_{k,b}^{*, \mathrm{shift}}(\mathcal{F},\mathcal{C})=\frac{\Delta_\cF(\mathbf{W}) - 2b}{4\cdot \left(p^*- 1/2\right)} \leq \frac{\Delta_\cF(\mathbf{W}) - 2b}{4\cdot \left[\mathrm{Pr}_s(\mathbf{W})  - 1/2\right]} .\label{eq:reachability-function-gap-distinguishability-ratio}
\end{align}
\end{lemma}

\begin{proof}
Given that $s^* = s_{k,b}^{*, \text{shift}}(\cF, \mathcal{C})$, the optimal solution $Y^*_\alpha, Z^*$
exists and admits an interpretation in terms of density matrix ensembles $\{\mathcal{E}_{\alpha, \pm}\}_{\alpha \in \cA}, \lambda_{\alpha}$ such that no allowed POVM can
distinguish the index $\{+,-\}$ with probability better than $p^* = 1/2 + 1/(2\Lambda)$ while $s^* = ({\Lambda}/{2})\Delta_\cF(\mathbf{W}) - b \Lambda$.  
This then gives a valid ensemble satisfying $\Delta_\cF(\mathbf{W}) = 2(2p^* - 1)s^* + 2b$.
Note that $\Delta_\cF(\mathbf{W}) -2b>0$ must be satisfied for small enough $b$.
Then the upper bound in Eq.\eqref{eq:reachability-function-gap-distinguishability-ratio} is proved by using the fact $p^* \geq \mathrm{Pr}_s(\mathbf{W}) > 1/2$.
\end{proof}
To indicate the tightness, we note that $\Delta_\cF(\mathbf{W}) = \Delta_{\cF+ \mathbf{t}}(\mathbf{W}), \forall \mathbf{t} \in \mathbb{R}^{|\cA|}$, thus the lower bound in Lemma~\ref{lem:function-gap-distinguishability-ratio} also applies to $s_{k,b}^{*, \text{shift}}(\cF, \mathcal{C})$.

\appsubsection{Subadditivity of $\mathrm{QB}_k$-norm}
In this section, we show that the $\mathrm{QB}_k$-norm $\|\mathcal{F}\|_k:= s_k^*(\cF)$ (Definition~\ref{def:k-copy-qBern_coeff}) is a well-behaved norm by regarding the pointwise addition and scalar multiplication as the operations on the property set.
As a starting point, we first specify the vector space that $\mathcal{F}$ belongs to:
\begin{align}
    &\mathcal{V}_{k,\mathcal{A}} := \left\{\cF = \{f_\alpha\}_{\alpha \in \mathcal{A}}
    \; : \; \forall \alpha \in \cA, \, f_{\alpha}(\rho) \text{~is a polynomial of~} \rho \text{~ with degree at most~} k
    \right\}.\label{eq:E1}
\\
    & \text{Addition rule:~} \forall \mathcal{F}:= \{f_{\alpha}\}_{\alpha \in \cA} \in \mathcal{V}_{k,\mathcal{A}} , \mathcal{G} := \{g_{\alpha}\}_{\alpha \in \cA} \in \mathcal{V}_{k,\mathcal{A}},\quad \mathcal{F} + \mathcal{G}:= \{f_{\alpha}(\rho) + g_{\alpha}(\rho)\}_{\alpha \in \cA}.\label{eq:E2}
\\
    & \text{Scalar multiplication rule:~} \forall \mathcal{F}:= \{f_{\alpha}\}_{\alpha \in \cA} \in \mathcal{V}_{k,\mathcal{A}}, \forall c \in \mathbb{R}, \quad c \cdot \mathcal{F}:= \{c \cdot f_{\alpha}(\rho)\}_{\alpha \in \cA}.\label{eq:E3}
\end{align}
The zero vector of $\mathcal{V}_{k,\mathcal{A}}$ is written as $\cF = \{f_{\alpha} := 0\}_{\alpha \in \cA} := \{0\}$.
Then we show that $\|\mathcal{F}\|_k$ is a norm on $\mathcal{V}_{k,\mathcal{A}}$.
\begin{lemma}\label{lem:isnorm}
    The $\mathrm{QB}_k$-norm $\|\mathcal{F}\|_k$ (Definition~\ref{def:k-copy-qBern_coeff}) satisfies:
    (i). positivity $\|\cF\|_k \geq 0$, (ii). homogeneity $\|c\cdot \cF\|_k = |c| \cdot \|\cF\|_k, \forall c \in \mathbb{R}$, and (iii). faithfulness: $\|\cF\|_k = 0 \Leftrightarrow \cF = \{0\}$.
\end{lemma}
\begin{proof}
     The positivity and the homogeneity follow directly from Definition~\ref{def:k-copy-qBern_coeff}. To show the faithfulness, we note that for any valid unbiased estimation strategy:
     \begin{align}
        |f_{\alpha}(\rho)| = \left|\sum_{j \in \cJ} h_{\alpha}(j) \tr(\Pi_j \rho^{\otimes k})\right| \leq \sum_{j \in \cJ} |h_{\alpha}(j)| \tr(\Pi_j \rho^{\otimes k}) \leq \max_{j\in\cJ, \alpha \in \cA} |h_{\alpha}(j)|.\label{eq:E4}
\end{align}
     If $\|\mathcal F\|_k=0$, then $\forall \varepsilon >0$, there exists an unbiased estimator such that $\max_{j\in\cJ, \alpha \in \cA} |h_{\alpha}(j)| \leq \varepsilon$. This directly indicates $|f_{\alpha}(\rho)| \leq \varepsilon, \forall \varepsilon >0 \Rightarrow f_{\alpha}(\rho) = 0, \forall \alpha \in \cA$.
\end{proof}
One critical property left is the triangle inequality, or, the subadditivity of $\mathrm{QB}_k$-norm. We show it below.
\begin{lemma}[Pointwise subadditivity of $\mathrm{QB}_k$-norm]\label{lem:pointwise-subadditivity}
    For two sets of function properties $\cF:= \{f_{\alpha}\}_{\alpha \in \cA}$ and $\cG := \{g_{\alpha}\}_{\alpha \in \cA}$, define the pointwise addition as $\mathcal{W} = \cF + \cG = \{f_{\alpha} + g_{\alpha}\}_{\alpha \in \cA}$. Then
    \begin{align}
        \|\cF + \cG\|_k \leq \|\cF\|_k + \|\cG\|_k.\label{eq:pointwise_subadditivity}
\end{align}
\end{lemma}
\begin{proof}
    Suppose $\|\cF\|_k>0, \|\cG\|_k>0$ without loss of generality.
    Let $\{\Pi_j, h_\alpha(j)\}$ and $\{\Gamma_l, g_\alpha(l)\}$ be the optimal strategies for $\cF$ and $\cG$ with scaling factors $s_1 = \|\cF\|_k$ and $s_2 = \|\cG\|_k$.
    By Definition~\ref{def:k-copy-qBern_coeff}, we have $s_1 := \max_{\alpha,j} |h_{\alpha}(j)| \geq|h_\alpha(j)|$ and $s_2 := \max_{\alpha,l} |g_{\alpha}(l)| \geq |g_\alpha(l)|$ for all $\alpha \in \cA$.
    Define the mixed POVM
    \begin{align}
        \tilde{\Pi}_{(1,j)} = \frac{s_1}{s_1 + s_2}\,\Pi_j, \qquad \tilde{\Pi}_{(2,l)} = \frac{s_2}{s_1 + s_2}\,\Gamma_l,\label{eq:E6}
\end{align}
    which satisfies $\sum_{j} \tilde{\Pi}_{(1,j)}+ \sum_{l} \tilde{\Pi}_{(2,l)} = \1$. Set estimators for the unbiased estimation of $f_{\alpha} + g_{\alpha}$ as:
     \begin{align}
        \tilde{h}_{\alpha}(1,j) = \frac{s_1+s_2}{s_1} h_{\alpha}(j), \quad  \tilde{h}_{\alpha}(2,l) = \frac{s_1+s_2}{s_2} g_{\alpha}(l),\label{eq:E7}
\end{align}
    with $|\tilde{h}_\alpha(1,j)|, |\tilde{h}_\alpha(2,l)| \leq s_1 + s_2, \forall j,l$, given that $s_1 \geq |h_\alpha(j)|$ and $s_2 \geq |g_\beta(l)|$. This rescaling also cancels the POVM dilution, thus yielding an unbiased estimator of $f_{\alpha} + g_{\alpha}, \forall \alpha \in \cA$:
    \begin{align}
        \sum_{j} \tilde{h}_\alpha(1,j)\,\tilde{\Pi}_{(1,j)} + \sum_{l} \tilde{h}_\alpha(2,l)\,\tilde{\Pi}_{(2,l)} &= \sum_j \frac{s_1+s_2}{s_1}\,h_\alpha(j) \cdot \frac{s_1}{s_1+s_2}\,\Pi_j + \sum_l \frac{s_1+s_2}{s_2}\,g_\alpha(l) \cdot \frac{s_2}{s_1+s_2}\,\Gamma_l\label{eq:E8}
\\
         &= \sum_j h_\alpha(j)\,\Pi_j  + \sum_l g_\alpha(l)\,\Gamma_l = F_\alpha + G_{\alpha}.\label{eq:E9}
\end{align}
    Here in the last equality, we use Lemma~\ref{lem:linearization_symmetrization}.
    This completes the proof.
\end{proof}

Besides these ordinary properties, two basic and physically reasonable operations on $\cF$ are subsetting and extending $\cF$ with more functions.

\begin{lemma}[Monotonicity under subsetting and convex combination]\label{lem:mono1}
    Given a set of normalized state properties $\cF := \{f_{\alpha}(\rho)\}_{\alpha \in \cA}$ with linearizations $\{F_\alpha\}_{\alpha \in \cA}$, the $\mathrm{QB}_k$-norm $\|\cF\|_k$ satisfies:
    \begin{enumerate}
        \item Subsetting. If $\cG \subseteq \cF$, then $\|\cG\|_k \leq \|\cF\|_k$.
        \item Convex combination. Replacing any $f_{\alpha_0} \in \cF$ with a convex combination $g(\rho) := \sum_{\alpha \in \cA} p(\alpha)\, f_{\alpha}(\rho)$, where $p(\alpha) \geq 0$ and $\sum_\alpha p(\alpha) = 1$, does not increase $\|\cF\|_k$.
    \end{enumerate}
\end{lemma}
\begin{proof}
\emph{Subsetting.}
    Let $\{\Pi_j,\, h_\alpha(j)\}$ be an optimal strategy for $\cF$ with $|h_\alpha(j)| \leq \|\cF\|_k$. Restricting to the estimators $\{h_\alpha(j)\}_{\alpha \in \cG}$ gives a valid strategy for $\cG$ with the same POVM and the same bound, so $\|\cG\|_k \leq \|\cF\|_k$.
\emph{Convex combination.}
    Let $\mathcal{C}(\cF)$ denote the set obtained by replacing only one of its functions $f_{\alpha_0}$ with a new function $g = \sum_\alpha p(\alpha) f_\alpha$. Using the same optimal POVM $\{\Pi_j\}$ for $\cF$, define the estimator for $g$ as $h_g(j) := \sum_\alpha p(\alpha)\, h_\alpha(j)$. Then $\sum_j h_g(j)\,\Pi_j = \sum_\alpha p(\alpha) F_\alpha = G$, and the triangle inequality gives $|h_g(j)| \leq \sum_\alpha p(\alpha)\,|h_\alpha(j)| \leq \|\cF\|_k$. All other estimators are unchanged, so $\|\mathcal{C}(\cF)\|_k \leq \|\cF\|_k$.
\end{proof}

It is well known that the averages of a set of observables $\{O_1,\ldots,O_{K}\}$ with $\|O_i\|_{\infty} \leq 1$ can be jointly estimated if they mutually commute.
We then show that inserting operators whose $k$-linearizations commute with \emph{every} element of the $k$-linearization of $\cF$ does not increase $\|\cF\|_k$.

\begin{lemma}[Commuting extensions]\label{lem:commuting_extension}
    Let $\cF, \{F_\alpha\}_{\alpha \in \cA}$ be a set of functions of $\rho \in \cH$ and its $k$-linearization with $\mathrm{QB}_k$-norm $\|\cF\|_k$, and let $c(\rho)$ be a function with $k$-linearization $C$ satisfying  $[C, F_\alpha] = 0$ for all $\alpha \in \cA$. Then $\|\cF \cup \{c\}\|_k = \max\!\left(\|\cF\|_k,\; \|\{c\}\|_k\right)$.
    In particular, if $\|\{c\}\|_k \leq \|\cF\|_k$, adding the function $c$ to $\cF$ does not increase the $\mathrm{QB}_k$-norm.
\end{lemma}

\begin{proof}
    The lower bounds $\|\cF \cup \{c\}\|_k \geq \|\cF\|_k$ and $\|\cF \cup \{c\}\|_k \geq \|\{c\}\|_k$ follow from monotonicity in Lemma~\ref{lem:mono1} for $\cF \subset \cF \cup \{c\}$ and $\{c\} \subset \cF \cup \{c\}$, respectively.
    For the upper bound, let $F_{\alpha} = \sum_{j} h_\alpha(j)\Pi_j$ be optimal for $\cF$ with $\max_{\alpha,j}|h_\alpha(j)|= \|\cF\|_k$.
     Let $C = \sum_\mu \lambda_\mu Q_\mu$ be the spectral decomposition of the Hermitian operator $C$.
     Since $C$ is the symmetric linearization of $\{c\}$, we conclude $\|\{c\}\|_k = \|C\|_\infty$.
    To prove this, following from Lemma~\ref{lem:linearization_symmetrization}, we note that
     \begin{align}
     \|C\|_\infty = \left\|\sum_j h(j)\Phi_k(\Pi_j)\right\|_\infty \leq \max_j|h(j)| \leq \|\{c\}\|_k,\label{eq:E10}
\end{align}
    and the upper bound $\|C\|_\infty \geq \|\{c\}\|_k$ is achieved by the spectral measurement $\{Q_\mu\}$ with estimators $h(\mu) = \lambda_\mu$.

     Since $[C, F_\alpha] = 0$ for all $\alpha$, the spectral projectors of $C$ commute with all $F_\alpha$. This follows from the spectral-calculus fact that each $Q_\mu$ is a polynomial in $C$, via Lagrange interpolation on the spectrum of $C$, and $[C, F_\alpha] = 0$ implies $[p(C), F_\alpha] = 0$ for any polynomial $p$.
    Define the fine-grained POVM
    $\Pi_{j,\mu} := Q_\mu\, \Pi_j\, Q_\mu$, which
    satisfies $\Pi_{j,\mu} \succeq 0$, since $Q_\mu$ is
    a projector and $\Pi_j \succeq 0$, and completeness:
    \begin{align}
        \sum_{j,\mu} \Pi_{j,\mu}
        = \sum_\mu Q_\mu\!\left(\sum_j \Pi_j\right)\!Q_\mu
        = \sum_\mu Q_\mu\,\1\,Q_\mu
        = \sum_\mu Q_\mu = \1.\label{eq:E11}
\end{align}
    Set estimators $\tilde{h}_\alpha(j,\mu) = h_\alpha(j)$
    for $\alpha \in \cA$ and $\tilde{h}_C(j,\mu) = \lambda_\mu$.
    Using $[Q_\mu, F_\alpha] = 0$, we get the unbiasedness for $F_{\alpha}$
    \begin{align}
        \sum_{j,\mu} \tilde{h}_\alpha(j,\mu)\,Q_\mu\,\Pi_j\,Q_\mu
        = \sum_\mu Q_\mu\,F_\alpha\,Q_\mu
        = F_\alpha \sum_\mu Q_\mu = F_\alpha.\label{eq:E12}
\end{align}
    The unbiasedness for $C$ follows as
    $\sum_{j,\mu} \lambda_\mu\,Q_\mu\,\Pi_j\,Q_\mu
    = \sum_\mu \lambda_\mu\,Q_\mu = C$.
    The scaling factor is
    \begin{align}
        \max\{|\tilde{h}_\alpha(j,\mu)|,|\tilde{h}_C(j,\mu)|\}
      = \max\{\|\cF\|_k,\,\|C\|_\infty\},\label{eq:E13}
\end{align}
    which completes the proof.
\end{proof}

By iterating Lemma~\ref{lem:commuting_extension}, any set of mutually commuting Hermitian operators $\{F_\alpha\}$ satisfies $\|\{F_\alpha\}\|_k = \max_\alpha \|F_\alpha\|_\infty$.
When the commutativity condition fails, the cost of extending $\cF$ is bounded but generically non-zero and turns out to give the subadditivity of $k$-replica $\mathrm{QB}_k$-norm measure.

\begin{lemma}[Set-union subadditivity of $\mathrm{QB}_k$-norm]\label{lem:subadditivity}
    For two sets of function properties $\cF$ and $\cG$,
    \begin{align}
        \max(\|\cF\|_k, \|\cG\|_k) \leq \|\cF \cup \cG\|_k \leq \|\cF\|_k + \|\cG\|_k.\label{eq:subadditivity}
\end{align}
    When $\cG$ contains a single function with Hermitian linearization $A$ such that $[A, F_\alpha] = 0$ for all $\alpha$, Lemma~\ref{lem:commuting_extension} tightens the upper bound to $\max(\|\cF\|_k, \|\cG\|_k)$.
\end{lemma}

\begin{proof}
    The identity $\|\cF \cup \cG\|_k \geq \max(\|\cF\|_k, \|\cG\|_k)$ follows directly from monotonicity.
    Let $\{\Pi_j, h_\alpha(j)\}$ and $\{\Gamma_l, g_\beta(l)\}$ be the optimal strategies for $\cF$ and $\cG$ with scaling factors $s_1 = \|\cF\|_k$ and $s_2 = \|\cG\|_k$.
    By Definition~\ref{def:k-copy-qBern_coeff}, we have $s_1 := \max_{\alpha,j} |h_{\alpha}(j)| \geq|h_\alpha(j)|$ and $s_2 := \max_{\beta,l} |g_{\beta}(l)| \geq |g_\beta(l)|$ for all $\alpha \in \cA$ and $\beta \in \cB$.
    Define the mixed POVM
    \begin{align}
        \tilde{\Pi}_{(1,j)} = \frac{s_1}{s_1 + s_2}\,\Pi_j, \qquad \tilde{\Pi}_{(2,l)} = \frac{s_2}{s_1 + s_2}\,\Gamma_l,\label{eq:E15}
\end{align}
    which satisfies $\sum_{(t,m)} \tilde{\Pi}_{(t,m)} = \1$. Set estimators $\tilde{h}_\alpha(1,j) = (s_1+s_2)s_1^{-1}\,h_\alpha(j)$ and $\tilde{h}_\alpha(2,l) = 0$ for each $F_\alpha \in \cF$. The rescaling cancels the POVM dilution:
    \begin{align}
        \sum_{(t,m)} \tilde{h}_\alpha(t,m)\,\tilde{\Pi}_{(t,m)} = \sum_j \frac{s_1+s_2}{s_1}\,h_\alpha(j) \cdot \frac{s_1}{s_1+s_2}\,\Pi_j = \sum_j h_\alpha(j)\,\Pi_j = F_\alpha,\label{eq:E16}
\end{align}
        with $|\tilde{h}_\alpha(1,j)| \leq s_1 + s_2$ given that $s_1 \geq |h_\alpha(j)|$. The estimators for $\cG$ are constructed symmetrically with $\tilde{g}_\beta(1,j) = 0$ and $\tilde{g}_\beta(2,l) ={s_2}^{-1}{(s_1+s_2)}g_\beta(l)$, whose absolute value is also upper bounded by $s_1 + s_2$.
    This completes the proof.
\end{proof}

The gap between the bounds in~\eqref{eq:subadditivity} is generically non-zero and reflects algebraic structure beyond pairwise relations. For instance, $\|\{X,Y\}\|_1 = \sqrt{2}$ and $\|\{Z\}\|_1 = \|Z\|_\infty = 1$ give the window $\sqrt{2} \leq \|\{X,Y,Z\}\|_1 \leq \sqrt{2}+1$, while the true value is $\sqrt{3}$, as computed in Eq~\eqref{eq:case2-povm} above.
 More generally, the subset values $\{\|\cF\|_k, \|\cG\|_k\}$ do not uniquely determine $\|\cF \cup \cG\|_k$. The full operator algebra such as anticommutation relations is needed.


\appsubsection{Relation to the worst-case variance}
In the $\mathrm{QB}_k$-norm framework,
we use the $\mathrm{QB}_k$-norm as a sampling complexity indicator.
However, a more standard approach in estimation theory is to use the worst-case estimation variance (or worst-case mean-squared risk in the biased case) to bound the sampling complexity.
Indeed, the Bernstein norm and quantum Bernstein norm themselves already provide an upper bound for the worst-case estimation variance by Hoeffding's inequality (see Eq.~\eqref{eq:bn_bd_variance} for the classical case).
However, further discussion is needed to show that computing the variance bound is computationally intractable compared to the $\mathrm{QB}_k$-norm approach.

\paragraph{Worst-case second moment for a single property.}
Fix any $k$-replica unbiased estimator $\hat{f}$ for a function $f(\rho)$ generated by strategy $\mathcal{S}:= \{\Pi_j, h(j)\}_{j \in \cJ}$ , the estimation variance is related to the second moment of the estimator $\mathbb{E}_\rho[\hat{f}^2] = \sum_j |h(j)|^2 \tr(\Pi_j \rho^{\ox k})$ through $\mathrm{Var}_\rho(\hat{f}) = \mathbb{E}_\rho[\hat{f}^2] - (\mathbb{E}_\rho[\hat{f}])^2$. Unbiasedness gives $\mathbb{E}_\rho[\hat{f}] = f(\rho)$, which is bounded if we consider a normalized property $\cF$ with $\max_\rho |f(\rho)| \leq 1$.
The worst-case variance of property $f(\rho)$ is thus related to the second moment $M_2(\mathcal{S},\rho) := \mathbb{E}_\rho[\hat{f}^2]$.
To search for the best strategy that minimizes $M_2(\mathcal{S},\rho)$ with respect to an arbitrary state $\rho$, we define the second-moment operator $\mathbf{M}(\mathcal{S})$ with respect to the strategy $\mathcal{S}$ that gives the second moment of the estimator $\hat{f}$ through $M_2(\mathcal{S},\rho) = \tr(\mathbf{M}(\mathcal{S}) \rho^{\ox k})$:
\begin{align}
    \mathbf{M}(\mathcal{S}) := \sum_{j \in \cJ} |h(j)|^2\, \Phi_k(\Pi_j).\label{eq:E33}
\end{align}
Here $\Phi_k$ is the $k$-fold symmetrization map.
Similar to the proof in linearization Lemma~\ref{thm:structure_fcomp}, we can show that the second-moment operator $\mathbf{M}(\mathcal{S})$ is not unique, and thus should be considered up to its symmetrization over replicas.
Below, we suppose that $\Pi_j$ itself is a symmetric POVM element and omit the symmetrization operator $\Phi_k$ for simplicity.
The worst-case second moment is then given by $M_2(\mathcal{S}) := \max_\rho M_2(\mathcal{S},\rho) = \max_\rho \tr(\mathbf{M}(\mathcal{S}) \rho^{\ox k})$. We also define the minimax second moment as
\begin{align}
   M_2^*(f) := \min_{\mathcal{S} = \{h(j), \Pi_j\}_{j \in \mathcal{J}}} \quad & M_2(\mathcal{S}) = \min_{\mathcal{S}} \left\|\mathbf{M}(\mathcal{S})\right\|_{s,\mathrm{inj}}  =\min_{\mathcal{S}} \left\|\sum_{j \in \cJ} |h(j)|^2\, \Pi_j\right\|_{s,\mathrm{inj}} ,\label{eq:E34}
\end{align}
which is related to the injective tensor norm of $\mathbf{M}(\mathcal{S})$ in Eq.\eqref{eq:injective-norm}.
We first prove that the square of the symmetric linearization of $f$ can serve as a lower bound of $\mathbf{M}(\mathcal{S})$.
\begin{lemma}\label{lem:variance-lower-bound-symmetric}
    If the function $f(\rho)$ admits  a symmetric $k$-linearization $F$, then for any $k$-replica unbiased estimation strategy $\mathcal{S}:= \{\Pi_j, h(j)\}_{j \in \cJ}$, the second-moment operator satisfies $\mathbf{M}(\mathcal{S}) \succeq F^2$.
\end{lemma}
\begin{proof}
    Based on Theorem~\ref{thm:structure_fcomp}, the symmetric $k$-linearization $F$ satisfies $F = \sum_j h(j)\,\Phi_k(\Pi_j)$.
    We then use the Naimark dilation to reduce the claim to the Kadison–Schwarz inequality~\cite{kadison1952generalized}.
    Define the Naimark isometry $V: \cH^{\ox k} \to \cH^{\ox k} \otimes \cH_{\mathrm{anc}}$ by  $V = \sum_{j} \Pi_j^{1/2} \otimes \ket{j}$ and $V^\dagger = \sum_{j} \Pi_j^{1/2} \otimes \bra{j}$, such that $V\ket{\psi} = \sum_j \Pi_j^{1/2}\ket{\psi} \otimes \ket{j}$. Here $\{\ket{j}\}$ is an orthonormal basis for the ancillary space $\cH_{\mathrm{anc}}$.
    Since $V^\dagger V = \sum_j \Pi_j = \1^{\ox k}$, the map $\Psi(\cdot) := \Phi_k\left[V^\dagger(\cdot)V\right]$, where $\Phi_k$ is the $k$-fold symmetrization map, is a unital completely positive (hence positive) map.
    Introducing the diagonal operator $H := \sum_j h(j)\,\proj{j}$ on $\cH_{\mathrm{anc}}$, we have
    \begin{align}
        &\Psi(\1 \otimes H) =\Phi_k\left[\left(\sum_{m} \Pi_m^{1/2} \otimes \bra{m}\right) \left(\sum_j h(j)\,\1 \otimes \proj{j}\right) \left(\sum_{m'} \Pi_{m'}^{1/2} \otimes \ket{m'}\right)\right] =  \sum_j h(j)\,\Phi_k(\Pi_j) = F,\label{eq:E35}
\\
        &\Psi(\1 \otimes H^2) = \Phi_k\left[\left(\sum_{m} \Pi_{m}^{1/2} \otimes \bra{m}\right) \left(\sum_j  |h(j)|^2\,\1 \otimes \proj{j}\right) \left(\sum_{m'} \Pi_{m'}^{1/2} \otimes \ket{m'}\right)\right] =  \sum_j |h(j)|^2\,\Phi_k(\Pi_j) = \mathbf{M}.\label{eq:E36}
\end{align}
    The Kadison–Schwarz Inequality~\cite{kadison1952generalized} states that for the positive and unital linear map $\Psi$, any self-adjoint operator $O$ satisfies $ \Psi(O^2) \succeq \Psi(O)^2$.
    Applying this to the case $O = \1 \otimes H$ gives $ \mathbf{M} = \Psi(O^2) \succeq \Psi(O)^2 = F^2$.
\end{proof}
Lemma~\ref{lem:variance-lower-bound-symmetric} implies that $\tr(\mathbf{M}(\mathcal{S}) \rho^{\ox k}) \geq \tr(F^2\,\rho^{\ox k})$ for any state $\rho$.
Then any unbiased estimator $\hat{f}$ that is constructed from the $k$-replica strategy $(\Pi_j, h(j))$ satisfies
    \begin{align}
        M_2(\hat{f}) = \max_\rho \mathbb{E}_\rho[\hat{f}^2]  = \max_\rho \tr(\mathbf{M}(\mathcal{S}) \rho^{\ox k}) \geq \max_\rho \tr(F^2 \rho^{\ox k}).\label{eq:E37}
\end{align}
Applying this inequality to the single non-linear estimator $f(\rho) = \mathrm{tr}(\rho\cdot \rho^*)$ proves a no-go theorem for the two-replica unbiased estimation, which is sharply different from the case of purity estimation $f(\rho) = \mathrm{tr}(\rho^2)$ that admits $O(1)$ sampling complexity under Bell measurement.
\begin{corollary}\label{cor:variance-lower-bound-rho-rho-star}
Any 2-replica unbiased estimation strategy for $f(\rho) = \mathrm{tr}(\rho \rho^*), \forall \rho \in \cH(d)$ satisfies $M_2^*(f) \geq d$.
\end{corollary}
\begin{proof}
The symmetric $2$-linearization of $\mathrm{tr}(\rho\rho^*)$ is given by $ F = d \ket{\Phi}\bra{\Phi}$ with $\ket{\Phi} = \sqrt{d^{-1}} \sum_{j = 1}^d\ket{jj}$.
To verify this, one uses the fact $(\rho^*)_{kj} = \rho_{jk}$ and computes $\tr(F \rho^{\ox 2}) =\sum_{j,k = 1}^d \rho_{jk} \rho_{jk} = \tr(\rho \cdot \rho^*) = f(\rho)$.
Then based on Lemma~\ref{lem:variance-lower-bound-symmetric}, for any $2$-replica unbiased estimation strategy $\mathcal{S}$, we have $\max_\rho \mathbb{E}_\rho[\hat{f}^2] \geq \max_\rho \tr(F^2 \rho^{\ox 2}) = d$. This means that any unbiased $2$-replica estimation strategy for $f(\rho) = \mathrm{tr}(\rho\cdot \rho^*)$ with respect to an arbitrary density matrix must have a second moment lower bounded by $d$.
\end{proof}

\paragraph{Worst-case second moment for multiple properties.}
For the case of multiple properties $\cF = \{f_\alpha\}_{\alpha \in \cA}$,
let $\cS := \{(\Pi_j,h_\alpha(j))\}_{j \in \cJ,\alpha \in \cA}$ be a valid $k$-replica estimation strategy for a set of properties $\cF = \{f_\alpha\}_{\alpha \in \cA}$.
The strategy $\cS = \{(\Pi_j,h_\alpha(j))\}_{j \in \cJ, \alpha \in \cA}$ now constructs a $k$-replica unbiased estimator $\hat{f}_{\alpha}$ for each function $f_{\alpha}(\rho)$ respectively.
Define the set of \emph{second-moment operators} $\{\mathbf{M}_\alpha(\cS)\}_{\alpha \in \cA}$ that satisfy $\mathbb{E}_{\rho}(\hat{f}^2_{\alpha}) = \tr({\mathbf{M}_{\alpha}(\cS)} \rho^{\ox k})$ for each $\alpha \in \mathcal{A}$:
\begin{align}
     \mathbf{M}_\alpha(\cS) = \sum_{j \in \cJ} |h_\alpha(j)|^2\, \Phi_k(\Pi_j), \quad \forall \alpha \in \cA.\label{eq:second_moment_op}
\end{align}
The individual lower bound also applies through $\mathbf{M}_\alpha(\cS) \succeq F_{\alpha}^2$ for each $\alpha \in \cA$ and arbitrary unbiased estimation strategy $\cS$.
However, this lower bound is not tight due to the incompatibility of the properties in $\cF$.
Using the $\|\cdot\|_{s,\mathrm{inj}}$-norm notation defined in Eq.~\eqref{eq:injective-norm}, an optimization problem in analogy with $\mathrm{QB}_k$-norm can be proposed:

\begin{equation}
     {
    \begin{aligned}
    M_2^*(\cF,\mathcal{C}) := \min_{(h_\alpha(j), \Pi_j)} \quad & \max_{\rho,\alpha} \mathrm{tr}(\mathbf{M}_\alpha(\mathcal{S}) \rho^{\ox k}) = \max_{\alpha} \left\|\sum_{j \in \cJ} |h_\alpha(j)|^2\, \Pi_j\right\|_{s,\mathrm{inj}} \\
    \text{Subject to:} \quad & \sum_{j} h_\alpha(j)\,\Pi_j = F_\alpha, \quad \forall\,\alpha \in \cA ,\\
    & \sum_j \Pi_j = \,\1^{\ox k},\quad  \Pi_j \in \mathcal{C}.
    \end{aligned}
    }\label{variance-sdp-problem}
\end{equation}

This benchmark is generally sharper than the $\mathrm{QB}_k$-norm,
and has been partially investigated by \cite{innocenti2023shadow} in the single-replica case and for certain simple choices of $\Pi_j$.
It is substantially less tractable compared to the $\mathrm{QB}_k$-norm problem.
The reason is that the bounded-output constraint
$|h_\alpha(j)|\le r$ admits an exact sign-vector convexification (See proof in Lemma~\ref{lem:dual-robustness}) that converts the $\mathrm{QB}_k$ problem into a finite
linear conic program. In contrast, the second moment benchmark requires square terms of $h_{\alpha}(j)$ on the second-moment operator $\mathbf{M}_\alpha$.
This mirrors the classical situation that the exact minimax risk is usually hard
to compute directly~\cite{wu2020polynomial,jiao2015minimax}, and Bernstein-type quantities are useful to replace the original risk optimization by a tractable convex approximation
problem.

Here below we show how $M_2^*(\cF)$ relates to the $\mathrm{QB}_k$-norm through $\|\cF\|_{k}^2$.
We first show that the optimal strategy, recovered via the complementary slackness procedure, achieves variance \emph{exactly} $\|\cF\|_{k}^2$.
This gives an upper bound on the worst-case variance of the $k$-replica unbiased estimation strategy.

\begin{lemma}\label{lem:sdp-variance-saturation}
    Let $\cS^* = \{(\tilde{\Pi}_j, h_\alpha(j))\}_{j \in \mathcal{J}}$ be an estimation strategy for $\cF = \{f_{\alpha}(\rho)\}_{\alpha \in \cA}$ that reaches the optimal $\mathrm{QB}_k$-norm $s^*_k(\cF)$ in Definition~\ref{def:k-copy-qBern_coeff}, and gives a $k$-replica unbiased estimator $\hat{f}_{\alpha}$ for each function $f_{\alpha}(\rho)$ respectively.
    Then $\mathbb{E}_{\rho}(\hat{f}^2_{\alpha})  = \|\cF\|_{k}^2, \forall \alpha \in \cA, \rho \in \mathcal{H}$.
\end{lemma}
\begin{proof}
    From the complementary slackness in Eq.\eqref{eq:POVM_recovery}, $h_\alpha(c) = s^*_k(\cF)\,c_\alpha$ for every active sign vector $c$, so $|h_\alpha(c)| = s^*_k(\cF) =\|\cF\|_{k}$ universally. This then gives the second-moment operator as $\mathbf{M}_\alpha(\cS^*) = \sum_{c \in \cJ} |h_\alpha(c)|^2 \tilde{\Pi}_c = {[s^*_k(\cF)]}^2 \1^{\ox k}$. By definition of $\mathbf{M}_\alpha(\cS^*)$, we have $\mathbb{E}_{\rho}(\hat{f}^2_{\alpha})= \tr(\mathbf{M}_\alpha \rho^{\ox k}) =\|\cF\|_{k}^2$.
\end{proof}

Because the optimal strategy $\cS^*$ always outputs a result $h_{\alpha}(j)$ with the same magnitude $|h_{\alpha}(j)| = \|\cF\|_{k}$,
Lemma~\ref{lem:sdp-variance-saturation} indicates that the optimal strategy $\cS^*$ in terms of minimizing the $\mathrm{QB}_k$-norm
always gives an $\alpha$ independent and state independent second-moment operator.
We now propose a fundamental lower bound based on this geometry of the measurement strategy.

\begin{definition}[Concentration factor of strategy]\label{def:concentration-factor}
    For any estimation strategy $\cS = \{(\Pi_j,h_\alpha(j))\}_{j \in \cJ, \alpha \in \cA}$,
    let $\rho$ be a fixed state. The concentration index set is defined as:
    \begin{align}
       \cJ_{\alpha}^*(\eta, \mathcal{S}):= \{j \in \cJ: |h_\alpha(j)| \geq (1- \eta) \max_{j \in \mathcal{J}, \alpha \in \cA} |h_\alpha(j)|\}, \quad \eta \in [0,1].\label{eq:E40}
\end{align}
    Then the concentration factor $w_{\eta}(\cS, \rho)$ is defined as $w_{\eta}(\cS, \rho) :=\min_{\alpha \in \cA}\tr\left(\sum_{j\in \cJ_{\alpha}^*(\eta, \mathcal{S})}\Pi_{j} \rho^{\ox k}\right)$.
\end{definition}

The geometric concentration bound perfectly captures how sub-optimal measurement geometries inflate the actual variance.

\begin{lemma}[Variance lower bounds]\label{lem:variance-lower-bound}
For any function set $\cF :=\{f_\alpha(\rho)\}_{\alpha \in \cA}$,
consider a strategy $\mathcal{S} = \{(\Pi_j, h_\alpha(j))\}_{j \in \cJ, \alpha \in \cA}$, which gives a $k$-replica unbiased estimator $\hat{f}$ and satisfies $w_{\eta}(\mathcal{S}, \rho) \geq r$ for the fixed state $\rho$ and $\eta \in [0,1]$,
then $\mathbb{E}_\rho[\hat{f}_\alpha^2] \geq  r \cdot \left(1- \eta\right)^2 \cdot \|\cF\|_{k}^2, \forall \alpha \in \cA$.
\end{lemma}

\begin{proof}
The bound can be obtained by discarding the non-extreme outcomes and substituting the peak value, we have:
\begin{align}
     \mathbb{E}_\rho[\hat{f}_\alpha^2] &=  \sum_{j \in \cJ} |h_\alpha(j)|^2 \tr(\Pi_j \rho^{\ox k}) \geq \sum_{j \in \cJ_{\alpha}^*(\eta, \mathcal{S})} |h_\alpha(j)|^2 \tr(\Pi_j \rho^{\ox k})\label{eq:E41}
\\
    &\geq  (1- \eta)^2 \left[\max_{j \in \mathcal{J}, \alpha \in \cA} |h_\alpha(j)|\right]^2 \tr\left(\sum_{j \in \cJ_{\alpha}^*(\eta, \mathcal{S})}\Pi_j \rho^{\ox k}\right) \geq (1- \eta)^2 \|\cF\|_{k}^2 w_{\eta}(\cS,\rho).\label{eq:E42}
\end{align}
The final inequality follows from the universal definition of the $\mathrm{QB}_k$-norm $\|\cF\|_{k}\leq s_{\mathrm{strat}}$ and the definition of the concentration factor $w_{\eta}(\cS)$ in Definition~\ref{def:concentration-factor}.
\end{proof}

\paragraph{Variance lower bounds for multi-copy shadows.}
Shadow tomography~\cite{aaronson2018shadow} is a family of strategies that enable us to learn many properties of a quantum state with very few measurements.
One critical problem is that given a known set of target properties $\cF$,
 what is the optimal shadow strategy to estimate it?
Using the $\mathrm{QB}_k$-norm framework, we partially answer this question by showing an inherent limitation of the classical shadow strategy for a fixed target property set.
We first review the standard non-adaptive multi-copy local shadow tomography~\cite{liu2026auxiliary}:
\begin{itemize}
    \item Set $U \sim \cU$, where $\cU$ is a $d^k$-dimensional unitary ensemble.
    \item Measure the state $U \rho^{\ox k} U^\dagger$ in the computational basis with outcomes $\ket{\mathbf{b}}\!\bra{\mathbf{b}}$. Record the classical description of the circuit $U$ and the measurement outcome string $\mathbf{b} \in \mathbb{F}_d^{kn}$.
    \item \textit{Virtually} compute the shadow matrix $\hat \sigma_{U, \mathbf{b}}:= \cM^{-1}(U^\dagger \ket{\mathbf{b}}\!\bra{\mathbf{b}} U)$, where $\cM$ is the entanglement-breaking channel on the Hilbert space $\cH^{\ox k}$ that maps $\sigma \to \mathbb{E}_{U \sim \cU} \sum_{\mathbf b}\bra{\mathbf b}U \sigma U^\dagger \ket{\mathbf{b}} \cdot  U^\dagger\ket{\mathbf b}\!\bra{\mathbf b} U$.
\end{itemize}
The explicit form of the channel $\cM$ relies on the ensembles $\cU_i$.
For an informationally complete classical shadow protocol, the inverse $\cM^{-1}$ exists.
As shown in Lemma~\ref{lem:single_shadow_unbiased},
the shadow matrix $\hat{\sigma}$ can be used to construct an arbitrary unbiased estimator of $\tr(O \rho^{\ox k})$ for an arbitrary observable $O \in \mathrm{Herm}(\mathcal{H})$.

\begin{lemma}[Unbiasedness of multi-copy shadows]\label{lem:single_shadow_unbiased}
    The estimator $\hat{o}^{(k)} := \tr(O \hat{\sigma})$ is an unbiased estimator for the expectation $\tr(O \rho^{\ox k})$, where $\hat{\sigma}$ is sampled according to the multi-copy shadow tomography protocol above.
\end{lemma}
\begin{proof}
$ \mathbb{E}_{\hat \sigma} \tr(O \hat{\sigma}) = \mathbb{E}_{U \sim \cU}\mathbb{E}_{\mathbf{b} \sim U \rho^{\ox k} U^\dagger} \tr\left[O \cM^{-1}(U^\dagger \ket{\mathbf{b}}\!\bra{\mathbf{b}} U)\right]=  \tr\left[O \cM^{-1}(\mathbb{E}_{U\sim \cU}\mathbb{E}_{\mathbf{b} \sim U \rho^{\ox k} U^\dagger}U^\dagger \ket{\mathbf{b}}\!\bra{\mathbf{b}} U)\right]=\tr(O \rho^{\ox k}).$
\end{proof}

By fixing the target observable set $O \in \{F_{\alpha}\}_{\alpha \in \cA}$,
any classical shadow strategy gives a specific unbiased estimation strategy $\mathcal{S}_{\mathcal{U}}:= \{(\Pi_j, h_{\alpha}(j))\}_{j \in \mathcal{J}}$ that jointly estimates $f_{\alpha} := \tr(F_{\alpha} \rho^{\otimes k})$.
Specifically, the combined measurement outcome is indexed by $j = (U, \mathbf{b})$, and the corresponding POVM elements are $\Pi_j = p(U) U^\dagger \ket{\mathbf{b}}\!\bra{\mathbf{b}} U$, where $p(U)$ is the probability of drawing the unitary $U$ from the ensemble $\cU$.
The single-shot estimator value for the outcome $j =(U, \mathbf{b})$ is exactly the trace inner product between the target operator and the inverted shadow:
\begin{align}
    h_\alpha(j) = \tr\left[F_\alpha \cdot \cM^{-1}(U^\dagger \ket{\mathbf{b}}\!\bra{\mathbf{b}} U)\right] = \bra{\mathbf{b}} U \cM^{-1}(F_\alpha) U^\dagger \ket{\mathbf{b}},\label{eq:E43}
\end{align}
where we use the self-adjoint property of the inverse measurement channel under the Hilbert-Schmidt inner product, i.e., $\tr\left(O \cM^{-1}(A)\right) = \tr\left(\cM^{-1}(O)A\right)$.
To show an effective lower bound on the estimation variance, we first calculate the scaling factor $s$ for $\mathcal{S}_{\mathcal{U}}$, as defined in Definition~\ref{def:k-copy-qBern_coeff}, and the concentration factor $w_{\eta}(\cS_{\mathcal{U}},\rho)$ in Definition~\ref{def:concentration-factor}:
\begin{align}
    s(\mathcal{S}_{\mathcal{U}}) &:=\max_{j \in \cJ, \alpha \in \cA} |h_\alpha(j)| = \max_{U \in \cU, \mathbf{b},\alpha \in \cA} \left| \bra{\mathbf{b}} U \cM^{-1}(F_\alpha) U^\dagger \ket{\mathbf{b}} \right|,\label{eq:E44}
\\
    w_{\eta}(\cS_{\mathcal{U}},\rho)
    &= \min_{\alpha \in \cA} \sum_{U}p(U)  \cdot \sum_{\substack{\mathbf{b}}} I\left[ |\tr(\hat \sigma_{U, \mathbf{b}} F_{\alpha})| \geq (1- \eta) s(\mathcal{S}_{\mathcal{U}})\right] \cdot \tr\left(\ket{\mathbf{b}}\!\bra{\mathbf{b}} \cdot U\rho^{\otimes k}U^\dagger\right).\label{eq:E45}
\end{align}
The concentration factor $w_{\eta}(\cS_{\mathcal{U}},\rho)$ actually quantifies the probability that the absolute value of the output of the classical shadow result, $|\tr(\hat \sigma_{U, \mathbf{b}} F_{\alpha})|$, is concentrated on its maximum value.
Crucially, by invoking the second-moment operator bound from Lemma~\ref{lem:variance-lower-bound}, we immediately obtain a worst-case estimation variance bound for classical shadows without needing to compute the full state-dependent trace:
\begin{align}
  \mathbb{E}_\rho[\hat{f}_\alpha^2] \geq  w_{\eta}(\cS_{\mathcal{U}},\rho) \cdot \left(1- \eta\right)^2 \cdot \|\cF\|_{k}^2 , \quad \forall \alpha \in \cA, \quad \forall \eta \in [0,1].\label{eq:E46}
\end{align}
This inequality bounds the worst-case estimation variance through the identity $\mathbb{E}_\rho[\hat{f}_\alpha^2] = \mathrm{Var}_\rho[\hat{f}_\alpha] + \mathbb{E}_\rho[\hat{f}_\alpha]^2$.
Specifically, when the result with sharply large and small values rarely happens (i.e., $w_{\eta}(\cS_{\mathcal{U}},\rho)$ is a constant) and the function property $\{f_{\alpha}\}_{\alpha \in \cA}$ we consider contains too much incompatibility (i.e., $\|\cF\|_{k}$ is exponentially large), the estimation variance of this classical shadow strategy will always be large.

\begin{corollary}[classical shadow lower bounds]\label{cor:shadow-variance-lower-bound}
    Consider a $k$-replica classical shadow strategy devoted to constructing an unbiased estimator $\hat{f}_{\alpha}$ for a set of polynomial properties $\cF = \{f_{\alpha}(\rho)\}_{\alpha \in \cA}$.
    Suppose the absolute value of the estimator's output is concentrated in the range $[(1- \eta)\cdot r, r]$ with probability $w_{\eta}$, where $r$ is the peak output value and $\eta \in [0,1]$.
    Then the estimation variance for each property $f_{\alpha}(\rho)$ is lower bounded by $\mathrm{Var}_\rho[\hat{f}_\alpha] \geq  w_{\eta} \cdot \left(1- \eta\right)^2 \cdot \|\cF\|_{k}^2 - f_{\alpha}(\rho)^2, \,\forall \alpha \in \cA$.
\end{corollary}

The problem left is whether an appropriate $\eta$ can generate a large enough concentration factor such that Corollary~\ref{cor:shadow-variance-lower-bound} gives a useful lower bound.
This depends strongly on the specific choice of the shadow tomography ensemble $\cU$ and the corresponding properties $\{f_{\alpha}\}_{\alpha \in \cA}$.
For some uniform classical shadow strategies~\cite{mcnulty2023estimating,nguyen2022optimizing}, Corollary~\ref{cor:shadow-variance-lower-bound} gives a tight lower bound.
For general strategies,
we believe that Corollary~\ref{cor:shadow-variance-lower-bound} at least provides a feasible way to check the effectiveness of the shadow tomography strategy for a given set of properties by measuring $w_{\eta}$.

\paragraph{Target-aware strategy optimization.}
After the first proposal of practical shadow tomography~\cite{huang2020predicting},
many strategies have been designed to optimize the shadow tomography protocol for known target observables.
Some of them fix the available measurement to the single qubit Pauli projectors and then derandomize~\cite{huang2021efficient} or add bias~\cite{hadfield2022measurements} to the protocol based on the target Hamiltonian.
Such an optimization strategy has been extended to more general fixed POVMs~\cite{korhonen2025improving,innocenti2023shadow,cha2026operator}, and has also been considered from the perspective of measurement compatibility~\cite{mcnulty2023estimating}.
In the following, we will show that the $\mathrm{QB}_k$-norm optimization problem also gives a valid target-aware shadow strategy in parallel with those previous works.
Such a target-aware shadow strategy can be applied to targets specified by a set of non-linear polynomials while only involving a linear program.

Following the general definition in Eq.\eqref{eq:second_moment_op}, given a set of function properties $\cF = \{f_\alpha(\rho)\}_{\alpha \in \cA}$ with $k$-replica linearizations $\{F_\alpha\}_{\alpha \in \cA}$, the second-moment operator $\mathbf{M}_\alpha$ for this classical shadow strategy is exactly:
\begin{align}
    \mathbf{M}_{\alpha} &= \sum_{j \in \cJ} |h_\alpha(j)|^2\, \Pi_j = \sum_{U, \mathbf{b}} \left( \bra{\mathbf{b}} U \cM^{-1}(F_\alpha) U^\dagger \ket{\mathbf{b}} \right)^2 p(U)\, U^\dagger \ket{\mathbf{b}}\!\bra{\mathbf{b}} U \notag
\\
    &= \mathbb{E}_{U \sim \cU} \sum_{\mathbf{b}} U^\dagger \ket{\mathbf{b}}\!\bra{\mathbf{b}} U \bra{\mathbf{b}} U \cM^{-1}(F_\alpha) U^\dagger \ket{\mathbf{b}}^2 \notag
\\
    &= \mathbb{E}_{U \sim \cU} \sum_{\mathbf{b}} \tr_{2,3}\left[ \left(U^\dagger \ket{\mathbf{b}}\!\bra{\mathbf{b}} U\right)_{1} \ox \left(U^\dagger \ket{\mathbf{b}}\!\bra{\mathbf{b}} U \cdot \cM^{-1}(F_{\alpha})\right)_{2} \ox \left(U^\dagger \ket{\mathbf{b}}\!\bra{\mathbf{b}} U \cdot \cM^{-1}(F_{\alpha})\right)_{3} \right] \notag
\\
    &= \tr_{2,3}\left[ \left\{ \mathbb{E}_{U \sim \cU} (U^\dagger)^{\ox 3} \left( \sum_{\mathbf{b}} (\ket{\mathbf{b}}\!\bra{\mathbf{b}})^{\ox 3} \right) U^{\ox 3} \right\} \left( \1 \ox \cM^{-1}(F_{\alpha}) \ox \cM^{-1}(F_{\alpha})\right) \right].\label{eq:shadow_variance_povm}
\end{align}
This operator reflects the worst-case estimation variance when we use this classical shadow strategy to estimate the property $f_{\alpha}$:
\begin{align}
   \max_{\rho \in \cH} \left[ \mathrm{Var}_{\rho}(\hat{f}_{\alpha}) + [f_{\alpha}(\rho)]^2 \right] = \max_{\rho \in \cH} \mathbb{E}_{\rho}(\hat{f}^2_{\alpha}) = \max_{\rho \in \cH} \tr({\mathbf{M}_{\alpha}} \rho^{\ox k}) = \|\mathbf{M}_{\alpha}\|_{s, \mathrm{inj}}.\label{eq:E48}
\end{align}
Here the symmetrization is defined in Eq.\eqref{eq:injective-norm} before.
When $k = 1$, the value of $\|\mathbf{M}_{\alpha}\|_{s, \mathrm{inj}}$ is equal to the infinity norm, thus matching the original definition of shadow norm $\|F_{\alpha}\|^2_{\mathrm{shadow}}$~\cite{huang2020predicting}.

Since classical shadow tomography is a special case of the $\mathrm{QB}_k$-norm framework, a more general optimization problem can be proposed by reframing the dual problem~\eqref{eq:dual_unbiased} with an appropriate POVM restriction $\mathcal{C}$.
For example,
consider the POVMs that attain the solution $\|\cF\|_{k,\mathcal{C}}$ in Lemma~\ref{lem:dual-robustness},
with restriction $\mathcal{C}$ being the projectively simulable POVMs~\cite{oszmaniec2017simulating,guerini2017operational}.
Lemma~\ref{lem:sdp-variance-saturation} states that these POVMs form a valid classical shadow strategy that satisfies $\|\mathbf{M}_{\alpha}\|_{s, \mathrm{inj}}= \|\cF\|_{k,\mathcal{C}}^2$.
Thus, we directly get a classical shadow strategy for estimating $\mathcal{F}$ with variance upper bounded by $\|\cF\|_{k,\mathcal{C}}^2$ and reaches the optimum in terms of the minimum magnitude of estimators (as indicated by Definition~\ref{def:k-copy-qBern_coeff}).

More generally,
for a $k$-replica classical shadow strategy,
one can set up a set of unitaries $\mathbb{U}:=\{U_1,\cdots,U_M\}$ where $U_i$ is the available unitary operation on the $k$-replica space $\mathcal{H}^{\otimes k}$.
Then a classical shadow ensemble $\mathcal{U}$ is constructed by building up a discrete distribution $p(U)$ on this set.
For example, in the case of $k =1$, setting $\mathbb{U} = \mathrm{Cliff}(1),\mathrm{Cliff}(n)$ with uniform distribution recovers the classical shadow based on randomized single qubit Pauli measurement and global Clifford twirling~\cite{huang2020predicting}, respectively.
The choice of $\mathbb{U}$ restricts the feasible POVM operators to the combination:
\begin{align}
    \Pi_j &= \sum_{ U \in \mathbb{U},\mathbf{b} \in (\mathbb{F}_2^{n})^{k}} p(U) D(j|U,\mathbf{b}) U^\dagger \ket{\mathbf{b}}\!\bra{\mathbf{b}} U,\label{eq:E49}
\end{align}
where $D(j|U,\mathbf b)$ is a classical stochastic map satisfying $D(j|U,\mathbf b)\geq 0, \sum_j D(j|U,\mathbf b)=1$.
For notational simplicity, we define the rank-one effects $E_{U,\mathbf b}:=U^\dagger \ket{\mathbf b}\!\bra{\mathbf b}U$.
Take post-processing functions $h_{\alpha}(j)$ targeting the $k$-replica linearization operator $F_{\alpha}$.
The unbiasedness condition for estimating
$f_\alpha(\rho)=\tr(F_\alpha\rho^{\otimes k})$ is
\begin{align}
    F_\alpha = \sum_j h_\alpha(j)\Pi_j =
    \sum_{U,\mathbf b}p(U)\left[\sum_j D(j|U,\mathbf b)h_\alpha(j)\right]E_{U,\mathbf b} := \sum_{U,\mathbf b}p(U) \widetilde h_\alpha(U,\mathbf b)E_{U,\mathbf b},
    \qquad
    \forall \alpha\in\mathcal A,\label{eq:E50}
\end{align}
where the effective raw-outcome post-processing function satisfies $|\widetilde h_\alpha(U,\mathbf b)|\leq \max_j |h_\alpha(j)|$ since $D(\cdot|U,\mathbf b)$ is a probability distribution.
Hence, without loss of optimality, an exact
random-basis shadow strategy can be written in the form $\Pi_{U,\mathbf b}=p(U)E_{U,\mathbf b},\, h_\alpha(U,\mathbf b)$.
This motivates the following fixed-unitary-ensemble $\mathrm{QB}_k$-norm:
\begin{align}
    s_{\mathbb U,k}^{\mathrm{sh}}(\mathcal F)
    :=\inf_{p(U),h_{\alpha}(U, \mathbf{b})}\,\max_{\alpha,U,\mathbf b}|h_\alpha(U,\mathbf b)|,\quad
    \text{s.t.~} F_\alpha=\sum_{U,\mathbf b} p(U)h_\alpha(U,\mathbf b) E_{U,\mathbf b}, \,\forall \alpha\in\mathcal A ; \quad
     p(U)\geq 0, \,\sum_{U\in\mathbb U}p(U)=1\label{eq:E51}
\end{align}
This can be converted into a linear program by first proposing $ c_{\alpha,U,\mathbf b}:=p(U)h_\alpha(U,\mathbf b),\, g_U:=s\cdot p(U)$,
where $s$ denotes the estimator envelope, and then translating the normalization condition for $0\leq p(U) \leq 1$ and envelope restriction $|h_\alpha(U,\mathbf b)| \leq s$ to
$\sum_U g_U=s$ and $|c_{\alpha,U,\mathbf b}|\leq g_U$, respectively.
Therefore, $s_{\mathbb U,k}^{\mathrm{sh}}(\mathcal F)$ is equal to the optimal
value of the following linear program:
\begin{align}
 {
\begin{aligned}
s_{\mathbb U,k}^{\mathrm{sh}}(\mathcal F)=
\min_{\{c_{\alpha,U,\mathbf b}\},\{g_U\}}\quad&\sum_{U\in\mathbb U} g_U\\
\mathrm{s.t.}\quad&F_\alpha=\sum_{U\in\mathbb U}\sum_{\mathbf b\in(\mathbb F_2^n)^k}c_{\alpha,U,\mathbf b}E_{U,\mathbf b},
\qquad\forall \alpha\in\mathcal A,\\
&-g_U\leq c_{\alpha,U,\mathbf b}\leq g_U, \qquad\forall \alpha,U,\mathbf b,\\
&g_U\geq 0, \qquad \forall U\in\mathbb U .
\end{aligned}
}\label{eq:lp-qbn-norm}
\end{align}
Given an optimal solution
$s^* = s_{\mathbb U,k}^{\mathrm{sh}}(\mathcal F)$, $\{c^*_{\alpha,U,\mathbf b}\}$ and $\{g_U^*\}$,
we construct
\begin{align}
    \sum_U g_U^* = s^*,\quad  p^*(U) := \frac{g_U^*}{s^*}, \quad h_\alpha^*(U,\mathbf b) =\frac{s^*\cdot c^*_{\alpha,U,\mathbf b}}{g_U^*},\text{~if~} g_U^* > 0.\label{eq:E53}
\end{align}
If $g_U^*=0$, then the value of $h_\alpha^*(U,\mathbf b)$ can be chosen arbitrarily.
Then the resulting classical shadow strategy is:
(i). sample $U\sim p^*$, apply $U$ to $\rho^{\otimes k}$, (ii). measure in the
computational basis with outcome $\mathbf{b}$, and (iii). output $h_\alpha^*(U,\mathbf b)$ for the estimation of $f_\alpha(\rho)$.
One can verify that this strategy is unbiased and the second-moment operator satisfies the upper bound:
\begin{align}
    \mathbf M_\alpha =
    \sum_{U,\mathbf b} p^*(U) \left[h_\alpha^*(U,\mathbf b)\right]^2 E_{U,\mathbf b}  =  \sum_{U:g_U^* > 0,\mathbf b} \frac{s^*}{g_U^*} \cdot \left(c^*_{\alpha,U,\mathbf b}\right)^2 E_{U,\mathbf b} \preceq  \sum_{U} {s^*}{g_U^*} \cdot \sum_{\mathbf b}E_{U,\mathbf b} = \left[s_{\mathbb U,k}^{\mathrm{sh}}(\mathcal F)\right]^2 \1^{\otimes k},\label{eq:E54}
\end{align}
where in the last inequality we use the fact that $c^2_{\alpha,U,\mathbf b}\leq g^2_U$ and in the last equality we use the fact that $\sum_{\mathbf b}E_{U,\mathbf b} = \1^{\otimes k}$.

\appsection{Analytical benchmarks}\label{app:examples}\label{app:compare}

In this appendix, we analytically evaluate the dual conic programs
associated with the $\mathrm{QB}_k$-norm in two representative settings:
single-copy Pauli shadow tomography and purity estimation under
memory-free measurement restrictions. The purpose of this appendix is
not to rederive the full sample-complexity theorems of shadow tomography,
but rather to show how their characteristic lower-bound scalings emerge
from the geometry of the $\mathrm{QB}_k$-norm dual program.

The two examples play complementary roles.
The Pauli shadow example
illustrates how measurement incompatibility among a large family of
linear observables appears as a large dual norm.
The purity example
illustrates how a nonlinear function can be handled by the same
replica-linearized formalism, with the unitary symmetry reducing the
dual program to a low-dimensional invariant problem.
 In both cases, the
main technical step is to use group symmetry to restrict the dual
variables to the commutant algebra: Clifford covariance for Pauli
observables and unitary covariance for purity.
This reduces the original
semidefinite program, whose variables act on high-dimensional replica
Hilbert spaces, to an optimization problem with scalar variables only.
For convenience, we denote $\|\cdot \|_{N,\mathcal{C}_{\mathrm{mem}}}$ and $\|\cdot \|_{N,\mathcal{C}_{\mathrm{loc}}}$ as $\|\cdot \|_{N,{\mathrm{mem}}}$ and $\|\cdot \|_{N,{\mathrm{loc}}}$, respectively, and summarize the conclusions as follows.

\begin{lemma}[$\mathrm{QB}_k$-norm for Pauli shadow tomography]
\label{lem:Pauli-analytical-solutions}
Let $d=2^n$ and $\cF = \{\tr(P\rho)\}_{P\in\mathcal P_n}$ for
$\mathcal P_n=\{\1,X,Y,Z\}^{\otimes n}$. The following statements hold:
\begin{align}
    &\|\cF\|_1
    =\frac{d^2-1}{\Sigma}, \qquad  \|\cF\|_{N,{\mathrm{mem}}} \geq \frac{d^2 - 1}{N \Sigma} ,\qquad \|\cF\|_{1,{\mathrm{loc}}}
    =(\sqrt{3})^n\label{eq:C1}
\\
    &\text{where~}\Sigma:=\sup_{\psi} \sum_{P \in \{\1,X,Y,Z\}^{\ox n} \setminus \{\1_n\}} \left|   \langle \psi | P | \psi \rangle\right| \in \left[(1 + \sqrt{3})^n -1 ,
    \sqrt{(2^n - 1)(4^n - 1)}\right].\label{eq:C2}
\end{align}
\end{lemma}
\begin{lemma}[$\mathrm{QB}_k$-norm for purity estimation]\label{lem:Purity-analytical-solutions}
Let $d =2^n$ and $\cF = \{\tr(\rho^2)\}$. The following statements hold:
\begin{align}
    \|\mathcal F\|_{2,\mathrm{mem}}=d, \quad
    \|\mathcal F\|_{N,\mathrm{mem}}\ge
     \frac{2(d-1)(d+1)}{d N (N-1)} + \frac{1}{d}.\label{eq:C3}
\end{align}
\end{lemma}

We also emphasize the limitations of these analytical reductions.
 First,
except in the simplest cases, the symmetry reduction gives computable
finite-dimensional programs rather than closed-form optimal values.
For example, the exact value of the $N$-copy Pauli overlap parameter
$\Sigma$ is a magic measure~\cite{leone2022stabilizer} and usually hard to compute exactly~\cite{cuffaro2024quantum}.
 Second, the $N$-copy
bounds above are lower bounds obtained from symmetric dual ansatzes or
restricted feasible reductions.
They need not be tight for all adaptive
or memory-assisted protocols.
Finally, these calculations focus on the
unbiased $\mathrm{QB}_k$-norm. Extending the same analytical program to
biased estimation and more general measurement
restriction cones remains an important direction for future work.

\appsubsection{Examples for Pauli shadow tomography}

\appsubsubsection{The single-replica strategy}

We first consider $\|\cF\|_1$ for $\cF = \{\tr(P\rho)\}_{P\in\mathcal P_n}$ by exploiting the algebraic structure of the Pauli group.
The measurement is not restricted so $\mathcal{K}_{\mathcal{C}}^*$ contains all positive semidefinite operators.
Since $\1_k$ commutes with all Pauli operators, by Lemma~\ref{lem:commuting_extension}, we are free to remove $\1_k$ from the function set $\cF$ without changing the $\mathrm{QB}_k$-norm. Thus we are considering $\cA' = \cA/\{\1_k\}$ instead.
\begin{lemma}[Optimal Pauli encoding]\label{lem:pauli-optimal-encoding}
    For the Pauli estimation task $\cF = \{\tr(P\rho)\}_{P \in \cA}$ with $F_P = P$, the optimal dual variables in Eq.~\eqref{eq:dual_unbiased} can always be set to the form $Z^* = {\1}/{2^n}, Y_P^* = \lambda_P\, P, \forall\, P \in \cA$.
\end{lemma}
\begin{proof}
Consider any feasible dual solution $(Z, \{Y_P\})$, expanding in the $n$-qubit Pauli basis $\{Q\}_{Q \in \{I,X,Y,Z\}^{\ox n}}$ as $Z = \sum_Q z_Q\, Q$ and $Y_P = \sum_Q y_{P,Q}\, Q$.
Let $g(Q,P) \in \{0,1\}$ denote the commutation indicator such that $QPQ = (-1)^{g(Q,P)} P$.
We construct $(\bar{Z}, \{\bar{Y}_P\})$ via Pauli twirling:
\begin{align}
    \bar{Z} = \frac{1}{4^n}\sum_{Q \in \{I,X,Y,Z\}^{\ox n}} Q Z Q, \quad \bar{Y}_P = \frac{1}{4^n}\sum_{Q \in \{I,X,Y,Z\}^{\ox n}} (-1)^{g(Q,P)} Q Y_P Q, \quad \forall P \in \cA,\label{eq:C4}
\end{align}
Evaluating the twirled operators gives:
\begin{align}
    \bar{Z} &= \sum_R z_R\, R \cdot \frac{1}{4^n}\sum_Q (-1)^{g(Q,R)} = \1/2^n,\label{eq:C5}
\\
    \bar{Y}_P &= \sum_R y_{P,R}\, R \cdot \frac{1}{4^n}\sum_Q (-1)^{g(Q,P) + g(Q,R)} = y_{P,P}\, P,\label{eq:C6}
\end{align}
where we use the Pauli group orthogonality identity $\frac{1}{4^n}\sum_{Q \in \{I,X,Y,Z\}^{\ox n}} (-1)^{g(Q,P) + g(Q,R)} = \delta_{P,R}$, since each non-identity Pauli anticommutes with exactly half of the group.
The constant $z_{I}$ is determined by the restriction $\tr Z = 1$.

We then show that the dual feasibility condition is still satisfied by the twirled operators.
Suppose $(Z, \{Y_P\})$ is dual-feasible, i.e., $\forall c, \, Z + \sum_{P \in \cA} c_P\, Y_P \succeq 0$.
For any fixed Pauli $Q \in \{I,X,Y,Z\}^{\ox n}$, the map $c_P \mapsto c_P(-1)^{g(Q,P)}$ is a bijection on $\{-1,1\}^{|\cA|}$.
So $\forall c, \,Z + \sum_{P} c_P(-1)^{g(Q,P)}\, Y_P \succeq 0$.
Conjugating both sides by $Q$, which preserves positive semidefiniteness, gives $Q Z Q + \sum_{P} c_P(-1)^{g(Q,P)}\, Q Y_P Q \succeq 0$.
Averaging over all $4^n$ Pauli operators $Q$, the positive semidefinite cone is preserved by convexity:
\begin{align}
    {\frac{1}{4^n}\sum_Q Q Z Q} + \sum_{P \in \cA} c_P \frac{1}{4^n}\sum_Q (-1)^{g(Q,P)} Q Y_P Q= \bar{Z} + \sum_{P \in \cA} c_P  \bar{Y}_P \succeq 0.\label{eq:pauli-twirl-result}
\end{align}
Thus, the symmetrized solution $(\bar{Z}, \{\bar{Y}_P\}) = (z_I\,\1,\, \{y_{P,P}\,P\})$ is dual-feasible.

To show that the dual objective is not decreased by the twirled operators, we first note that the dual objective $\sum_P \tr(Y_P\, F_P) = \sum_{P \in \cA} y_{P,P}$ for $F_P =P$ depends only on the diagonal coefficients.
All cross-components $y_{P,Q}$ with $Q \neq P$ make no contribution to the objective.
Therefore, the dual objective is not decreased by the twirled operators.
\end{proof}
Based on Lemma~\ref{lem:pauli-optimal-encoding}, the optimal payoff reduces to the scalar program~\eqref{eq:pauli-dual-reduced}
\begin{align}
     {
    d^* = \max_{\{\lambda_P\}} \; 2^n \sum_{P \in \cA} \lambda_P \qquad \text{s.t.} \quad \frac{\1_n}{2^n} + \sum_{P \in \cA} c_P\, \lambda_P\, P \succeq 0, \quad \forall\, c \in \{-1,1\}^{|\cA|}
    }\label{eq:pauli-dual-reduced}
\end{align}
This scalar program relies heavily on the structure of $\cA \subset \{\1,X,Y,Z\}^{\ox n}$.
When $\cA = \{I,X,Y,Z\}^{\ox n}/\{\1_n\} :=  \mathcal{P}_n^{\times}$, all $\lambda_P$ can be set equal to a common value $\lambda$ since $\cA$ is invariant under a transitive subgroup of the Clifford group.

\begin{lemma}[Clifford reduction for full Pauli set]
    \label{lem:clifford-reduction-full}
    For $\mathcal{A}$ closed under the Clifford group, after Lemma~\ref{lem:pauli-optimal-encoding} yields
    $Z = 2^{-n}\1_{n}$ and $Y_P = \lambda_P\,P$, the
    optimal $\{\lambda_P\}$ can without loss of generality
    be taken to have a common value $\lambda$
    for all $P \neq I$.
\end{lemma}

\begin{proof}
    Let $(Z, \{Y_P = \lambda_P P\}_{P \in \mathcal{P}_n})$
    be any dual-feasible solution after Pauli twirling. For
    any Clifford unitary $U \in \mathrm{Cl}(n)$, define the
    conjugation map $\sigma_U$ on Paulis by
    $UPU^\dagger = \epsilon(U,P)\,\sigma_U(P)$, where
    $\epsilon(U,P) \in \{\pm 1\}$ is the sign and
    $\sigma_U : \mathcal{A} \to \mathcal{A}$ is the
    symplectic permutation within $\mathcal{A}$ since it is closed under Clifford operations.
    Define:
    \begin{align}
        Z_{U} := UZU^\dagger = 2^{-n}\1_{n}, \qquad
        Y_{U,Q} = \lambda_{\sigma_U^{-1}(Q)}\,Q, \, \forall Q \in \mathcal{A}.\label{eq:C9}
\end{align}
    The operator $Y_{U,Q}$ is just a coefficient shift of the original dual variables.
     We will show that $(Z_{U}, \{Y_{U,Q}\})$ is a feasible solution with the same objective.
    Since the original solution
    is feasible for \emph{all} sign vectors, $(Z_{U}, \{Y_{U,Q}\})$
    is equally feasible.
  For the objective preservation, we use the fact that $Q = \sigma_U(P)$ is a bijection for all $P \in \mathcal{A}$ to conclude
    \begin{align}
        \sum_{Q \in \mathcal{A}} \tr(Y_{U,Q}\,Q)
        = \sum_{Q \in \mathcal{A}} \lambda_{\sigma_U^{-1}(Q)}\,\tr(Q^2)
        = 2^n \sum_{Q \in \mathcal{A}} \lambda_{\sigma_U^{-1}(Q)}
        = 2^n \sum_{P \in \mathcal{A}} \lambda_P.\label{eq:C10}
\end{align}
    By convexity of the PSD cone, we take the Clifford-averaged solution:
    \begin{align}
        \bar{Y}_Q := \frac{1}{|\mathrm{Cl}(n)|} \sum_{U \in \mathrm{Cl}(n)} Y_{U,Q}
        = \left(\frac{1}{|\mathrm{Cl}(n)|} \sum_{U \in \mathrm{Cl}(n)} \lambda_{\sigma_U^{-1}(Q)}\right) Q
        =: \bar{\lambda}_Q\,Q.\label{eq:C11}
\end{align}
    such that $(\bar{Z}, \{\bar{Y}_Q\})$ is feasible with the same objective as $(Z, \{Y_P\})$.
    For $Q =  \1_n$, since $\sigma_U( \1_n) = \1_n$ for all $U$, we have $\bar{\lambda}_I = \lambda_I$.
    For $Q \neq  \1_n$, by the orbit-stabilizer theorem~\cite{dankert2009exact},
    $\sigma_U^{-1}(Q)$ hits each non-identity Pauli the same number of times.
    Therefore:
    \begin{align}
        \bar{\lambda}_Q
        = \frac{1}{|\mathrm{Cl}(n)|} \sum_U
        \lambda_{\sigma_U^{-1}(Q)}
        = \frac{1}{4^n - 1} \sum_{P \neq  \1_n} \lambda_P
        =: \lambda, \qquad \forall\, Q \neq  \1_n.\label{eq:C12}
\end{align}
    This completes the proof that $\bar{\lambda}_Q = \lambda$ for all $Q \neq I$.
\end{proof}

Based on Lemma~\ref{lem:clifford-reduction-full}, the dual conic convex program
reduces to a two-parameter scalar program over $\lambda$:
\begin{align}
     {
    d^*(\cP_n) = \max_{\lambda}\;
    2^n \cdot (4^n - 1)\lambda
    \qquad \text{s.t.} \quad
    \frac{\1_n}{2^n}
    + \lambda \sum_{P \neq I} c_P\, P \succeq 0,
    \quad \forall\, c \in \{-1,1\}^{4^n -1}.
    }\label{eq:pauli-dual-reduced-AllPauli}
\end{align}
The constraint becomes ${2^{-n}} \;\geq\; |\lambda|\,\max_{c} \bigl\|\textstyle\sum_{P \neq I} c_P P\bigr\|_\infty =: |\lambda|\,\Sigma$, where $\Sigma$ is the maximum signal strength for Pauli estimation.

\begin{lemma}\label{lem:pauli-max-signal}
         $\Sigma := \max_{c} \bigl\|\textstyle\sum_{P \neq I} c_P P\bigr\|_\infty$ satisfies:
        \begin{align}
            (1 + \sqrt{3})^n -1 \leq \Sigma \leq
            \sqrt{(2^n - 1)(4^n - 1)} < 2^{3n/2}.\label{eq:C14}
\end{align}
    \end{lemma}
    \begin{proof}
        For the upper bound, we write
        $ \Sigma = \max_{|\psi\rangle} \sum_{P \neq I} |\langle\psi|P|\psi\rangle|
        = \max_{|\psi\rangle}\left(\sum_{P \neq I}|\langle\psi|P|\psi\rangle|\right).$
        By Cauchy--Schwarz we get $ \sum_{P \neq I} |\langle\psi|P|\psi\rangle| \leq \sqrt{\sum_{P \neq I} \langle\psi|P|\psi\rangle^2} \cdot \sqrt{4^n - 1}$.
        The purity identity
        $2^{-n}\sum_P \langle\psi|P|\psi\rangle^2 = 1$
        for any pure state gives
        $\sum_{P \neq I} \langle\psi|P|\psi\rangle^2
        = 2^n - 1$, yielding
        $\Sigma \leq \sqrt{(2^n-1)(4^n-1)} < 2^{3n/2}$.
        For the lower bound, the magic state~\cite{seddon2021quantifying}
        $|\Psi\rangle = |H\rangle^{\otimes n}$ with
        $|\langle X\rangle| = |\langle Y\rangle|
        = |\langle Z\rangle| = 1/\sqrt{3}$ gives:
        $    \sum_{P \neq I} |\langle\Psi|P|\Psi\rangle|
        = \prod_{i=1}^n (1 + \sqrt{3}) -1
        = (1+\sqrt{3})^n -1.$
    \end{proof}

Substituting $\lambda = 2^{-n}/\Sigma$
into the objective gives:
\begin{align}
    d^* &=
    2^n\!\left[
    (4^n - 1)\,\frac{2^{-n} }{\Sigma}
    \right] =
    \frac{4^n - 1}{\Sigma}.\label{eq:C15}
\end{align}
Combining this with
Lemma~\ref{lem:pauli-max-signal} and strong duality
(Lemma~\ref{lem:dual-robustness}), the
$\mathrm{QB}_k$-norm for the full Pauli set
$\cF = \{I,X,Y,Z\}^{\otimes n}$ is bounded by:
\begin{align}
    \frac{4^n - 1}{2^{3n/2}} < \|\cF\|_1
    = \frac{4^n - 1}{\Sigma}
    \leq \frac{4^n}{(1+\sqrt{3})^n -1}
    \approx 1.464^n,\label{eq:pauli-s-star-bounds}
\end{align}
giving the variance bound $\left(2^{n/2} - 2^{-3n/2}\right)^2= \Omega(2^n) \leq \|\cF\|_1^2 \leq 2.143^n$.
The lower bound $2^n$ matches the information-theoretic
bound of~\cite{chen2024optimal} (which is based on the
quantity $\delta_A^{-1} = 4^n/2^n = 2^n$), while the
upper bound $2.143^n$ reflects the tighter constraint
imposed by the $\mathrm{QB}_k$-norm framework on single-POVM
estimation strategies.

\appsubsubsection{The $N$-replica memory restricted strategy}
The above discussion only uses identical POVMs for each round and excludes the $N$-replica memory restricted strategy.
We thus consider the $N$-replica adaptive memory restricted strategy~\cite{chen2024optimal} with POVMs restricted to the convex set $\mathcal{C}_{\mathrm{mem}}$ defined in Eq.~\eqref{eq:block-positive-cone-mem}. That is, we evaluate $\|\cF\|_{N,{\mathrm{mem}}}$ defined in Eq.\eqref{eq:dual_unbiased}.
The N-copy linearization is written as $\{F^{(N)}_P = \1^{\otimes (N-1)} \odot P\}_{P \in \cA}$.
Following the same reasoning as in the last section, we conclude that $Z,Y_P$ actually can be represented as a linear combination of Clifford commutant $\mathrm{Com}(\mathrm{Cl}(n)^{\ox (N)}) : \{A \in \cH^{\otimes N} : [A, C^{\ox N}] = 0, \forall C \in \mathrm{Cl}(n)\}$.

\begin{lemma}
    For the $N$-replica memory restricted strategy, the dual variables $Z,Y_P$ can be taken as
    \begin{align}
        Z &\in \mathrm{Com}(\mathrm{Cl}(n)^{\ox (N)}), \quad Y_P = \frac{1}{d} \tr_0\left[(P \otimes \1^{\otimes N}) \Gamma\right],\label{eq:C17}
\end{align}
    where $\Gamma \in \mathrm{Com}(\mathrm{Cl}(n)^{\ox (N+1)}), \tr(\Gamma) = 0$, $d = 2^n$, and $\tr_0$ denotes the partial trace over the first replica of qubits.
\end{lemma}
\begin{proof}
    Following the same reasoning as in the proof of Lemma~\ref{lem:pauli-optimal-encoding}, we can show that $Z,Y_P$ can be taken to satisfy the covariant condition:
    \begin{align}
         U^{\otimes N} Z U^{\dagger,\otimes N} = Z , \quad U^{\otimes N} Y_P U^{\dagger,\otimes N} = \varepsilon(U,P) Y_{\sigma_U(P)}.\label{eq:C18}
\end{align}
    This immediately proves $Z \in \mathrm{Com}(\mathrm{Cl}(n)^{\ox (N)})$. For the other one, we construct an extended operator in the space $\cH^{\ox (N+1)}$ as $\Gamma = \sum_{P \in  \mathcal{P}_n^{\times}} P \otimes Y_P$. By applying the $N+1$-fold Clifford $U^{\otimes (N+1)}$, we get
    \begin{align}
        U^{\otimes (N+1)} \Gamma U^{\dagger,\otimes (N+1)} &= \sum_{P \in  \mathcal{P}_n^{\times}} U^{\otimes (N+1)} P \otimes Y_P U^{\dagger,\otimes (N+1)} \notag
\\
        &= \sum_{P \in  \mathcal{P}_n^{\times}} \varepsilon^2(U,P)\sigma_U(P) \otimes  Y_{\sigma_U(P)} = \sum_{P \in  \mathcal{P}_n^{\times}} P \otimes Y_{P}\label{eq:C19}
\end{align}
    This implies that $\Gamma \in \mathrm{Com}(\mathrm{Cl}(n)^{\ox (N+1)})$, and thus $Y_P ={2^{-n}} \tr_0\left[(P \otimes \1^{\otimes N}) \Gamma\right]$. For the converse, $\Gamma \in \mathrm{Com}(\mathrm{Cl}(n)^{\ox (N+1)}) \Rightarrow Y_P ={2^{-n}} \tr_0\left[(P \otimes \1^{\otimes N}) \Gamma\right]$ is covariant, as one can verify by expanding $\Gamma$ on $\sum_{Q} Q \otimes Y_Q$.
\end{proof}

Based on the extension of $Y_P$, we define
\begin{align}
    F_{\mathcal{P}} := 2^{-n} \sum_{P \in  \mathcal{P}_n^{\times}} P \otimes F^{(N)}_P = 2^{-n} \sum_{P \in  \mathcal{P}_n^{\times}} P \otimes \1^{\otimes (N-1)} \odot P = \frac{1}{N}\sum_{i = 1}^N W_{0,i} - \frac{1}{d}\1^{\otimes (N+1)},\label{eq:C20}
\end{align}
where we use the fact $W_{i,j} =d^{-1} \sum_{P \in \mathcal{P}(n)}{P_i \otimes P_j}$ for $d = 2^n$ and $W_{i,j}$ is the n-qubit SWAP operator on the $i$-th and $j$-th replicas.
The target function is then
\begin{align}
   \sum_{P \in  \mathcal{P}_n^{\times}} \tr(Y_{{P}} F^{(N)}_P)  = \sum_{P \in  \mathcal{P}_n^{\times}}2^{-n} \tr_0 \tr_{1 \cdots N} \left[(P \otimes \1^{\otimes N}) \cdot \Gamma \cdot (\1 \otimes F^{(N)}_P)\right]  = \tr(F_{\mathcal{P}} \cdot \Gamma) = \frac{1}{N}\sum_{i = 1}^N \tr(\Gamma \cdot W_{0,i}),\label{eq:C21}
\end{align}
where we use the fact $\tr(\Gamma) = 0$.
Based on the lifted representation, the original dual problem can be rewritten
as an optimization over two Clifford-commutant variables $Z$ and $\Gamma$.
Define
\begin{align}
    \mathfrak A_N:=\mathrm{Com}(\mathrm{Cl}(n)^{\otimes N})\cap \mathrm{Com}(S_N), \;
     \mathfrak B_{N+1}:=\left\{ \Gamma\in \mathrm{Com}(\mathrm{Cl}(n)^{\otimes(N+1)}): [\1 \otimes W_\pi,\Gamma]=0,\ \forall \pi\in S_N \right\}.\label{eq:C22}
\end{align}
For $c\in\{\pm1\}^{|\mathcal A|}$, let $C_c:=\sum_{P\in\mathcal A}c_P P$ with $\mathcal L_c(\Gamma):=\frac1d\tr_0[(C_c\otimes \1^{\otimes N})\Gamma]$.
Then the dual is rigorously equivalent to
\begin{align}
     {
\begin{aligned}
\|\mathcal F\|_{N,{\mathrm{mem}}}
=
\max_{Z,\Gamma}\quad&
\tr(\Gamma F_{\mathcal P})\\
\mathrm{s.t.}\quad&
Z-\mathcal L_c(\Gamma)\in
\mathcal K_{{\mathrm{mem}}}^*,
\quad \forall c,\\
&\tr Z=1,\\
&Z\in\mathfrak A_N,\quad \Gamma\in\mathfrak B_{N+1}.
\end{aligned}
}\label{eq:n-replica-Pauli-dual}
\end{align}
As proved in Theorem 57 of~\cite{bittel2025}, the $N$-fold Clifford commutant is the algebra generated by $S_N \cup \{\Omega_4, \Omega_6\}$ and $\Omega_N$ if $N$ is divisible by 4, where $\Omega_k := \sum_{P \in \mathcal{P}(n)} P^{\otimes k}$ is a simple Pauli monomial supported on $k$ replicas of $n$-qubit states.
However, the direct computation of Eq.\eqref{eq:n-replica-Pauli-dual} is still challenging since the Clifford-Weingarten calculus for $N \geq 5$ is not well-developed.
Here we provide a lower bound by restricting $\Gamma$ and $Z$ to be a linear combination of $\Omega_k$. To be specific, we set up
\begin{align}
    \Gamma =\sum_{l \leq N,l \text{~is odd} } a_l G^{(N+1)}_{l}, \quad G^{(N+1)}_l :=\binom{N}{l}^{-1} \sum_{\substack{S \subseteq [N], |S| = l}} \sum_{Q \in \mathcal{P}_n^{\times}}Q \otimes Q^{\otimes S} \otimes \1^{\otimes \bar{S}},\label{eq:feasible_Gamma}
\\
    Z=\frac{\1^{\otimes N}}{d^{N}} + \sum_{\substack{m \leq N\\ m \text{~is even}}}^N z_m H_{m}^{(N)}, \quad H_{m}^{(N)} := \binom{N}{m}^{-1} \sum_{\substack{S \subseteq [N], |S| = m}} \sum_{Q \in \mathcal{P}_n^{\times}} Q^{\otimes S} \otimes \1^{\otimes \bar{S}},\label{eq:feasible_Z}
\end{align}
where $G^{(N+1)}_{l}$ for odd $l$ and $H_{m}^{(N)}$ for even $m$ are both $N+1$-fold and $N$-fold Clifford commutants since they are primitive Pauli monomials (see Definition 28 in \cite{bittel2025}) that are supported on subsets $S$ of $N+1$ and $N$ qubits, respectively.
Under this restriction, we note that
\begin{align}
    \tr(G^{(N+1)}_l F_{\mathcal{P}}) &= N^{-1}\binom{N}{l}^{-1} \sum_{i = 1}^N \sum_{\substack{S \subseteq [N], |S| = l}} \sum_{Q \in \mathcal{P}_n^{\times}} \tr\left[(Q \otimes Q^{\otimes S} \otimes \1^{\otimes \bar{S}}) \cdot S_{0,i}\right]\label{eq:C26}
\\
    &= N^{-1}\binom{N}{l}^{-1} \sum_{i = 1}^N \sum_{\substack{S \subseteq [N], |S| = l}} \sum_{Q \in \mathcal{P}_n^{\times}} \delta( i \in S) \delta(S/\{i\} = \varnothing)d \cdot d^{N-1}\label{eq:C27}
\\
    &= N^{-1} \binom{N}{l}^{-1} \sum_{i = 1}^N\sum_{\substack{S \subseteq [N], |S| = l}}
    (d^2 - 1) \delta_{S=\{i\}}d^N =\delta_{l,1} (d^2 - 1) \frac{d^N}{N},\label{eq:C28}
\end{align}
where in the second equality we use the property that $\tr Q = 0$ for any $Q \in \mathcal{P}_n^{\times}$ and $\tr(S_{i,j}\cdot(A_i\otimes B_j)) = \tr(A_i\otimes B_j)$. This means that our target function only depends on $a_1$ and ignores the other coefficients $a_l$ for $l \neq 1$.
For the restriction, we note that
\begin{align}
     \sum_{P \in \mathcal{P}_n^{\times}}c_PY_P &= \frac{1}{d}\sum_{P \in \mathcal{P}_n^{\times}}c_P\mathrm{tr}_0\left[ (P \otimes \1^{\otimes N}) \Gamma\right] =\sum_{l \leq N,l \text{~is odd} } a_lD_{l,c}^{(N)} ,\label{eq:C29}
\\
     D_{l,c}^{(N)} &:= \sum_{\substack{l\leq N\\ l \text{~is odd}}}^N \binom{N}{l}^{-1}\sum_{S \subseteq [N], |S| = l} \sum_{P \in \mathcal{P}_n^{\times}}c_P P^{\otimes S} \otimes \1^{\otimes \bar{S}} .\label{eq:C30}
\end{align}
This reduces the dual problem of Eq.\eqref{eq:n-replica-Pauli-dual} to
\begin{align}
     {
    \begin{aligned}
        d = \max_{\{a_l\},\{z_m\}}  & (d^2 - 1)\frac{d^N}{N} a_1 \\
    \mathrm{s.t.}\quad & \frac{\1^{\otimes N}}{d^{N}} + \sum_{\substack{m = 1\\ m \text{~is odd}}}^N z_m H_{m}^{(N)} -  \sum_{\substack{m = 2\\ m \text{~is even}}}^N a_m D_{m,c}^{(N)} \in \mathcal{K}_{{\mathrm{mem}}}^*,
    \end{aligned}
    }\label{eq:C31}
\end{align}
If we further take $a_1 >0, a_{l >1} = 0, z_m = 0$ and restrict ourselves to the memory-free strategy $\mathcal{C}_{\mathrm{mem}}:= \mathrm{conv}(\Pi_1 \otimes \cdots \otimes \Pi_{N}: \Pi_i \succeq 0)$. Based on Eq.~\eqref{eq:site-product-cone}, $W \in \mathcal{K}_{\mathcal{C}}^*$ is equivalent to $\langle \Phi | W | \Phi \rangle \geq 0$ for arbitrary tensor product state $|\Phi\rangle = |\psi_1\rangle \otimes \cdots \otimes |\psi_N\rangle$.
Then the problem is reduced to a formula similar to that for the $N = 1$ case:
\begin{align}
     {
    \begin{aligned}
        d = \max_{a_1} \quad & (d^2 - 1)\frac{d^N}{N} a_1 \\
    \mathrm{s.t.}\quad & \frac{1}{d^N} \geq a_1 \sup_{\psi_1,\cdots,\psi_N} \sum_{P \in \mathcal{P}_n^{\times}} \left| \frac{1}{N} \sum_{i =1}^N \langle \psi_i | P | \psi_i \rangle\right|.
    \end{aligned}
    }\label{eq:C32}
\end{align}
Using Lemma~\ref{lem:pauli-max-signal} and convexity, we conclude that
\begin{align}
    \sup_{\psi_1,\cdots,\psi_N} \sum_{P \in \mathcal{P}_n^{\times}} \left| \frac{1}{N} \sum_{i =1}^N \langle \psi_i | P | \psi_i \rangle\right| = \sup_{\psi} \sum_{P \in \mathcal{P}_n^{\times}} \left| \langle \psi | P | \psi \rangle\right| = \Sigma.\label{eq:C33}
\end{align}
Therefore, the optimal value satisfies:
\begin{align}
    \|\cF\|_{k,{\mathrm{mem}}} \geq \frac{d^2 - 1}{N \Sigma} \geq \frac{d^2 - 1}{N \sqrt{(d^2 - 1)(d - 1)}}=\frac{\sqrt{d + 1}}{N}.\label{eq:C34}
\end{align}
This means that any constant-envelope unbiased memory-free ($c = 1$) $N$-copy estimator must
satisfy $N=\Omega(\sqrt d)=\Omega(2^{n/2})$.
This reproduces an exponential obstruction for full Pauli shadow tomography,
although it does not match the tight $\Theta(d)=\Theta(2^n)$ memory-free
sample complexity~\cite{chen2022exponential,chen2024optimal}.
The loss arises from the restricted feasible-point ansatz in
Eq.~\eqref{eq:feasible_Gamma} and Eq.~\eqref{eq:feasible_Z} and optimizing the full
$N$-replica dual program in Eq.\eqref{eq:n-replica-Pauli-dual} may yield a stronger bound, which we leave for
future work.

\appsubsection{Examples for purity estimation}
For the purity estimation task $\cF := \{\tr(\rho^2)\}$, the 2-replica symmetric linearization is the permutation operator $F = W_{(12)}$ on $2$-replicas of $n$-qubit space $\cH^{\ox 2}$ with $d := \dim(\cH) = 2^n$.
We consider the POVM restriction $\mathcal{C}_{\mathrm{mem}}$, which is formally defined in Eq.~\eqref{eq:block-positive-cone-mem} and gives the dual cone constraints in Eq.~\eqref{eq:block-positive-cone}.
\appsubsubsection{The 2-replica memory restricted strategy}

For $\|\cF\|_{2,{\mathrm{mem}}}$,
 the POVM operators are restricted to the convex set $\mathcal{C} = \mathrm{conv}\{\Pi_j = \Pi_j^{(1)} \otimes \Pi_j^{(2)}\}$, where $\Pi_j^{(1)}, \Pi_j^{(2)} \in \mathrm{Herm}(\cH)$ are $n$-qubit POVM elements on the first and second replicas, respectively.
The dual cone of $\mathcal{C}$ is therefore
\begin{align}
    \mathcal{K}_{\mathrm{mem}}^{*} = \bigl\{O \in \mathrm{Herm}(\cH^{\ox 2}) \;\big|\; \tr(O\, A \ox B) \geq 0, \;\forall\, A, B \succeq 0\bigr\},\label{eq:C35}
\end{align}
which is the cone of \emph{block-positive} operators~\cite{skowronek2009cones}.
Equivalently, it is the cone of operators with non-negative expectation on all product states $\ket{\phi}\ket{\psi}$:
\begin{align}
    \bra{\phi}\bra{\psi}O\ket{\phi}\ket{\psi} \geq 0, \;\forall\, \ket{\phi},\ket{\psi} \in \cH.\label{eq:C36}
\end{align}
In quantum information theory, the duality between the cone of separable unnormalized states and the cone of block-positive operators forms the foundational mathematical structure for entanglement witnesses~\cite{horodecki2001separability, guhne2009entanglement}.
Since $|\cA| = 1$, the sign vector in Lemma~\ref{lem:dual-robustness} is $c \in \{-1,+1\}$.
The dual problem then reduces to
\begin{align}
     {
    \begin{aligned}
    d^* = \max_{Z, Y} \tr(Y W_{(12)}) \quad \text{s.t.~} Z \pm Y \in \mathcal{K}_{\mathcal{C}}^*,\quad  \tr(Z) = 1, \quad Y, Z \in \mathrm{Herm}(\cH^{\ox k})
    \end{aligned}
    }\label{eq:purity-dual-problem}
\end{align}
Then, using the symmetry of the target function $\tr(Y W_{(12)})$, it can be shown that without loss of generality, we can take $Y,Z$ as a linear combination of the identity and the SWAP operators.
\begin{lemma}[Optimal purity encoding under separable measurements]\label{lem:purity-optimal-encoding}
    For the purity estimation task $\cF = \{\tr(\rho^2)\}$ with separable 2-copy measurements, the dual variables $\{Z,Y\}$ of Eq.~\eqref{eq:purity-dual-problem} can be taken in the form $Z = \alpha\, \1_{2n} + \beta\, W_{(12)}, \, Y = a\, \1_{2n} + b\, W_{(12)}$ for some $\alpha, \beta, a, b \in \mathbb{R}$.
\end{lemma}
\begin{proof}
We apply a unitary twirling argument analogous to the Pauli twirling in Lemma~\ref{lem:pauli-optimal-encoding}, replacing the discrete Pauli group with the continuous unitary group.
For any unitary $U \in \mathrm{U}(d)$, the conjugation $A \mapsto (U \ox U) A (U^\dagger \ox U^\dagger)$ preserves:
(i) the block-positive cone $\mathcal{K}_\mathcal{C}^*$ (since $U^\dagger A U$ and $U^\dagger B U$ remain PSD for $A, B \succeq 0$),
(ii) the dual objective $\tr(Y \cdot W_{(12)})$ since the SWAP operator satisfies $[U^{\ox 2}, W_{(12)}] = 0$,
(iii) the trace restriction $\tr(Z) = 1$ since $\tr\left[(U^{\ox 2}) Z (U^{\dagger})^{\ox 2}\right] = \tr(Z)$.
Therefore, given any feasible $(Z, Y)$, the Haar-averaged solution
\begin{align}
    \bar{Z} = \int_{\mathrm{Haar}} d U  (U^{\ox 2}) Z (U^{\dagger})^{\ox 2}, \quad
     \bar{Y} =  \int_{\mathrm{Haar}} d U  (U^{\ox 2}) Y (U^{\dagger})^{\ox 2}.\label{eq:C38}
\end{align}
is equally feasible with the same objective.
By Schur--Weyl duality for two copies, the commutant of $\{U^{\ox 2} : U \in \mathrm{U}(d)\}$ is $\mathrm{span}\{\1_{2n}, W_2\}$.
Hence, without loss of generality, $Z = \alpha \1_{2n} + \beta\, W_{(12)}, Y = a \1_{2n} + b\, W_{(12)}$ can be chosen.
\end{proof}
By Lemma~\ref{lem:purity-optimal-encoding}, the dual conic convex program~\eqref{eq:purity-dual-problem} reduces to a scalar optimization over four real parameters $(\alpha, \beta, a, b)$. We now evaluate each constraint explicitly.
The dual variable must satisfy $Z \pm Y = (\alpha \pm a)\,\1_{2n} + (\beta \pm b)\,W_2 \in \mathcal{K}_{\mathcal{C}}^*$. This is equivalent to $\bra{\phi,\psi} (Z \pm Y) \ket{\phi,\psi} \geq 0$ for all product states $\ket{\phi}\ket{\psi}$. Using $\bra{\phi,\psi}W_2\ket{\phi,\psi} = |\braket{\phi}{\psi}|^2$, we define $t := |\braket{\phi}{\psi}| \in [0,1]$ and obtain $(\alpha \pm a) + (\beta \pm b)\,t^2 \geq 0, \forall\, t \in [0,1]$. Since each expression is affine in $t^2 \in [0,1]$, it suffices to enforce non-negativity at the endpoints $t = 0$ and $t = 1$, yielding four scalar constraints:
\begin{align*}
    \text{(C1):~} \alpha + a \geq 0, \quad \text{(C2):~} \alpha + a + \beta + b \geq 0,\\
    \text{(C3):~} \alpha - a \geq 0, \quad \text{(C4):~} \alpha - a + \beta - b \geq 0.
\end{align*}
Using $\tr(\1_{2n}) = d^2$ and $\tr(W_2) = d$, the trace constraint is $\text{(C5):~}\tr(Z) = \alpha\,d^2 + \beta\,d = 1$ and the objective is given by
\begin{align}
    \tr(Y W_2) = a\,\tr(W_2) + b\,\tr(W_2^2) = a\,d + b\,d^2.\label{eq:C39}
\end{align}
Combining all ingredients, the dual conic convex program reduces to:
\begin{equation}
     {
    \begin{aligned}
        d^*&  = \max_{\alpha,\beta,a,b} \quad a\,d + b\,d^2,\\
        \text{s.t.} \quad & \alpha + a \geq 0 \quad \text{(C1)}, \quad \alpha + a + \beta + b \geq 0 \quad \text{(C2)}, \\
        & \alpha - a \geq 0 \quad \text{(C3)}, \quad \alpha - a + \beta - b \geq 0 \quad \text{(C4)}, \\
        & \alpha\,d^2 + \beta\,d = 1 \quad \text{(C5)}.
    \end{aligned}
    }\label{eq:scalar_dual_SDP_purity}
\end{equation}
For fixed $\alpha,\beta, a$, the objective is increasing in $b$.
Thus, $b$ should saturate its upper bound from (C4), i.e., $b=\alpha-a+\beta$.
Substituting this into the objective gives
\begin{align}
ad+bd^2=ad+d^2(\alpha-a+\beta)=d^2(\alpha+\beta)-a(d^2-d).\label{eq:C41}
\end{align}
Since $d \ge 2$, the coefficient of $a$ is negative. Hence $a$ should be minimized subject to (C1), namely $a=-\alpha$.
This completely fixes $a,b$ as functions of $\alpha, \beta$ as $a = -\alpha, b = 2\alpha + \beta$.
Using the trace constraint (C5), we get $\beta=d^{-1}{(1-\alpha d^2)}$.
Therefore the target function is given by a function of $\alpha$ as
\begin{align}
ad+bd^2 = -\alpha d+(2\alpha+\beta)d^2 =d-\alpha d(d-1)^2,\label{eq:C42}
\end{align}
which is decreasing in $\alpha$.
The remaining constraints (C2) and (C3) imply
$\alpha\ge 0, \alpha+\beta\ge 0$,
so the optimum is attained at $\alpha=0$.
This gives $\beta = 1/d$, $a = 0$, $b = 1/d$ and $d^* = d$.
The optimal dual variables are $Z^* = Y^* =d^{-1}\,W_2$.
It can be verified that the constraints $Z + Y = (2d^{-1})\,W_2$ and $Z - Y = 0$ are block-positive since $\bra{\phi}\bra{\psi}W_2\ket{\phi}\ket{\psi} = |\braket{\phi}{\psi}|^2 \geq 0$.
Since the trivial POVM $\{\1^{\ox 2}\}$ lies in the interior of the separable cone $\mathcal{C}$,
following from Lemma~\ref{lem:dual-robustness}, strong duality and the equality $\|\cF\|_{2,{\mathrm{mem}}} = d^* = 2^n$ are satisfied.

Using Hoeffding's inequality, this $\mathrm{QB}_2$-norm gives a sample complexity scaling $O(d^2/\varepsilon^2)$,
which is higher than the real lower bound $2^{n/2}$~\cite{chen2024optimal,gong2024sample,chen2022exponential}.
This is expected since $\mathrm{QB}_2$-norm considers only independent two-copy separable block estimators.
To show the real lower bound, we need to consider the memory-free $N$-replica strategy setting, i.e., $\|\cF\|_{N,{\mathrm{mem}}}$.

\appsubsubsection{The $N$-replica memory restricted strategy}

We now consider the calculation of $\|\cF\|_{N,{\mathrm{mem}}}$.
The unique symmetric linearization of $\tr(\rho^2)$ on $\cH^{\ox N}$ is the symmetrized pairwise SWAP:
\begin{align}
    F^{(N)} := \binom{N}{2}^{-1} \sum_{1 \leq i < j \leq N} W_{i,j} = K^{-1} C_\tau, \text{~where~}  K = \binom{N}{2}, \tau := (2, 1^{N-2}).\label{eq:purity-N-target}
\end{align}
Here $C_\tau := \sum_{\pi \in \mathcal{C}_\tau} W_\pi$ sums over the permutation group $S_N$'s conjugacy class $\mathcal{C}_\tau$, where $\tau \vdash N$ is a partition of $N$ and $W_{\pi} \ket{\psi_1} \otimes \cdots \otimes \ket{\psi_N} = \ket{\psi_{\pi(1)}} \otimes \cdots \otimes \ket{\psi_{\pi(N)}}$ defines the representation of $S_N$ on $N$ replicas of the state.
The POVM constraint is the single-copy tensor product cone $\mathcal{C}^{(N)}_{\mathrm{loc}} = \mathrm{conv}\{\bigotimes_{i=1}^N \Pi_{x_i}: \Pi_{x_i} \succeq 0\}$, with dual cone $\mathcal{K}^{*,(N)}_{\mathrm{loc}}$ consisting of operators block-positive on all $N$-fold product states.
The $N$-replica dual conic convex program reads:
\begin{align}
     {
    d_N^* = \max_{Z,Y} \; \tr(Y\, F^{(N)}) \quad \text{s.t.} \quad Z \pm Y \in \mathcal{K}^{*,(N)}_{\mathrm{loc}}, \quad \tr(Z) \leq 1.
    }\label{eq:purity-N-dual}
\end{align}
We reduce the dual variables to the center of the symmetric group algebra $\mathbb{C}[S_N]$ using two successive symmetrization arguments.
\begin{lemma}[Reduction to Weingarten algebra]\label{lem:purity-N-weingarten}
For $d \geq N$, the optimal dual variables of~\eqref{eq:purity-N-dual} can without loss of generality be taken in the center of $\mathbb{C}[S_N]$:
\begin{align}
    Z = \sum_{\lambda \vdash N} z_\lambda\, C_\lambda, \qquad Y = \sum_{\lambda \vdash N} y_\lambda\, C_\lambda, \quad C_\lambda = \sum_{\pi:\, \mathrm{type}(\pi) = \lambda} W_\pi.\label{eq:N-center-expansion}
\end{align}
where the sum runs over the integer partitions $\lambda$ of $N$ that label the conjugacy classes of the permutation group $S_N$.
\end{lemma}
\begin{proof}
The argument of Lemma~\ref{lem:purity-optimal-encoding} extends to $N$ copies.
For $U \in \mathrm{U}(d)$, the conjugation $A \mapsto U^{\ox N} A\, (U^\dagger)^{\ox N}$ preserves (i) the tensor-product block-positive cone $\mathcal{K}^{*,(N)}_{\mathrm{loc}}$ since $U^\dagger \Pi_{x_i} U \succeq 0$ for each factor, (ii) the target $F^{(N)}$ since $[W_{i,j}, U^{\ox N}] = 0$, and (iii) the trace since $\tr\left[U^{\ox N} Z U^{\ox N}\right] = \tr(Z)$.
The Haar-averaged solution is therefore equally feasible with identical objective value.
By the Weingarten integration formula~\cite{collins2006integration}, for any operator $A$ on $\cH^{\ox N}$:
\begin{align}
    \bar{A} := \int_{\mathrm{Haar}} dU \; U^{\ox N} A\, (U^\dagger)^{\ox N} = \sum_{\sigma, \mu \in S_N} \mathrm{Wg}(\sigma^{-1}\mu, d)\, \tr(W^\dagger_\mu\, A)\, W_\sigma,\label{eq:weingarten-formula}
\end{align}
where $\mathrm{Wg}(\pi, d)$ is the Weingarten function, defined as the inverse of the Gram matrix $G_{\sigma\mu} = d^{c(\sigma^{-1}\mu)}$ with $c(\pi)$ being the number of cycles in the permutation $\pi$ on $\mathbb{C}[S_N]$ such that $\sum_{\mu \in S_N} \mathrm{Wg}(\sigma\mu^{-1}, d) \cdot d^{c(\mu\nu^{-1})} = \delta_{\sigma,\nu}$.
For any $\sigma \in S_N$, the replica permutation $A \mapsto W_\sigma A W_\sigma^{-1}$ also preserves:
(i) the block-positive cone, since permuting tensor factors maps product states to product states,
(ii) the target $\tr(Y F^{(N)})$ since $C_\tau$ is central in $\mathbb{C}[S_N]$, thus $[F^{(N)}, W_\sigma] = 0$ for all $\sigma \in S_N$.
(iii) the trace, since $\tr\left[W_\sigma Z W_\sigma^{-1}\right] = \tr Z$.
Again we average over $S_N$ and get
\begin{align}
    \tilde{A}&:= \frac{1}{N!}\sum_{\pi \in S_N} W_\pi \bar{A} W_\pi^{-1} =
     \sum_{\sigma, \mu \in S_N} \mathrm{Wg}(\sigma^{-1}\mu, d)\, \tr(W^\dagger_\mu\, A)\,\frac{1}{N!}\sum_{\pi \in S_N} W_\pi W_\sigma W_\pi^{-1}\label{eq:center-projection}
\\
     &=\sum_{\sigma \in S_N} \left[\sum_{\mu \in S_N} \mathrm{Wg}(\sigma^{-1}\mu, d)\, \tr(W^\dagger_\mu\, A)\,\frac{1}{N!}\right]\sum_{\pi \in S_N}  W_{\pi\sigma\pi^{-1}} =\sum_{\lambda \vdash N} c_{\lambda} C_\lambda.\label{eq:C48}
\end{align}
Here in the last line we use the fact that $\sum_{\pi \in S_N} W_{\pi\sigma\pi^{-1}} \propto C_{\lambda}$ for $\lambda = \mathrm{type}(\sigma)$.
Since both averaging steps preserve feasibility and the objective, applying the averaging step to $(Z, Y)$ gives a no-worse solution $(Z, Y)$ in the form~\eqref{eq:N-center-expansion}.
\end{proof}

By Lemma~\ref{lem:purity-N-weingarten}, the dual conic convex program reduces to a scalar optimization over $2\,p(N)$ real parameters $\{z_\lambda, y_\lambda\}_{\lambda \vdash N}$, where $p(N)$ is the number of integer partitions of $N$. We now evaluate each ingredient.
Since $\tr(W_\pi) = d^{c(\pi)}$ where $c(\pi)$ is the number of cycles in the permutation $\pi$, and all $\pi \in \mathcal{C}_\lambda$ share the cycle count $\ell(\lambda) := \sum_k m_k$ for $\lambda = (1^{m_1} 2^{m_2} \cdots)$, we conclude that for the conjugacy class $\mathcal{C}_\lambda$,
\begin{align}
    \tr(C_\lambda) = |\mathcal{C}_\lambda|\, d^{\ell(\lambda)}, \qquad |\mathcal{C}_\lambda| = \frac{N!}{\prod_k k^{m_k}\, m_k!}.\label{eq:class-sum-trace}
\end{align}
For example, $|\mathcal{C}_{(1^N)}| = 1$ and $|\mathcal{C}_\tau| = {N!}/({2^{1} (N-2)!}) = {N(N-1)}/{2} = K$, where $\tau,K$ are defined in Eq.~\eqref{eq:purity-N-target}.
The trace constraint becomes:
\begin{align}
    \tr Z = \sum_{\lambda \vdash N} z_\lambda\, \tr \left[C_\lambda\right] = \sum_{\lambda \vdash N} z_\lambda\, |\mathcal{C}_\lambda|\, d^{\ell(\lambda)} = 1.\label{eq:N-trace-constraint}
\end{align}
For the objective $\tr(Y \cdot F^{(N)}) = K^{-1}\tr(Y \cdot C_\tau)$, we need the product structure in the center. Specifically, we calculate the value for $C_{1^N} = \1_{Nn}$ and $C_\tau$ defined in Eq.~\eqref{eq:purity-N-target} in the following lemma:
\begin{lemma}
    \label{lem:trace-formulas}
    Let $C_{(1^N)} = \1^{\otimes N}$, $K = \binom{N}{2}$, and $C_\tau = \sum_{1 \leq i < j \leq N} W_{i,j}$. Then $\tr(C_{(1^N)}) = d^N, \quad \tr(C_\tau) = K\,d^{N-1}, \quad \tr(C_\tau^2) = K\,Q_N\,d^{N-2}$ for $Q_N = \bigl(d^2 + K - 1\bigr) $.
\end{lemma}
\begin{proof}
    The first two identities follow directly from the identity $\tr(P_\pi) = d^{c(\pi)}$, given that the identity has $N$ fixed points and $\tau = (2, 1^{N-2})$ has a single transposition $(ij)$ cycle and $N-2$ fixed points. For $\tr(C_\tau^2) = \sum_{\pi, \sigma \in \mathcal{C}_\tau}
    d^{c(\pi\sigma)}$, we classify the $K^2$ ordered pairs $(\pi, \sigma)$ by the overlap of their index sets:
        (i). $\pi = \sigma$, which gives $K$ pairs and $W_{i,j}^2 = \1_n^{\otimes N}$, giving $c(\pi\sigma) = N$, contributing $K\,d^N$ in total.
        (ii). $|\mathrm{supp}(\pi) \cap \mathrm{supp}(\sigma)| = 1$, which gives $2K(N-2)$ pairs and $W_{i,j}W_{j,l} = W_{(i\,l\,j)}$ is a 3-cycle with $c = (N-3) + 1 = N - 2$, contributing $2K(N-2)\,d^{N-2}$ in total.
        (iii). $\mathrm{supp}(\pi) \cap \mathrm{supp}(\sigma) = \varnothing$, which gives $K {(N-2)(N-3)}/{2}$ pairs and $c = (N-4) + 2 = N-2$, contributing $d^{N-2}K(N-2)(N-3)/2$ in total.
    Summing and factoring out $K\,d^{N-2}$ gives
    \begin{align}
    \tr(C_\tau^2) = K\,d^N + \left[2K(N{-}2) +
           \tfrac{K(N{-}2)(N{-}3)}{2}\right] d^{N-2} = K\,d^{N-2}\!\left[d^2 +
           \tfrac{(N{-}2)(N{+}1)}{2}\right].\label{eq:C51}
\end{align}
    The combinatorial factor simplifies via
    ${(N-2)(N+1)}/{2} = ({N^2 - N - 2})/{2} = K - 1$,
    yielding the stated formula
    $\tr(C_\tau^2) = K\,d^{N-2}(d^2 + K - 1)$.
\end{proof}
In general, the class sums $C_\lambda$ satisfy the algebra $C_\lambda C_\mu = \sum_{\nu \vdash N} c_{\lambda\mu}^\nu\, C_\nu$ with structure constants computable from the irreducible characters $\chi^\rho$ of $S_N$~\cite{collins2006integration}:
\begin{align}
    c_{\lambda\mu}^\nu = \frac{|\mathcal{C}_\lambda|\,|\mathcal{C}_\mu|}{N!} \sum_{\rho \vdash N} \frac{\chi^\rho_\lambda\, \chi^\rho_\mu\, \chi^\rho_\nu}{\chi^\rho_{(1^N)}},\label{eq:structure-constants}
\end{align}
where each $\rho \vdash N $ labels an irreducible representation of $S_N$, and $\chi_{\lambda}^{\rho}$ indicates the character $\tr\left(\rho(\lambda)\right)$ for $\rho \in \hat{S}_N, \lambda \in S_N$.
Defining the \emph{Weingarten trace matrix}
\begin{align}
    T_{\lambda,\mu}(d) := \tr(C_\lambda\, C_\mu) = \sum_{\nu \vdash N} c_{\lambda\mu}^\nu\, |\mathcal{C}_\nu|\, d^{\ell(\nu)},\label{eq:weingarten-trace-matrix}
\end{align}
the objective reads as
\begin{align}
    \tr(Y \cdot F^{(N)}) =\tr\left[\sum_{\lambda \vdash N} y_\lambda\, C_\lambda \cdot (K^{-1} C_\tau)\right] =  K^{-1} \sum_{\lambda \vdash N} y_\lambda\, T_{\lambda,\tau}(d).\label{eq:N-objective}
\end{align}
For a product state
$\ket{\Psi} = \bigotimes_{i=1}^N \ket{\psi_i}$,
the expectation of a permutation operator factorizes over the disjoint cycles of $\pi$. Writing each cycle as $(i_1\, i_2\, \cdots\, i_k)$ with $\pi(i_1) = i_2,\, \pi(i_2) = i_3,\, \ldots,\, \pi(i_k) = i_1$ and $i_1, \ldots, i_k \in \{1,\ldots,N\}$,
we obtain the product form:
\begin{align}
    \bra{\Psi} P_\pi \ket{\Psi}
    = \prod_{j=1}^N s_{j,\,\pi(j)}
    = \prod_{\substack{\text{cycles~}
      (i_1\, i_2\, \cdots\, i_k)  \text{~of } \pi}}
    G_{i_1 i_2}\, G_{i_2 i_3} \cdots G_{i_{k-1} i_k}\, G_{i_k i_1},\label{eq:perm-product-state}
\end{align}
where $G_{jk} := \braket{\psi_j}{\psi_k}$ defines a matrix that is automatically PSD since $\vec{x}^\dagger G\, \vec{x} = \|\sum_j x_j\ket{\psi_j}\|^2 \geq 0$ for all $\vec{x} \in \mathbb{C}^N$. All possible choices of normalized product states $\ket{\Psi}$ correspond to the set of $N \times N$ Gram matrices that satisfy:
\begin{align}
    \mathcal{G}_{N,d} := \{G \in \mathbb{C}^{N \times N} \mid G \succeq 0,\; G_{ii} = 1,\; \mathrm{rank}(G) \leq d\}.\label{eq:C56}
\end{align}
Since $Z \pm Y$ is related to the center of $\mathbb{C}[S_N]$,
for each conjugacy class $\lambda$, define the \emph{class moment} of the product state:
\begin{align}
    \Phi_\lambda(G) := \bra{\Psi} C_\lambda \ket{\Psi} = \sum_{\pi \in \mathcal{C}_\lambda} \prod_{\substack{\text{cycles~}
    (i_1\, i_2\, \cdots\, i_k) \text{~of } \pi}} G_{i_1 i_2}\, G_{i_2 i_3} \cdots G_{i_{k-1} i_k}\, G_{i_k i_1}.\label{eq:class-moment}
\end{align}
The block-positivity constraints then read:
\begin{align}
    \sum_{\lambda \vdash N} (z_\lambda \pm y_\lambda)\, \Phi_\lambda(G) \geq 0, \qquad \forall\, G \in \mathcal{G}_{N,d},\label{eq:N-block-pos}
\end{align}
Eq.~\eqref{eq:N-block-pos} is a system of polynomial inequalities with respect to every valid $N \times N$ Gram matrix $G$.
Some special cases of $G$ are: (i). $ \ket{\Psi} = \ket{\psi}^{\ox N}$ with $G = \mathbf{1}\mathbf{1}^\top$ with all entries $1$, giving inequality $\Phi_\lambda = |\mathcal{C}_\lambda|$ for all $\lambda$, (ii). $\ket{\Psi} = \bigotimes_{i = 1}^{N}\ket{\psi_i}$ with $\langle\psi_i|\psi_j\rangle = \delta_{ij}$ and $G = I_N$, giving $\Phi_{(1^N)} = 1$ and $\Phi_\lambda = 0$ for all $\lambda \neq (1^N)$.
These two extremes generalize the $t = 1$ and $t = 0$ endpoints of the 2-replica analysis.
Combining the objective~\eqref{eq:N-objective}, the trace constraint~\eqref{eq:N-trace-constraint}, and the block-positivity constraints~\eqref{eq:N-block-pos}, the $N$-replica dual conic convex program in Eq.~\eqref{eq:purity-N-dual} reduces to:
\begin{equation}
     {
    \begin{aligned}
        d^*_N = \max_{\{z_\lambda, y_\lambda\}} \quad & K^{-1} \sum_{\lambda \vdash N} y_\lambda\, T_{\lambda,\tau}(d) \\
        \text{s.t.} \quad & \sum_{\lambda \vdash N} (z_\lambda \pm y_\lambda)\, \Phi_\lambda(G) \geq 0, \quad \forall\, G \in \mathcal{G}_{N,d} \\
        & \sum_{\lambda \vdash N} z_\lambda\, |\mathcal{C}_\lambda|\, d^{\ell(\lambda)} = 1
    \end{aligned}
    }\label{eq:N-scalar-program}
\end{equation}
The program~\eqref{eq:N-scalar-program} has finitely many variables ($2\,p(N)$ real variables, where $p(N)$ is the number of integer partitions) but infinitely many Gram-matrix constraints.
The Weingarten trace matrix $T_{\lambda,\tau}(d)$ is computed via Eq.~\eqref{eq:weingarten-trace-matrix}, and the class moments $\Phi_\lambda(G)$ are polynomials in the overlaps determined by Eq.~\eqref{eq:class-moment}.
For general $N$, the exact program~\eqref{eq:N-scalar-program} involves the full $p(N)$-dimensional center of $\mathbb{C}[S_N]$, including contributions from 3-cycles, 4-cycles, and higher permutation structures. Thus, instead of evaluating the full program~\eqref{eq:N-scalar-program}, we prove a lower bound on the $\mathrm{QB}_N$-norm by considering the transposition sector only.
\begin{lemma}[Lower bound on $\mathrm{QB}_N$-norm]\label{lem:purity-N-unified}
    The optimal value $d_N^*$ of the program~\eqref{eq:N-scalar-program} satisfies:
    \begin{align}
        d_N^* \geq \frac{1}{d}\!\left(\frac{d^2 - 1}{\binom{N}{2}} + 1\right).\label{eq:purity-N-sandwich}
\end{align}
\end{lemma}

\begin{proof}

    We restrict the exact scalar program~\eqref{eq:N-scalar-program} in $2\,p(N)$ real
    variables $\{z_\lambda, y_\lambda\}_{\lambda \vdash N}$ to $\mathrm{span}\{C_{(1^N)}, C_\tau\}$,\ by setting $z_\lambda = y_\lambda = 0$ for $\lambda \notin \{(1^N), \tau\}$, where $\tau = (2, 1^{N-2})$. Using the fact that $|G_{ij}|^2 \leq 1$ for all $G \in \mathcal{G}_{N,d}$, we conclude that
    \begin{align}
        \Phi_{(1^N)}(G) = 1, \quad
        \Phi_\tau(G) = \sum_{i < j} |G_{ij}|^2 \in [0,K], \, \forall G \in \mathcal{G}_{N,d}.\label{eq:C61}
\end{align}
    This reduces the first restriction in Eq.~\eqref{eq:N-scalar-program}.
    Then we use the fact that
    $|\mathcal{C}_{(1^N)}| = 1$ and $|\mathcal{C}_\tau| = K$ to simplify the second restriction in Eq.~\eqref{eq:N-scalar-program}.
    Finally,
    we note that the target function is given by the nonzero coefficients:
    \begin{align}
       K^{-1}[y_1\,T_{(1^N),\tau}(d) + y_\tau\,T_{\tau,\tau}(d)] = K^{-1}[y_1\tr(C_{(1^N)} \cdot C_{\tau}) + y_\tau\tr(C_\tau \cdot C_\tau)],\label{eq:C62}
\end{align}
    We refer to the trace formulas in Lemma~\ref{lem:trace-formulas} as follows to simplify the target function:
    \begin{align}
        \tr(C_{(1^N)}) = d^N, \qquad \tr(C_\tau) = K\,d^{N-1}, \qquad
        \tr(C_\tau^2) = K\,d^{N-2}\!\left(d^2 +
        \tfrac{(N-2)(N+1)}{2}\right) = K\,Q_N\,d^{N-2}.\label{eq:C63}
\end{align}
    So the objective evaluates to
    \begin{align}
        K^{-1}\bigl[y_1\,K\,d^{N-1} + y_\tau\,K\,Q_N\,d^{N-2}\bigr] = y_1\,d^{N-1} + y_\tau\,Q_N\,d^{N-2}.\label{eq:C64}
\end{align}
    This reduces to the following scalar program:
    \begin{equation}
         {
        \begin{aligned}
            \tilde{d}^*_N &= \max_{z_1,\,z_\tau,\,y_1,\,y_\tau}
            \quad y_1\,d^{N-1} + y_\tau\,Q_N\,d^{N-2},\\
            \text{s.t.} \quad
            & z_1 + y_1 \geq 0 \;\;\text{(C1)}, \quad
              z_1 + y_1 + K(z_\tau + y_\tau) \geq 0 \;\;\text{(C2)}, \\
            & z_1 - y_1 \geq 0 \;\;\text{(C3)}, \quad
              z_1 - y_1 + K(z_\tau - y_\tau) \geq 0 \;\;\text{(C4)}, \\
            & z_1\,d^N + z_\tau\,K\,d^{N-1} = 1
              \;\;\text{(C5)}.
        \end{aligned}
        }\label{eq:N-scalar-program-restricted}
\end{equation}
    Since the dual conic convex program is a maximization, any feasible solution of the reduced program~\eqref{eq:N-scalar-program-restricted} yields a rigorous lower bound of $\tilde{d}^*_N \leq d^*_N$. In analogy with the optimization problem~\eqref{eq:scalar_dual_SDP_purity}, we exhibit the following explicit feasible solution:
    \begin{align}
        z_1 = 0, \quad
        z_\tau = \frac{1}{K\,d^{N-1}}, \quad
        y_1 = 0, \quad
        y_\tau = \frac{1}{K\,d^{N-1}}.\label{eq:N-feasible-point}
\end{align}
    It can be verified that all constraints are satisfied.
    The objective at this feasible point evaluates to:
    \begin{align}
        \tilde{d}_N^*
        \geq y_1\,d^{N-1} + y_\tau\,Q_N\,d^{N-2}
        = \frac{Q_N\,d^{N-2}}{K\,d^{N-1}}
        = \frac{Q_N}{Kd}
        = \frac{1}{d}\!\left(
        \frac{d^2 - 1}{\binom{N}{2}} + 1\right).\label{eq:C67}
\end{align}
    Since the dual conic convex program is a maximization and this point
    is feasible, we conclude:
    \begin{align}
        d_N^* \geq \tilde{d}_N^* \geq \frac{Q_N}{Kd}
        = \frac{1}{d}\!\left(
        \frac{d^2 - 1}{\binom{N}{2}} + 1\right) = \frac{2(d-1)(d+1)}{d N (N-1)} + \frac{1}{d} .\label{eq:purity-N-lower-final}
\end{align}

\end{proof}

The strong duality of the dual conic convex program then gives $\|\{\tr(\rho^2)\}\|_{N,{\mathrm{mem}}} = \Omega({d}/{N^2})$, which in the main-text notation is the resource-restricted norm $\|\cF_{\mathrm{pur}}\|_{k,\mathrm{mem}}$ with replica number $k=N$.
This should be interpreted as a bounded-unbiased-estimator obstruction
for an $N$-sample U-statistic estimator, where the $N$ single-copy
outcomes jointly generate $\binom{N}{2}$ virtual pairs. Requiring a
constant-bounded estimator, $s_N^*=O(1)$, gives
\begin{align}
    N=\Omega(\sqrt d)=\Omega(2^{n/2}).\label{eq:C69}
\end{align}
Thus, the $\mathrm{QB}_N$-norm lower bound recovers the memory-free purity-testing
scaling $\Omega(2^{n/2})$~\cite{chen2024optimal,gong2024sample}.

\appsubsection{Problems with tensor product structure}

Direct solution of the $k$-replica $\mathrm{QB}_k$-norm in Eq.~\eqref{eq:dual-robustness} (in the case $b = 0$) involves exponentially many constraints for each sign vector $c \in \{-1,1\}^{|\cA|}$ and matrix variables with $2^n$ dimensions, rendering it intractable for both large $|\cA|$ and large qubit number $n$.
A natural strategy is to decompose a large $\mathrm{QB}_k$-norm problem into smaller ones via a tensor product extension.
Concretely, consider two function families $\cF,\cG$ with $k$-replica linearizations $\mathbf F:=\{F_\alpha\}_{\alpha\in\cA}\subset \mathrm{Herm}(\cH_1^{\ox k})$
and
$\mathbf G:=\{G_\beta\}_{\beta\in\cB}\subset \mathrm{Herm}(\cH_2^{\ox k})$, respectively.
We define their tensor product extension as
\begin{align}
    \mathbf F\boxtimes \mathbf G
    :=
    \{F_\alpha\ox G_\beta\}_{(\alpha,\beta)\in\cA\times\cB}
    \subset \mathrm{Herm}((\cH_1\ox\cH_2)^{\ox k}).\label{eq:E21}
\end{align}
These operators define a new function family $\mathcal{T} := \{t_{\alpha,\beta}(\sigma) = \tr\left((F_\alpha\ox G_\beta) \sigma^{\ox k}\right)\}_{(\alpha,\beta) \in \cA \times \cB}$ for states $\sigma$ on $\cH_1 \otimes \cH_2$, which takes $\{T_{\alpha,\beta}:=F_\alpha\otimes G_\beta\}$ as one of its $k$-replica linearizations.
Conversely, one can decompose the $\mathrm{QB}_k$-norm problem for $\mathcal{T}$ into smaller problems for $\cF$ and $\cG$, each involving fewer qubits and fewer functions.
We denote by $s_1^*(\mathbf F)$ and $s_1^*(\mathbf G)$ the optimal scaling factors for the  problem in Lemma~\ref{lem:dual-robustness} (in the case $b = 0$), and by
$s_{1,\mathrm{loc}}^*(\mathbf F\boxtimes\mathbf G)$ the solution of Lemma~\ref{lem:dual-robustness} restricted to the separable POVM on $\cH_1^{\ox k}$ and $\cH_2^{\ox k}$ (a special case of Eq.\eqref{eq:site-product-cone-loc}), which gives the dual cone constraints in Eq.~\eqref{eq:site-product-cone}.
It is therefore valuable to establish the relation between $s_1^*(\mathbf F\boxtimes \mathbf G)$ and the individual factors $s_{1,\mathrm{loc}}^*(\mathbf F)$, $s_1^*(\mathbf G)$, which in turn yields information about $\|\mathcal{T}\|_{k,\mathrm{loc}}$ from the smaller problems $\|\cF\|_k$ and $\|\cG\|_k$.

\begin{lemma}[Binary tensor product factorization]\label{lem:binary-tensor-factor}
For any two operator families
$\mathbf F=\{F_\alpha\}_{\alpha\in\cA}$ and
$\mathbf G=\{G_\beta\}_{\beta\in\cB}$,
one has $s_{1,\mathrm{loc}}^*(\mathbf F\boxtimes\mathbf G)=s_1^*(\mathbf F)\,s_1^*(\mathbf G)$.
\end{lemma}

\begin{proof}
To prove $s_{1,\mathrm{loc}}^*(\mathbf F\boxtimes\mathbf G)=s_1^*(\mathbf F)\,s_1^*(\mathbf G)$, we need to show the upper bound $s_{1,\mathrm{loc}}^*(\mathbf F\boxtimes\mathbf G) \leq s_1^*(\mathbf F)\,s_1^*(\mathbf G)$ and the lower bound $s_{1,\mathrm{loc}}^*(\mathbf F\boxtimes\mathbf G) \geq s_1^*(\mathbf F)\,s_1^*(\mathbf G)$. For the first inequality, we let
    $\{\Pi_i,h_\alpha(i)\}$ and $\{\Gamma_j,g_\beta(j)\}$
    be optimal strategies for $\mathbf F$ and $\mathbf G$, respectively. Thus
    \begin{align}
        \sum_i h_\alpha(i)\Pi_i &= F_\alpha,
        \qquad
        \sum_i \Pi_i=\1_{\cH_1^{\ox k}},
        \qquad
        |h_\alpha(i)|\le s_1^*(\mathbf F),\label{eq:E22}
\\
        \sum_j g_\beta(j)\Gamma_j &= G_\beta,
        \qquad
        \sum_j \Gamma_j=\1_{\cH_2^{\ox k}},
        \qquad
        |g_\beta(j)|\le s_1^*(\mathbf G).\label{eq:E23}
\end{align}
    Consider the product POVM $M_{i,j}:=\Pi_i\ox\Gamma_j$
    and the product estimator $q_{\alpha,\beta}(i,j):=h_\alpha(i)g_\beta(j)$ such that $|q_{\alpha,\beta}(i,j)|\le s_1^*(\mathbf F)s_1^*(\mathbf G)$.
    Then
    \begin{align}
        \sum_{i,j}q_{\alpha,\beta}(i,j)M_{i,j}
        =
        \left(\sum_i h_\alpha(i)\Pi_i\right)
        \ox
        \left(\sum_j g_\beta(j)\Gamma_j\right)
        =
        F_\alpha\ox G_\beta,\label{eq:E24}
\end{align}
    and therefore $s_{1,\mathrm{loc}}^*(\mathbf F\boxtimes\mathbf G)\le s_1^*(\mathbf F)s_1^*(\mathbf G)$.
    For the lower bound, we use the dual SDP. For the two unrestricted local problems, the dual forms are
    \begin{align}
        s_1^*(\mathbf F)= \max_{Z_F,\{Y_\alpha\}} \sum_{\alpha\in\cA}\tr(Y_\alpha F_\alpha) \quad \text{subject to} \quad Z_F-\sum_{\alpha\in\cA}a_\alpha Y_\alpha\succeq 0, \quad \forall a\in\{-1,1\}^{|\cA|}, \quad \tr Z_F = 1,\label{eq:dual-sdp-F}
\\
        s_1^*(\mathbf G)= \max_{Z_G,\{W_\beta\}} \sum_{\beta\in\cB}\tr(W_\beta G_\beta) \quad \text{subject to} \quad Z_G-\sum_{\beta\in\cB}b_\beta W_\beta\succeq 0, \quad \forall b\in\{-1,1\}^{|\cB|}, \quad \tr Z_G = 1.\label{eq:dual-sdp-G}
\end{align}
    Take optimal dual solutions
    $(Z_F,\{Y_\alpha\})$ and $(Z_G,\{W_\beta\})$.
    Define $Z:=Z_F\ox Z_G,\,Y_{\alpha,\beta}:=Y_\alpha\ox W_\beta$.
    We show that this is feasible for the local dual problem of
    $\mathbf F\boxtimes\mathbf G$.
    First, the trace constraint is satisfied: $\tr Z=\tr Z_F\,\tr Z_G = 1$.
    It remains to verify the local dual cone constraint. Since the local measurement cone is generated by product PSD effects, its dual cone is the block-positive cone, as described in Eq.\eqref{eq:site-product-cone}:
    \begin{align}
        \mathcal K_{\mathrm{loc}}^*
        =
        \left\{
        X:
        \bra{u}\bra{v}X\ket{u}\ket{v}\ge 0,
        \ \forall\,\ket{u}\in\cH_1^{\ox k},
        \ \forall\,\ket{v}\in\cH_2^{\ox k}
        \right\}.\label{eq:E27}
\end{align}
    Fix arbitrary $\ket{u}, \ket{v}$ and define
    $z_F:=\bra{u}Z_F\ket{u}, \, y_\alpha:=\bra{u}Y_\alpha\ket{u},$
    and $z_G:=\bra{v}Z_G\ket{v},\, w_\beta:=\bra{v}W_\beta\ket{v}.$
    Plugging in Eq.~\eqref{eq:dual-sdp-F} and Eq.~\eqref{eq:dual-sdp-G} gives:
    \begin{align}
        z_F\ge \sum_{\alpha\in\cA}|y_\alpha|,\quad  z_G\ge \sum_{\beta\in\cB}|w_\beta|.\label{eq:E28}
\end{align}
    Hence for any global sign matrix
    $c=\{c_{\alpha,\beta}\}_{(\alpha,\beta)\in\cA\times\cB}$, we conclude:
    \begin{align}
    \bra{u}\bra{v} \left(Z-\sum_{\alpha,\beta}c_{\alpha,\beta}Y_{\alpha,\beta}\right)\ket{u}\ket{v}
    &= z_Fz_G-\sum_{\alpha,\beta}c_{\alpha,\beta}y_\alpha w_\beta
    \ge z_Fz_G-\sum_{\alpha,\beta}|y_\alpha||w_\beta| \notag
\\
    &=z_Fz_G-
    \left(\sum_{\alpha}|y_\alpha|\right)
    \left(\sum_{\beta}|w_\beta|\right)\ge 0.\label{eq:E29}
\end{align}
    Therefore
    $Z-\sum_{\alpha,\beta}c_{\alpha,\beta}Y_{\alpha,\beta}\in \mathcal K_{\mathrm{loc}}^*, \,\forall c\in\{-1,1\}^{|\cA||\cB|}$.
    Thus the constructed dual variables are feasible.
    Finally, their objective value is
    \begin{align}
        \sum_{\alpha,\beta}
        \tr\!\left[
        (Y_\alpha\ox W_\beta)(F_\alpha\ox G_\beta)
        \right]=
        \sum_{\alpha,\beta}
        \tr(Y_\alpha F_\alpha)\tr(W_\beta G_\beta)=
        \left(
        \sum_{\alpha}\tr(Y_\alpha F_\alpha)
        \right)
        \left(
        \sum_{\beta}\tr(W_\beta G_\beta)
        \right)=
        s_1^*(\mathbf F)s_1^*(\mathbf G).\label{eq:E30}
\end{align}
    The dual is a maximization problem, so the inequality gives
    $s_{1,\mathrm{loc}}^*(\mathbf F\boxtimes\mathbf G)\ge s_1^*(\mathbf F)s_1^*(\mathbf G)$.
    Combining the two inequalities proves the claim.
\end{proof}

Note that the linearizations in Lemma~\ref{lem:binary-tensor-factor} need not be the optimal symmetric ones, so by Theorem~\ref{thm:structure_fcomp} we only have the upper bounds
\begin{align}
    \|\cF\|_{k}:= s_k^*(\cF) \leq s_1^*(\mathbf{F}),\quad \|\cG\|_{k} :=s_k^*(\cG) \leq s_1^*(\mathbf G), \quad \|\mathcal{T}\|_{k,\mathrm{loc}}:= s_{k,\mathrm{loc}}^*(\mathcal{T}) \leq s_{1,\mathrm{loc}}^*(\mathbf F\boxtimes\mathbf G).\label{eq:E31}
\end{align}
If both $\mathbf F$ and $\mathbf G$ are not the canonical symmetric linearizations, Lemma~\ref{lem:binary-tensor-factor} remains valid but gives a loose bound between $\|\cF\|_{k},\|\cG\|_{k}$ and $\|\mathcal{T}\|_{k,\mathrm{loc}}$.
In the case where both $\mathbf F$ and $\mathbf G$ are symmetric, an upper bound $\|\mathcal{T}\|_{k,\mathrm{loc}}\leq \|\cF\|_{k} \cdot \|\cG\|_{k}$ is satisfied.

\paragraph{Example.}
As an illustration of the tensor product extension, consider the $n$-qubit Pauli operators $\mathcal{P}_n = \{\1,X,Y,Z\}^{\ox n}$.
So by Lemma~\ref{lem:binary-tensor-factor} we have:
\begin{align}
    \|\mathcal{P}_n\|_{1,\mathrm{loc}} = \|\{X,Y,Z,\1\}\|_{1}^n = (\sqrt{3})^n,\label{eq:E32}
\end{align}
whose variance scale $\|\mathcal{P}_n\|_{1,\mathrm{loc}}^2=3^n$
matches the variance lower bound and local classical-shadow scaling $3^{n}$ for full-weight Pauli observables for the full Pauli set~\cite{huang2020predicting}.
In the main-text notation, this proves the resource-restricted norm $\|\cF_{\mathrm P}\|_{1,\mathrm{loc}}=3^{n/2}$ of the Pauli family $\cF_{\mathrm P}=\{\tr(P\rho)\}_{P\in\cP_n}$.

\end{document}